\documentclass[12pt,letter]{article}
\usepackage{natbib}
\setcitestyle{authoryear,open={(},close={)},citesep={,}}
\usepackage[]{hyperref}
\usepackage{xcolor}
\hypersetup{colorlinks = true, citecolor = blue, linkcolor=blue, linkbordercolor = {white}, urlcolor=blue}
\usepackage{amsmath}
\usepackage{mathrsfs}
\usepackage{amsthm}
\usepackage{amssymb}
\usepackage{mathtools}
\usepackage{bm}
\usepackage{graphicx}
\usepackage{caption}
\usepackage{subcaption}
\usepackage{multirow}
\usepackage{dsfont}
\usepackage{cleveref}
\usepackage[normalem]{ulem}

\usepackage
[
        a4paper,
        left=1.75cm,
        right=1.75cm,
        top=2.5cm,
        bottom=2.5cm,
]
{geometry}

\newtheorem{theorem}{Theorem}[section]
\newtheorem{lemma}{Lemma}[section]

\newtheorem{example}{Example}[section]
\newtheorem{remark}{Remark}[section]
\newtheorem{condition}{Condition}[section]

\newenvironment{customcon}[1]
{\renewcommand\thecondition{#1}\condition}
{\endcondition}

\newtheorem{innercustomgeneric}{\customgenericname}
\providecommand{\customgenericname}{}
\newcommand{\newcustomtheorem}[2]{%
  \newenvironment{#1}[1]
  {%
   \renewcommand\customgenericname{#2}%
   \renewcommand\theinnercustomgeneric{##1}%
   \innercustomgeneric
  }
  {\endinnercustomgeneric}
}
\newcustomtheorem{customlemma}{Lemma}
\newcustomtheorem{customproposition}{Proposition}
\newcustomtheorem{customtheorem}{Theorem}
\newcustomtheorem{customdefinition}{Definition}
\newcustomtheorem{customexample}{Example}
\newcustomtheorem{customㆍㆍㆍremark}{Remark}

\newcommand{\me}{\mathbf{e}}

\newcommand{\mh}{\mathbf{h}}

\newcommand{\mt}{\mathbf{t}}
\newcommand{\mv}{\mathbf{v}}

\newcommand{\mx}{\mathbf{x}}

\newcommand{\mz}{\mathbf{z}}
\newcommand{\bu}{\mathbf{u}}

\newcommand{\mD}{\mathbf{D}}

\newcommand{\mH}{\mathbf{H}}

\newcommand{\mX}{\mathbf{X}}

\def\baa#1\eaa{\begin{align}#1\end{align}}
\def\ba#1\ea{\begin{align*}#1\end{align*}}
\def\bs#1\es{\begin{split}#1\end{split}}

\DeclareMathOperator*{\argmin}{\arg\min}

\newcommand{\tran}{^{\mathstrut\scriptscriptstyle{\top}}}

\newcommand{\E}{{\rm E}}

\newcommand{\Leb}{{\rm Leb}}

\begin{document}

\title{Local Fr\'echet regression with toroidal predictors}

\author{Chang Jun Im$^1$ and Jeong Min Jeon$^2$}

\date{$^1$The Institute for Data Innovation in Science, \\Seoul National University, South Korea\\
$^2$Department of Statistics and School of Transdisciplinary Innovations, \\Seoul National University, South Korea\\}

\maketitle

\begin{abstract}
We provide the first regression framework that simultaneously accommodates responses taking values in a general metric space and predictors lying on a general torus. We propose intrinsic local constant and local linear estimators that respect the underlying geometries of both the response and predictor spaces. Our local linear estimator is novel even in the case of scalar responses. We further establish their asymptotic properties, including consistency and convergence rates. Simulation studies, together with an application to real data, illustrate the superior performance of the proposed methodology.
\end{abstract}

\small
\begin{quotation}
\end{quotation}
\normalsize

\section{Introduction}

Modern data often reside in spaces that lack a vector space structure or possess geometric features, such as dimensions, distances and angles, that are fundamentally different from those of Euclidean spaces. These characteristics pose substantial challenges for statistical analysis and motivate the development of intrinsic approaches that respect the underlying data geometry. For an introduction to non-Euclidean data analysis, we refer to \cite{Sanborn et al. (2024)}.

We study nonparametric regression for toroidal predictors taking values in the $d$-dimensional torus $\mathbb{T}^d := \mathbb{S}^1 \times \cdots \times \mathbb{S}^1$, which is the Cartesian product of $d\geq1$ unit circles $\mathbb{S}^1:=\{(\cos\theta,\sin\theta)\tran\in\mathbb{R}^2:\theta\in[-\pi,\pi)\}$. Data exhibiting natural periodicity can be viewed as lying on $\mathbb{S}^1$, and data lying on $\mathbb{T}^d$ naturally arise in many scientific disciplines when observations exhibit multiple periodic components. For instance, in Earth science, analyzing multiple wind and wave directions is an important task, and such joint directional measurements give rise to data lying on a torus (e.g., \cite{Xu and Wang (2023)}; \cite{Biswas and Banerjee (2025)}). In molecular biology, protein backbone conformations are commonly characterized by paired dihedral angles that lie on a torus (e.g., \cite{Di Marzio et al. (2011)}; \cite{Eltzner et al. (2018)}; \cite{Jung et al. (2021)}; \cite{Xu and Wang (2023)}). In many real-world datasets, events of interest depend simultaneously on the time of day and the day of the year, which together define observations on $\mathbb{T}^2$. Despite the abundance of such data, statistical methodologies for toroidal predictors remain limited. For example, \cite{Di Marzio et al. (2009)} studied local linear regression for
toroidal predictors with scalar responses, and \cite{Biswas and Banerjee (2025)} investigated semi-parametric regression for predictors and responses both taking values in \(\mathbb{T}^2\). More generally, regression methods that accommodate manifold-valued predictors, including toroidal predictors, have been developed in \cite{Pelletier (2006)}, \cite{Jeon et al. (2021)} and \cite{Jeon et al. (2022)}, among others. However, the latter works focus on scalar or Hilbert-space-valued responses and thus exclude manifold-valued responses or more generally, metric-space-valued responses.

The Fr\'echet regression framework, introduced by \cite{Petersen and Muller (2019)}, provides a powerful tool for modeling responses that take values in a general metric space, including Wasserstein spaces, network spaces and manifolds. Subsequent studies, such as \cite{Zhou and Muller (2022)}, \cite{Bhattacharjee and Muller (2023)}, \cite{Tucker et al. (2023)} and \cite{Steyer et al. (2025)}, have extended this framework in various directions. However, most existing Fr\'echet regression methods focus on Euclidean predictors. Recently, \cite{Tucker and Wu (2025)} considered predictors taking values in a general metric space, but their approach is based on local constant estimation and relies on a H\"older continuity assumption on the regression function, which yields suboptimal convergence rates when the true regression function is sufficiently smooth. More recently, \cite{Im et al. (2025)} developed local constant and local linear Fr\'echet regression for spherical predictors taking values in $\mathbb{S}^p := \{\mv \in \mathbb{R}^{p+1} : \|\mv\|_2 = 1\}$.

In this work, we develop both local constant and local linear regression methods for toroidal predictors with metric-space-valued responses, fully respecting the underlying data geometry. To the best of our knowledge, local linear estimation in this setting has not been studied. Our local linear estimator is novel even in the case of scalar responses, as it differs from that of \cite{Di Marzio et al. (2009)}. Both the proposed local constant and local linear estimators achieve the optimal error rate. We demonstrate the practical utility of the proposed methods through an analysis of New York taxi network data, where the response is a graph Laplacian that is neither vector-space-valued nor manifold-valued.

The rest of this paper is organized as follows. Section~\ref{section:2} introduces the problem setting and the proposed estimation methods. Section~\ref{section:3} presents the asymptotic theory, including consistency and convergence rates. Sections~\ref{section:4} and~\ref{section:5} evaluate the finite-sample performance of the proposed methods through simulation studies and a real data application, respectively. All technical proofs are provided in the Supplementary Material.

\section{Problem setting and estimators} \label{section:2}
\setcounter{equation}{0}
\subsection{Problem setting}

Let $(\Omega,\mathcal{F},\mathbb{P})$ be an underyling probability space and $\mathbb{M}$ be a general totally bounded metric space equipped with a metric $d_{\mathbb{M}}$. Also, let $Y:\Omega\rightarrow\mathbb{M}$ be the response variable and $\mathbf{X}:\Omega\rightarrow\mathbb{T}^{d}$ be the predictor. For any $\mx \in \mathbb{T}^{d}$, we write $\mx = \big( \mx_1 \tran, \dots, \mx_d \tran \big) \tran$, where $\mx_\ell \in \mathbb{S}^1$.
We define a function $M_{\oplus} : \mathbb{T}^{d} \times \mathbb{M} \rightarrow \mathbb{R}$ as
\begin{align*}
    M_{\oplus}(\mathbf{x}, y) := \mathbb{E} \big[ d_{\mathbb{M}}^2(Y, y) \big| \mathbf{X} = \mathbf{x} \big].
\end{align*}
In this paper, we aim to estimate a function $m_{\oplus}:\mathbb{T}^{d}\rightarrow\mathbb{M}$, defined as
\begin{align*}
    m_{\oplus}(\mathbf{x}):=\mathop{\mathrm{argmin}}_{y \in \mathbb{M}} M_{\oplus}(\mathbf{x}, y). \label{eq:3.3}
\end{align*}
When $Y\in\mathbb{R}$, the usual regression function is characterized by $\argmin_{y\in\mathbb{R}}\E[(Y-y)^2|\mX=\mx]$. Hence, the function $m_{\oplus}$ can be understood as a natural extension of the usual regression function and is often called the Fr\'echet regression function. We propose two estimators for $m_{\oplus}$: a local constant estimator and a local linear estimator. Before we introduce the estimators, we present some examples of $\mathbb{M}$.

\begin{example}\label{example1}\leavevmode
For $[a,b]\subset\mathbb{R}$ with $-\infty<a<b<\infty$, let $\mathcal{W}([a,b])$ denote the space of all distribution functions $F:\mathbb{R}\rightarrow[0,1]$ such that $F^{-1}(t)=\inf\{x\in\mathbb{R}:F(x)\geq t\}\in [a,b]$ for all $t\in(0,1]$. Also, let $d_W$ denote the 2-Wasserstein distance, defined as
    \begin{align*}
    d_W(F,G):=\bigg(\int_0^1(F^{-1}(t)-G^{-1}(t))^2 \mathrm{d}t\bigg)^{1/2}, \quad F, G\in \mathcal{W}([a,b]).
    \end{align*}
    Then, $\mathcal{W}([a,b])$ equipped with $d_W$, called the 2-Wasserstein space on $[a,b]$ is totally bounded; see Corollary 2.2.5 and Proposition 2.2.8 of \cite{Panaretos and Zemel (2020)}.
\end{example}

\begin{example}\label{example4} \leavevmode
Let $G = (V, E)$ be a graph network with nodes
$V = \{v_1, \dots, v_k\}$ and edge weights
$E = \{w_{ij} : 1 \le i, j \le k\}$ satisfying
$0 \le w_{ij} \le C_w$ for some constant $C_w > 0$,
where $w_{ij} = 0$ indicates that nodes $v_i$ and $v_j$ are not connected.
We assume that $G$ has no self-loops or multi-edges and is undirected, that is,
$w_{ij} = w_{ji}$ for all $1 \le i, j \le k$.
Such a graph network can be uniquely represented by its graph Laplacian
$\mathbf{L} = (\mathbf{L}_{ij}) \in \mathbb{R}^{k \times k}$, whose entries are defined as
\[
\mathbf{L}_{ij} :=
\begin{cases}
-w_{ij}, & i \neq j, \\
\sum_{k \neq i} w_{ik}, & i = j.
\end{cases}
\]
We consider the space of graph Laplacians
\[
\mathcal{L}_k := \left\{ \mathbf{L} : \mathbf{L} = \mathbf{L}\tran, \, \mathbf{L} \mathbf{1}_k = \mathbf{0}_k, \, -C_w \leq \mathbf{L}_{ij} \leq 0 \right\},
\]
where $\mathbf{1}_k$ and $\mathbf{0}_k$ denote the $k$-dimensional vectors of ones and zeros,
respectively.
We equip $\mathcal{L}_k$ with the Frobenius distance
\[
d_F(\mathbf{L}, \mathbf{M})
:= \left( \sum_{i=1}^k \sum_{j=1}^k
(\mathbf{L}_{ij} - \mathbf{M}_{ij})^2 \right)^{1/2}.
\]
Then, the metric space $(\mathcal{L}_k, d_F)$ is totally bounded;
see Proposition~1 of \cite{Zhou and Muller (2022)}.
\end{example}

\begin{example}\label{example2}\leavevmode
Any finite-dimensional compact Riemannian manifold without boundary, equipped with the geodesic distance, is totally bounded. Such Riemannian manifolds include the spheres, toruses, planar shape spaces (\cite{Le and Kendall (1993)}) and special orthogonal groups ${\rm SO}(k)=\{A\in\mathbb{R}^{k\times k}:A\text{~is orthogonal and~}\det(A)=1\}$.
\end{example}

\subsection{Local constant estimator}

We introduce our local constant estimation. For a vector $\mh =  ( h_1, \dots, h_d ) \tran$ of bandwidths $h_\ell>0$, we define the toroidal product kernel $\mathcal{L}_{\mx, \mh} : \mathbb{T}^{d} \to [0, \infty)$ as
\begin{align}\label{Lx def}
    \mathcal{L}_{\mx, \mh} \left( \mz\right) := \prod_{\ell=1}^{d} L \left( \frac{1 - \mz_\ell \tran \mx_\ell}{h_\ell^2} \right), \quad  \mz = \big( \mz_1 \tran, \dots, \mz_d \tran \big)\tran \in \mathbb{T}^{d},
\end{align}
where $L : [0,\infty) \rightarrow [0,\infty)$ is a decreasing function. Note that the angle $\vartheta_\ell\in[0,\pi]$ between $\mz_\ell$ and $\mx_\ell$ equals $\arccos(\mz_\ell\tran\mx_\ell)$, which implies that $1-\mz_\ell\tran\mx_\ell=2\sin^2(\vartheta_\ell/2)$.
Hence, $1-\mz_\ell\tran\mx_\ell$ is larger as $\mz_\ell$ is further away from $\mx_\ell$. Let $\{ (\mathbf{X}^{(i)}, Y^{(i)}) \}_{i=1}^n$ be a random sample of $(\mX,Y)$. We estimate $M_{\oplus}(\mx,y)$ by the local constant estimator
\begin{gather*}
    \hat{M}_{\mh, 0}(\mx, y) := \dfrac{n^{-1} \sum_{i=1}^{n} \mathcal{L}_{\mx, \mh} \left( \mathbf{X}^{(i)} \right) d_{\mathbb{M}}^2(Y^{(i)}, y)}{n^{-1} \sum_{i=1}^{n} \mathcal{L}_{\mx, \mh} \left( \mathbf{X}^{(i)} \right)}
\end{gather*}
and estimate $m_{\oplus}(\mx)$ by
\begin{align*}
    \hat{m}_{\mh, 0}(\mathbf{x}) := \mathop{\mathrm{argmin}}_{y \in \mathbb{M}} \hat{M}_{\mh, 0}(\mathbf{x},y). 
\end{align*}
We call $\hat{m}_{\mh, 0}(\mx)$ a local constant Fr\'echet regression estimator at $\mx$.

\subsection{Local linear estimator}\label{local linear section}

We introduce our local linear estimation. We first consider the case where $Y$ is scalar and then extend the construction to the general metric space. Our estimation method is new even in this setting. Let $g:\mathbb{T}^d \to \mathbb{R}$ be a sufficiently smooth function. We begin by introducing a notion of the gradient of $g$, which will be used in the construction of the local linear estimator. We define
\begin{align*}
    \mathcal{N}_\ell := \big\{ \big( \mz_1 \tran, \dots, \mz_d \tran \big)\tran \in \mathbb{R}^{2d} : \mz_\ell \tran = \bm{0}_2 \big\}, \quad \mathcal{N} := \bigcup_{\ell=1}^d \mathcal{N}_\ell,
\end{align*}
where $\bm{0}_p$ denotes the zero vector in $\mathbb{R}^p$ for $p\in\mathbb{N}$. We also define the homogeneous extension $\bar g:\mathbb{R}^{2d}\setminus\mathcal{N}\to\mathbb{R}$ of $g$ as
\begin{align}\label{homo ext}
     \bar{g} (\mz) := g \left( \left( \frac{\mz_1 \tran}{\| \mz_1 \|_2} , \dots, \frac{\mz_d \tran}{\| \mz_d \|_2} \right) \tran \right), \quad \mz = \big( \mz_1 \tran, \dots, \mz_d \tran \big)\tran \in \mathbb{R}^{2d} \setminus\mathcal{N}.
\end{align}
We define the block-diagonal matrix $\mathbf{D}_{\mz}\in\mathbb{R}^{2d\times d}$ as
\begin{align*}
    \mathbf{D}_{\mz} := \begin{bmatrix}
        \mz_1 & \cdots & \bm{0}_{2} \\
        \vdots & \ddots & \vdots \\
        \bm{0}_{2} & \cdots & \mz_d
    \end{bmatrix}.
\end{align*}
Note that $\bar{g}$ is homogeneous of degree $0$, that is, 
\[
\bar g(\mathbf{D}_{\mz}\mt) = \bar g
(t_1 \mz_1\tran, \dots, t_d \mz_d\tran) = \bar g(\mz)
\]
for all $\mt=(t_1,\ldots,t_d)\tran\in(0,\infty)^d$. Since the function $\mt \mapsto \bar g(\mathbf{D}_{\mz}\mt)$
is constant with respect to $\mt$, the chain rule implies that
\begin{align}\label{homo eq}
0 = \frac{\partial}{\partial t_\ell} \bar g(\mathbf{D}_{\mz}\mt) = \nabla \bar g(\mathbf{D}_{\mz}\mt)\tran \mathbf{D}_{\mz} \me_\ell, \quad \ell = 1,\dots,d,
\end{align}
where $\nabla \bar g(\mathbf{D}_{\mz}\mt)\in\mathbb{R}^{2d}$ denotes the gradient of $\bar{g}$ at $\mathbf{D}_{\mz}\mt$, and $\me_\ell=(0,\ldots,0,1,0,\ldots,0)\tran$ denotes the $\ell$-th standard basis vector in $\mathbb{R}^d$. Evaluating \eqref{homo eq} at $\mt=(1,\ldots,1)\tran\in(0,\infty)^d$ yields
\begin{align}\label{constraint}
\mathbf{D}_{\mz} \tran\nabla \bar{g}(\mz)=\bm{0}_d.
\end{align}
This imposes a natural constraint on the gradient.

Now, we introduce a first-order approximation of $g$. For $\mx \in \mathbb{T}^{d}$, we define a function $\bm{\Phi}_{\mx,\ell}:[-\pi,\pi) \to \mathbb{S}^1$ by
\[
\bm{\Phi}_{\mx,\ell}(\theta_\ell)
:= \mx_\ell \cos \theta_\ell + (\mathbf{R}\mx_\ell)\sin \theta_\ell, \quad \theta_\ell\in[-\pi,\pi),
\]
where
\[
   \mathbf{R} =
   \begin{bmatrix}
   0 & -1 \\ 1 & 0
   \end{bmatrix}
\]
denotes the counterclockwise rotation matrix by $\pi/2$ in $\mathbb{R}^2$. This function $\bm{\Phi}_{\mx}$ is a bijection due to the tangent-normal decomposition (e.g., Equation (9.1.2) of \cite{Mardia and Jupp (2000)}). We also define $\bm{\Phi}_{\mx} : [-\pi,\pi)^d \to \mathbb{T}^d$ as
\begin{gather}
    \bm{\Phi}_{\mx} ( \bm{\theta} ) := \big( \bm{\Phi}_{\mx, 1} \left( \theta_1 \right)\tran , \dots , \bm{\Phi}_{\mx, d} \left( \theta_d \right) \tran \big) \tran,\quad \bm{\theta} = \left( \theta_1, \dots, \theta_d \right) \tran \in [-\pi, \pi)^{d}, \label{tangent-normal}
\end{gather}
which is also a bijection. We define the block-diagonal matrix $\mathbf{R}_{\mx}\in\mathbb{R}^{2d\times d}$ as 
\begin{align*}
    \mathbf{R}_{\mx} := \begin{bmatrix}
        \mathbf{R} \mx_1 & \cdots & \bm{0}_{2} \\
        \vdots & \ddots & \vdots \\
        \bm{0}_{2} & \cdots & \mathbf{R} \mx_d
    \end{bmatrix}.
\end{align*}
If $\mz\in\mathbb{T}^d$ and $\mx$ are close enough, then \Cref{lemma:2.1} in the Appendix implies that the following first-order approximation holds:
\begin{align}
g(\mz)
\approx
g(\mx)
+
\left(\mathbf{R}_{\mx}\bm{\theta}_\mz\right)\tran
\nabla \bar g(\mx).
\label{1st taylor}
\end{align}
Here, $\bm{\theta}_\mz\in[-\pi,\pi)^d$ is the vector satisfying $\mz=\bm{\Phi}_{\mx}(\bm{\theta}_\mz)$.

Now, we establish the local linear estimator. For each observation $\mathbf{X}^{(i)}$, let
$\bm{\theta}^{(i)}\in[-\pi,\pi)^d$ denote the vector satisfying
$\mathbf{X}^{(i)}=\bm{\Phi}_{\mx}(\bm{\theta}^{(i)})$.
We define
\begin{align}
\begin{gathered}
(\hat\alpha_{\mh}(\mx),\hat{\bm\beta}_{\mh}(\mx))
:=
\argmin_{\alpha\in\mathbb{R},\,\bm\beta\in\mathbb{R}^{2d}}
n^{-1}\sum_{i=1}^n
\mathcal{L}_{\mx,\mh}(\mathbf{X}^{(i)})
\Big(
Y^{(i)}-\alpha-(\mathbf{R}_{\mx}\bm{\theta}^{(i)})\tran\bm\beta
\Big)^2 \\
\text{subject to }\;
\mathbf{D}_{\mx}\tran\bm\beta=\bm{0}_d.
\end{gathered}
\label{eq:main5.2}
\end{align}
Motivated by \eqref{constraint} and \eqref{1st taylor}, we may take
$\hat\alpha_{\mh}(\mx)$ as an estimator of $\mathbb{E}(Y\mid\mX=\mx)$. To remove the constraint in \eqref{eq:main5.2}, suppose that $\bm\beta= \big( \bm\beta_1 \tran, \dots, \bm\beta_d \tran \big)\tran$ satisfies $\mathbf{D}_{\mx}\tran\bm\beta=\bm{0}_d$. Since $\mx_\ell\tran\bm\beta_\ell=0$ and $\mx_\ell\tran\mathbf{R}\mx_\ell=0$ for each $\ell$, we can write $\bm\beta_\ell = \gamma_\ell \mathbf{R}\mx_\ell$ for some $\gamma_\ell \in \mathbb{R}$. This implies that $\bm\beta=\mathbf{R}_{\mx}\bm\gamma$ for some $\bm\gamma = (\gamma_1, \dots, \gamma_d)\tran \in \mathbb{R}^{d}$. Let $\mathbf{I}_p$ denote the $p \times p$ identity matrix for $p\in\mathbb{N}$. Since
\begin{align*}
    \mathbf{R}_{\mx}\tran \mathbf{R}_{\mx} = \begin{bmatrix}
        \left( \mathbf{R} \mx_1 \right)\tran & \cdots & \bm{0}_{2}\tran  \\
        \vdots & \ddots & \vdots \\
        \bm{0}_{2} \tran & \cdots & \left( \mathbf{R} \mx_d \right)\tran
    \end{bmatrix} \begin{bmatrix}
        \mathbf{R} \mx_1 & \cdots & \bm{0}_{2} \\
        \vdots & \ddots & \vdots \\
        \bm{0}_{2} & \cdots & \mathbf{R} \mx_d
    \end{bmatrix} = \mathbf{I}_d,
\end{align*}
\eqref{eq:main5.2} reduces to the unconstrained optimization problem
\begin{align}
\begin{gathered}
(\hat\alpha_{\mh}(\mx),\hat{\bm\gamma}_{\mh}(\mx))
=
\argmin_{\alpha\in\mathbb{R},\,\bm\gamma\in\mathbb{R}^d}
n^{-1}\sum_{i=1}^n
\mathcal{L}_{\mx,\mh}(\mathbf{X}^{(i)})
\Big(
Y^{(i)}-\alpha-\bm\gamma\tran\bm{\theta}^{(i)}
\Big)^2.
\end{gathered}
\label{eq:main5.5}
\end{align}

To get an explicit expression for $\hat{\alpha}_{\mh}(\mathbf{x})$, we define
\begin{align*}
\hat{\mu}_{\mh, 0}(\mathbf{x}) &:= n^{-1} \sum_{i=1}^{n} \mathcal{L}_{\mx, \mh} \left( \mathbf{X}^{(i)} \right),\\ \bm{\hat{\mu}}_{\mh, 1}(\mathbf{x}) &:= n^{-1} \sum_{i=1}^{n} \mathcal{L}_{\mx, \mh} \left( \mathbf{X}^{(i)} \right) \bm{\theta}^{(i)},\\
\bm{\hat{\mu}}_{\mh, 2}(\mathbf{x}) &:= n^{-1} \sum_{i=1}^{n} \mathcal{L}_{\mx, \mh} \left( \mathbf{X}^{(i)} \right) \bm{\theta}^{(i)} \left( \bm{\theta}^{(i)} \right)\tran,\\
\hat{\sigma}_{\mh}(\mathbf{x}) &:= \hat{\mu}_{\mh, 0}(\mathbf{x}) - \bm{\hat{\mu}}_{\mh, 1}(\mathbf{x})\tran  \bm{\hat{\mu}}_{\mh, 2}(\mathbf{x})^{-1} \bm{\hat{\mu}}_{\mh, 1}(\mathbf{x}),\\
\hat{W}_{\mathbf{x}, \mh} (\mX) &:= \frac{\mathcal{L}_{\mx, \mh} \left( \mathbf{X} \right)}{\hat{\sigma}_{\mh}(\mathbf{x})} \left( 1 - \bm{\hat{\mu}}_{\mh, 1}(\mathbf{x})\tran \bm{\hat{\mu}}_{\mh, 2}(\mathbf{x})^{-1} \bm{\Phi}_{\mx}^{-1} (\mX) \right),
\end{align*}
given that $\bm{\hat{\mu}}_{\mh, 2}(\mathbf{x})$ is invertible and that $\hat{\sigma}_{\mh}(\mathbf{x})\neq0$. 
According to \Cref{lemma:temp4.6} and \Cref{lemma:temp4.7} in the Appendix, $\bm{\hat{\mu}}_{\mh, 2}(\mathbf{x})$ is invertible and $\hat{\sigma}_{\mh}(\mathbf{x})>0$ with probability tending to one, under mild conditions.
By solving \eqref{eq:main5.5}, we get
\begin{gather*}
\hat{\alpha}_{\mh}(\mathbf{x}) = n^{-1} \sum_{i=1}^n \hat{W}_{\mathbf{x}, \mh} (\mathbf{X}^{(i)} ) Y^{(i)}.
\end{gather*}
In other words, we get
\[
\hat\alpha_{\mh}(\mx)
=
\argmin_{y\in\mathbb{R}}
n^{-1}\sum_{i=1}^n
\hat W_{\mx,\mh}(\mathbf{X}^{(i)})
\big(Y^{(i)}-y\big)^2.
\]

We extend this formulation to the general metric space by replacing the Euclidean distance with $d_{\mathbb{M}}$.
Specifically, we define
\begin{align*}
\hat M_{\mh,1}(\mx,y)
&:=
n^{-1}\sum_{i=1}^n
\hat W_{\mx,\mh}(\mathbf{X}^{(i)})
d_{\mathbb{M}}^2(Y^{(i)},y),\\
\hat m_{\mh,1}(\mx)
&:=
\argmin_{y\in\mathbb{M}}
\hat M_{\mh,1}(\mx,y).
\end{align*}
We call $\hat m_{\mh,1}(\mx)$ a local linear Fr\'echet regression estimator at $\mx$.

\begin{remark}
In the case of scalar responses, the local linear method of \cite{Di Marzio et al. (2009)} is based on locally approximating the target regression function $m:[-\pi, \pi)^d\rightarrow\mathbb{R}$ at a given vector $\bm{\theta} = (\theta_1, \dots, \theta_d)\tran \in [-\pi, \pi)^d$ by
\begin{align*}
m(\bm{\theta}) \approx \beta_0 + \sum_{l=1}^d \beta_l \sin(\theta_l - \psi_l),
\end{align*}
where $\beta_0, \beta_\ell \in \mathbb{R}$ and  $\bm{\theta}\approx\bm{\psi} = (\psi_1 , \dots, \psi_d)\tran \in [-\pi, \pi)^d$. This approach is fundamentally different from ours, which relies on the tangent-normal decomposition in the embedded submanifold.
\end{remark}

\section{Asymptotic theory} \label{section:3}

\setcounter{equation}{0}

\subsection{Consistency} \label{section:consistency}

In this section, we derive the consistency of our estimators at a given point $\mx\in\mathbb{T}^d$. We first introduce a condition on the function $L$.

\begin{customcon}{L} \label{con:L1}
The function $L$ satisfies that
\begin{align*}
0 < \int_0^{\infty} L^{k}(r^2) r^{j} \mathrm{d}r < \infty
\end{align*}
for all $k\in\{1,2\}$ and nonnegative integers $j$.
\end{customcon}

Similar conditions were adopted in \cite{Hall et al. (1987)}, \cite{Bai et al. (1988)} and \cite{Garcia-Portugues et al. (2013)}. Examples of $L$ satisfying the condition \ref{con:L1} include the von Mises kernel $L(r)=e^{-r}$, the exponential kernel $L(r)=e^{-\sqrt{r}}$ and the uniform kernel $L(r)=\mathds{1}_{[0, 1]} (r)$, where $\mathds{1}_{[0, 1]}$ denotes the indicator function of $[0,1]$. We also make a typical condition on the bandwidth vector $\mh$.

\begin{customcon}{B} \label{con:B1}
The bandwidth vector $\mh$ satisfies that $\lim_{n \rightarrow \infty} \left\| \mh \right\|_2= 0$ and $\lim_{n \rightarrow \infty} \left(n \prod_{\ell=1}^{d} h_\ell\right) = \infty$.
\end{customcon}

Note that \cite{Ruppert and Wand (1994)} assumed that the ratio of the largest to the smallest bandwidth components remains bounded as the sample size grows. We do not impose this restriction, allowing for more flexible bandwidth vectors. Now, we make some conditions on the distributions of $\mX$ and $Y$. Let $\mathcal{B}(\mathbb{S}^1)$ and $\mathcal{B}(\mathbb{T}^d)$ denote the Borel $\sigma$-fields of $\mathbb{S}^1$ and $\mathbb{T}^{d}$, respectively. Note that $\mathcal{B}(\mathbb{T}^d)$ equals the product $\sigma$-field of $d$ Borel $\sigma$-fields $\mathcal{B}(\mathbb{S}^1)$. Also, let $\omega_1^d$ denote the scaled toroidal measure on $\mathbb{T}^{d}$, defined as the product measure of $d$ circular measures on $\mathbb{S}^1$. Note that 
\begin{align*} 
\omega_1^d (A_1 \times \cdots \times A_d)=2^d \prod_{\ell=1}^{d}\Leb_{2}(\{t\cdot\mx:t\in[0,1],\mx\in A_\ell\}),\quad A_1, \cdots, A_d\in\mathcal{B}(\mathbb{S}^1),
\end{align*}
where $\Leb_{2}$ denotes the Lebesgue measure on $\mathbb{R}^{2}$. This measure satisfies that $\omega_1^d(\mathbb{T}^{d})=(2 \pi)^{d}$. Let $f: \mathbb{T}^{d} \to [0,\infty)$ denote the density of $\mathbf{X}$ with respect to $\omega_1^d$.

\begin{customcon}{D1} \label{con:D1}
The density $f$ satisfies that $f(\mx)>0$.
\end{customcon}

\begin{customcon}{D2} \label{con:D2}
The density $f$ is continuous at $\mx$ and is bounded on $\mathbb{T}^{d}$.
\end{customcon}


Most parametric distributions on $\mathbb{T}^{d}$, including the uniform distribution on $\mathbb{T}^{d}$ and the von Mises-Fisher distributions on $\mathbb{T}^{d}$, satisfy the conditions \ref{con:D1} and \ref{con:D2}. Let $P_{(\mathbf{X}, Y)}$ denote the joint distribution of $\left( \mathbf{X}, Y\right)$, and let $P_{\mathbf{X}}$ and $P_Y$ denote the distributions of $\mathbf{X}$ and $Y$, respectively. Also, let $P_{Y|\mathbf{X}=\mx}$ denote the conditional distribution of $Y$ given $\mathbf{X}=\mx$. We assume that $P_{Y|\mathbf{X}=\mx}$ is absolutely continuous with respect to $P_Y$ for each $\mx\in\mathbb{T}^d$, so that there exists the Radon-Nikodym derivative $d P_{Y|\mathbf{X}=\mx}/d P_Y$ of $P_{Y|\mathbf{X}=\mathbf{x}}$ with respect to $P_Y$. For each $y \in \mathbb{M}$, we define $g_y : \mathbb{T}^{d} \rightarrow [0,\infty)$ as
\begin{align}\label{g_y def}
        g_y(\mx) := \frac{d P_{Y|\mathbf{X}=\mx}}{d P_Y } (y).
\end{align}

\begin{customcon}{D3} \label{con:D3}
    The family $\{ g_y : y \in \mathbb{M} \}$ of functions is equicontinuous at $\mx$ and it holds that $\sup_{y\in\mathbb{M}}\sup_{\mz\in\mathbb{T}^d}g_y(\mz)<\infty$.
\end{customcon}

Similar conditions were used in \cite{Di Marzio et al. (2014)} and \cite{Petersen and Muller (2019)}. Now, we introduce a condition on $\mathbb{M}$. Recall the definitions of $\hat{\mu}_{\mh, 0}(\mathbf{x})$, $\bm{\hat{\mu}}_{\mh, 1}(\mathbf{x})$, $\bm{\hat{\mu}}_{\mh, 2}(\mathbf{x})$, $\hat{\sigma}_{\mh}(\mathbf{x})$ and $\hat{W}_{\mathbf{x}, \mh} (\mX)$ given in Section \ref{local linear section}. We define their population versions as
\begin{align*} 
    \tilde{\mu}_{\mh, 0}(\mathbf{x}) &:= \mathbb{E} \left[ \mathcal{L}_{\mx, \mh} \left( \mathbf{X} \right)  \right],\\ \bm{\tilde{\mu}}_{\mh, 1}(\mathbf{x}) &:= \mathbb{E} \left[ \mathcal{L}_{\mx, \mh} \left( \mathbf{X} \right) \bm{\Phi}_{\mx}^{-1} (\mathbf{X})  \right], \\
    \bm{\tilde{\mu}}_{\mh, 2}(\mathbf{x}) &:= \mathbb{E} \left[ \mathcal{L}_{\mx, \mh} \left( \mathbf{X} \right) \bm{\Phi}_{\mx}^{-1} (\mathbf{X}) \bm{\Phi}_{\mx}^{-1} (\mathbf{X})\tran \right], \\
    \tilde{\sigma}_{\mh}(\mathbf{x}) &:= \tilde{\mu}_{\mh, 0}(\mathbf{x}) - \bm{\tilde{\mu}}_{\mh, 1}(\mathbf{x})\tran  \bm{\tilde{\mu}}_{\mh, 2}(\mathbf{x})^{-1} \bm{\tilde{\mu}}_{\mh, 1}(\mathbf{x}), \\
    \tilde{W}_{\mathbf{x}, \mh} (\mX ) & := \frac{\mathcal{L}_{\mx, \mh} \left( \mathbf{X} \right)}{\tilde{\sigma}_{\mh}(\mathbf{x})} \left( 1 - \bm{\tilde{\mu}}_{\mh, 1}(\mathbf{x})\tran \bm{\tilde{\mu}}_{\mh, 2}(\mathbf{x})^{-1} \bm{\Phi}_{\mx}^{-1} (\mX) \right).
\end{align*}
We also define the population versions of $\hat{M}_{\mh, 0}(\mathbf{x}, y)$, $\hat{m}_{\mh, 0}(\mathbf{x})$, $\hat{M}_{\mh, 1}(\mathbf{x}, y)$ and $\hat{m}_{\mh, 1}(\mathbf{x})$ as
\begin{align*}
    \tilde{M}_{\mh, 0}(\mathbf{x}, y) & := \frac{\mathbb{E} \left[ \mathcal{L}_{\mx, \mh} \left( \mathbf{X} \right) d_{\mathbb{M}}^2(Y, y) \right] }{\tilde{\mu}_{\mh, 0}(\mathbf{x})}, \\
    \tilde{m}_{\mh, 0}(\mathbf{x}) & := \mathop{\mathrm{argmin}}_{y \in \mathbb{M}} \tilde{M}_{\mh, 0}(\mathbf{x}, y),\\
    \tilde{M}_{\mh,1} (\mathbf{x}, y) & := \mathbb{E} \left[ \tilde{W}_{\mathbf{x}, \mh} (\mathbf{X} ) d_{\mathbb{M}}^2(Y, y) \right],  \\
    \tilde{m}_{\mh,1} (\mathbf{x}) & := \mathop{\mathrm{argmin}}_{y \in \mathbb{M}} \tilde{M}_{\mh,1} (\mathbf{x}, y).
\end{align*}

\begin{customcon}{M1} \label{con:M1}
    For each $s\in\{0,1\}$, (i) $m_{\oplus}(\mathbf{x})$, $\tilde{m}_{\mh, s}(\mathbf{x})$ and $\hat{m}_{\mh, s}(\mathbf{x})$ uniquely exist, the latter almost surely; (ii) for any $\epsilon >0$, it holds that
    \begin{gather*}
            \liminf_{n} \inf_{y\in\mathbb{M}:\,d_{\mathbb{M}}(y, \tilde{m}_{\mh, s}(\mathbf{x})) > \epsilon}\left[ \tilde{M}_{\mh, s} (\mathbf{x}, y) - \tilde{M}_{\mh, s} (\mathbf{x}, \tilde{m}_{\mh, s}(\mathbf{x})) \right] > 0, \\
        \inf_{y\in\mathbb{M}:\,d_{\mathbb{M}}(y, m_{\oplus}(\mathbf{x})) > \epsilon} \left[ M_{\oplus} (\mathbf{x}, y) - M_{\oplus} (\mathbf{x}, m_{\oplus}(\mathbf{x})) \right] > 0.
    \end{gather*}
\end{customcon}

The condition \ref{con:M1} is satisfied for various metric spaces. For example, the spaces in Examples \ref{example1} and \ref{example4} satisfy this condition by Proposition 1 of \cite{Petersen and Muller (2019)} and Section B.3 of \cite{Zhou and Muller (2022)}, respectively. For the space in Example \ref{example2}, the condition \ref{con:M1}-(i) is satisfied under various manifold conditions; see \cite{Afsari (2011)} and \cite{Charlier (2013)}, for example. Now, we present the consistency.

\begin{theorem} \label{thm:consistency}
    Assume that the conditions \ref{con:L1}, \ref{con:B1}, \ref{con:D1}--\ref{con:D3} and \ref{con:M1} hold. For each $s\in\{0,1\}$, it holds that
    \begin{align*}
        d_{\mathbb{M}}(\hat{m}_{\mh, s}(\mathbf{x}), m_{\oplus}(\mathbf{x})) = o_{\mathbb{P}}(1).
    \end{align*}
\end{theorem}

\subsection{Rate of convergence}\label{section:rate of convergence}

In this section, we establish the convergence rates of the proposed estimators.
To this end, we introduce additional conditions.
Let $\bar{f} : \mathbb{R}^{2d} \setminus \mathcal{N} \to [0, \infty)$ denote an extension of the density $f$, defined as
    \begin{align*}
        \bar{f} (\mz) := f \left( \frac{\mz_1 \tran}{\| \mz_1 \|_2} , \dots, \frac{\mz_d \tran}{\| \mz_d \|_2} \right), \quad \mz = \big( \mz_1 \tran, \dots, \mz_d \tran \big)\tran \in \mathbb{R}^{2d} \setminus\mathcal{N}.
    \end{align*}
Additionally, let $\nabla^2 \bar{f}$ denote the Hessian of $\bar{f}$.
With a slight abuse of notation, we again use $\|\cdot\|_2$ to denote the matrix operator norm.

\begin{customcon}{D4} \label{con:D4}
    The homogeneous extension $\bar{f}$ is twice differentiable on $\mathbb{R}^{2d} \setminus\mathcal{N}$ and it holds that $\sup_{\mz\in\mathbb{T}^d}\|\nabla^2 \bar{f}(\mz)\|_2<\infty$.
\end{customcon}

The condition \ref{con:D4} implies the condition \ref{con:D2}. For $g_y$ defined in \eqref{g_y def}, let $\bar{g}_y : \mathbb{R}^{2d} \setminus\mathcal{N} \to [0,\infty)$ denote the homogeneous extension of $g_y$, defined as
\begin{align*}
    \bar{g}_y (\mz) := g_y \left( \left( \frac{\mz_1 \tran}{\| \mz_1 \|_2} , \dots, \frac{\mz_d \tran}{\| \mz_d \|_2} \right) \tran \right), \quad \mz = \big( \mz_1 \tran, \dots, \mz_d \tran \big)\tran \in \mathbb{R}^{2d} \setminus\mathcal{N}.
\end{align*}
Additionally, let $\nabla \bar{g}_y$ and $\nabla^2 \bar{g}_y$ denote the gradient and Hessian of $\bar{g}_y$, respectively.

\begin{customcon}{D5} \label{con:D5}
    The homogeneous extension $\bar{g}_y$ is twice differentiable on $\mathbb{R}^{2d} \setminus\mathcal{N}$ and it holds that
        $\sup_{y \in \mathbb{M}} \sup_{\mz\in\mathbb{T}^d} g_y (\mz) < \infty$, $\sup_{y \in \mathbb{M}} \sup_{\mz\in\mathbb{T}^d} \| \nabla \bar{g}_y (\mz) \|_{2} < \infty$ and $\sup_{y \in \mathbb{M}} \sup_{\mz\in\mathbb{T}^d} \| \nabla^2 \bar{g}_y (\mz) \|_{2} < \infty$.
\end{customcon}

Conditions of this type were employed in earlier works, including
\cite{Di Marzio et al. (2014)} and \cite{Petersen and Muller (2019)}. The condition~\ref{con:D5} implies
the condition~\ref{con:D3}. Now, we introduce two additional conditions on the metric space $\mathbb{M}$.

\begin{customcon}{M2} \label{con:M2}
    For each $s\in\{0,1\}$, there exist constants $h_{\oplus} >0$, $\eta_{\oplus} >0$, $C_{\oplus} >0$ and $\beta_{\oplus} \in (1, \infty)$ such that
    \begin{align*}
         \tilde{M}_{\mh, s} (\mathbf{x}, y) - \tilde{M}_{\mh, s} (\mathbf{x}, \tilde{m}_{\mh, s} (\mathbf{x}) )  \geq C_{\oplus}\cdot d_{\mathbb{M}}(y, \tilde{m}_{\mh, s} (\mathbf{x}) )^{\beta_{\oplus}}
    \end{align*}
    whenever $\left\| \mh \right\|_2 < h_{\oplus}$ and $d_{\mathbb{M}}(y, \tilde{m}_{\mh, s}(\mathbf{x})) < \eta_{\oplus}$, and that
    \begin{align*}
        M_{\oplus} (\mathbf{x}, y) - M_{\oplus} (\mathbf{x}, m_{\oplus}(\mathbf{x}))  \geq C_{\oplus} \cdot d_{\mathbb{M}}(y, m_{\oplus}(\mathbf{x}))^{\beta_{\oplus}}
    \end{align*}
    whenever $d_{\mathbb{M}}(y, m_{\oplus}(\mathbf{x})) < \eta_{\oplus}$.
\end{customcon}

Conditions analogous to the condition \ref{con:M2} appeared in earlier works (e.g., \cite{Petersen and Muller (2019)}). When $\eta_\oplus$ is sufficiently large, the condition~\ref{con:M2} implies the condition~\ref{con:M1}-(ii). The condition~\ref{con:M2} with $\beta_\oplus = 2$ holds for the spaces in Examples~\ref{example1} and~\ref{example4} by Proposition~1 of \cite{Petersen and Muller (2019)} and Section B.3 of \cite{Zhou and Muller (2022)}, respectively. An analogous result holds with $\beta_\oplus = 2$ for the space in Example~\ref{example2}, under the additional assumptions in Proposition~3 of \cite{Petersen and Muller (2019)}. To introduce the next condition, let $B_{\mathbb{M}}(y, \delta)$ denote the open ball in $\mathbb{M}$ centered at $y \in \mathbb{M}$ with radius $\delta > 0$ and let $N(r, B_{\mathbb{M}}(y, \delta), d_{\mathbb{M}})$ denote the $r$-covering number of $B_{\mathbb{M}}(y, \delta)$ with respect to the metric $d_{\mathbb{M}}$.

\begin{customcon}{M3} \label{con:M3}
There exist constants $r_{\mathbb{M}}>0$ and $\alpha_{\mathbb{M}} \in (0, 1]$ such that
\begin{gather*}
    \sup_{y \in \mathbb{M}:d_{\mathbb{M}}(y,m_{\oplus}(\mx))<r_{\mathbb{M}}} \int_{0}^{\frac{1}{2}} \sqrt{1+\log N(\delta \epsilon, B_{\mathbb{M}}(y, \delta), d_{\mathbb{M}})} \, \mathrm{d} \epsilon =O(\delta^{\alpha_{\mathbb{M}}-1})  ~ \mathrm{as} ~ \delta \rightarrow 0.
\end{gather*}
\end{customcon}

The metric entropy condition~\ref{con:M3} extends the corresponding assumption in \cite{Petersen and Muller (2019)}, which appears as the special case $\alpha_{\mathbb{M}} = 1$. This extension accommodates a wider class of metric spaces and is satisfied by the spaces in Examples~\ref{example1}--\ref{example2}; see Propositions 1 and 2 of \cite{Im et al. (2025)} and Section B.3 of \cite{Zhou and Muller (2022)}, for example. Now, we derive the convergence rates of the proposed estimators.

\begin{theorem} \label{thm:main4.3}
    Assume that the conditions \ref{con:L1}, \ref{con:B1}, \ref{con:D1}, \ref{con:D4}, \ref{con:D5} and \ref{con:M1}--\ref{con:M3} hold. For each $s\in\{0,1\}$ and for the constants $\beta_{\oplus} \in (1, \infty)$ and $\alpha_{\mathbb{M}} \in (0, 1]$ defined in the conditions \ref{con:M2} and \ref{con:M3}, respectively, it holds that
    \begin{align*}
        d_{\mathbb{M}}(\hat{m}_{\mh, s}(\mathbf{x}), m_{\oplus}(\mathbf{x})) = 
        O \left( \left\| \mh \right\|_2^{\frac{2}{\beta_{\oplus} - 1}} \right) +  O_{\mathbb{P}} \left( \left( n \prod_{\ell=1}^{d} h_\ell \right) ^{\frac{-1}{2 \left( \beta_{\oplus} - \alpha_{\mathbb{M}} \right)}} \right).
    \end{align*}
\end{theorem}

When $\alpha_{\mathbb{M}} = 1$, consider bandwidths of the form
$\mh = \left( c\cdot n^{-\gamma}, \dots, c\cdot n^{-\gamma} \right)\tran$ for some constant $c>0$.
The error rate in \Cref{thm:main4.3} is then optimized by choosing $\gamma = 1/(d+4)$, which
yields
\begin{align*}
d_{\mathbb{M}}\!\left(\hat{m}_{\mh,s}(\mathbf{x}), m_{\oplus}(\mathbf{x})\right)
=
O_{\mathbb{P}}\!\left(
n^{-\frac{2}{(d+4)(\beta_{\oplus}-1)}}
\right).
\end{align*}
When $\beta_{\oplus} = 2$, which is achieved by various $(\mathbb{M},d_{\mathbb{M}})$, this rate simplifies to
\begin{align*}
d_{\mathbb{M}}\!\left(\hat{m}_{\mh,s}(\mathbf{x}), m_{\oplus}(\mathbf{x})\right)
=
O_{\mathbb{P}}\!\left(n^{-2/(d+4)}\right),
\end{align*}
which coincides with the optimal error rate in nonparametric regression with a $d$-dimensional predictor.

It is worth noting that the local constant estimator \(\hat{m}_{\mh,0}(\mathbf{x})\) attains the same error rate as the local linear estimator \(\hat{m}_{\mh,1}(\mathbf{x})\) for any \(\mathbf{x} \in \mathbb{T}^d\). This contrasts with classical nonparametric smoothing results, in which boundary effects typically lead to slower error rates for local constant estimators. Since the torus \(\mathbb{T}^d\) has no boundary, such effects do not arise in the present setting. Nevertheless, the numerical studies presented in the following sections reveal that the local linear estimator exhibits superior finite-sample performance compared with the local constant estimator.

\section{Simulation study} \label{section:4}

We conducted a simulation study with $d=2$ and a spherical response ($\mathbb{M}=\mathbb{S}^2$). Note that $\mathbb{S}^2$ is a two-dimensional compact Riemannian manifold without boundary, equipped with the geodesic distance $$d_{\mathbb{S}^2}(\mv,\bu):=\arccos(\mv\tran \bu).$$ We compared the performance of our estimators $\hat{m}_{\mh,0}$ and $\hat{m}_{\mh,1}$ with the local constant Fr\'echet regression estimator $\hat{m}_{h,0}^{\rm TW}$ of \cite{Tucker and Wu (2025)} and the local linear Fr\'echet regression estimator $\hat{m}_{\mh,1}^{\rm PM}$ of \cite{Petersen and Muller (2019)}. Although $\hat{m}_{\mh,1}^{\rm PM}$ is designed for Euclidean predictors, we applied it by interpreting the angles of the toroidal predictors as Euclidean predictors and using a product kernel.

Bandwidth parameters for these methods were selected using 5-fold cross-validation. For the estimators employing multiple bandwidths 
($\hat{m}_{\mh,0}$, $\hat{m}_{\mh,1}$ and $\hat{m}_{\mh,1}^{\rm PM}$), we adopted a two-stage grid search to reduce the computational burden while maintaining accuracy. In the first stage, we conducted a coarse search over a $10 \times 10$ grid. For our estimators, the initial grid was set to $\{0.1 \times k : k = 1, \ldots, 10\}^2$. For $\hat{m}_{\mh,1}^{\rm PM}$, we considered a wider initial grid, namely $\{0.5 \times k : k = 1, \ldots, 20\}^2$, reflecting differences in the underlying kernel schemes. In the second stage, we performed a finer search over a $5 \times 5$ grid centered at the bandwidth vector $\mh_{\mathrm{opt}}^{(1)}$ selected in the first stage. The grid spacing in this stage was set to one-quarter of that used in the first stage, namely $0.025$ for our estimators and $0.125$ for $\hat{m}_{\mh,1}^{\rm PM}$. For $\hat{m}_{h,0}^{\rm TW}$, we searched its optimal bandwidth over the one-dimensional grid $\{0.01 \times k : k = 1, \ldots, 50\}$. We employed the von Mises kernel for our estimators, and the Epanechnikov kernel for $\hat{m}_{h,0}^{\rm TW}$ and $\hat{m}_{\mh,1}^{\rm PM}$.

We specified the regression function $m_{\oplus} : \mathbb{T}^2 \to \mathbb{S}^2$ as
\begin{align*}
    m_{\oplus}(\mathbf{x}) = \frac{( \cos \psi , \sin \phi , \sin \psi \cos \phi )\tran}{\| ( \cos \psi , \sin \phi , \sin \psi \cos \phi )\tran \|_2}, \quad \mathbf{x} = (\cos \psi ,  \sin \psi,  \cos \phi ,  \sin \phi) \tran \in \mathbb{T}^2,
\end{align*}
where $\psi, \phi  \in [-\pi, \pi)$. For sample sizes $n\in\{50, 100, 200\}$ and noise levels $\sigma\in\{0.1,0.25\}$, we generated i.i.d. random samples $\{ (\mathbf{X}^{(i)}, Y^{(i)}) \}_{i=1}^n$ over $R=100$ Monte Carlo replications according to
\begin{align*}
    \mathbf{X}^{(i)}  \sim \mathcal{U} (\mathbb{T}^2), \quad Y^{(i)} | \mathbf{X}^{(i)}  \sim \mathcal{VMF}(m_{\oplus}(\mathbf{X}^{(i)}), 1 / \sigma ),
\end{align*}
where $\mathcal{U} (\mathbb{T}^2)$ denotes the uniform distribution on $\mathbb{T}^2$, and $\mathcal{VMF}(\mu, \kappa)$ denotes the von Mises-Fisher distribution with mean direction $\mu\in\mathbb{S}^2$ and concentration parameter $\kappa>0$. Note that $\mathcal{VMF}(\mu, \kappa)$ exhibits larger variability for smaller values of $\kappa$. We evaluated the performance using the mean integrated squared error (MISE), defined as
\begin{align*}
    & \frac{1}{R}\sum_{r=1}^R \int_{\mathbb{T}^2} d_{\mathbb{S}^2}^2 \big(\hat{m}_{\oplus}^{[r]}(\mathbf{x}), m_{\oplus}(\mathbf{x}) \big) \mathrm{d} \omega_1^2 (\mathbf{x})\\
    & = \frac{1}{R}\sum_{r=1}^R \int_{-\pi}^{\pi}\int_{-\pi}^{\pi} d_{\mathbb{S}^2}^2 \big(\hat{m}_{\oplus}^{[r]}(\cos \psi ,  \sin \psi,  \cos \phi ,  \sin \phi), m_{\oplus}(\cos \psi ,  \sin \psi,  \cos \phi ,  \sin \phi) \big) \,\mathrm{d} \phi \,\mathrm{d} \psi,
\end{align*}
where $\hat{m}_{\oplus}^{[r]}(\mathbf{x})$ denotes the estimator evaluated at $\mx$ and computed from the $r$-th Monte Carlo replication.

Table~\ref{table:simulation} summarizes the simulation results. 
The results show that our local linear estimator $\hat{m}_{\mh,1}$ attains the smallest MISE across all scenarios. 
In contrast, the estimator $\hat{m}_{\mh,1}^{\rm PM}$ exhibits markedly inferior performance, as it fails to account for the intrinsic geometry of $\mathbb{T}^2$. Additionally, our local constant estimator $\hat{m}_{\mh,0}$ outperforms $\hat{m}_{h,0}^{\rm TW}$ at the lower noise level, whereas the opposite pattern is observed at the higher noise level. Overall, these results indicate that the proposed methods are promising options for Fr\'echet regression with toroidal predictors.

\begin{table}
\centering
\caption{\it MISE comparison in spherical-toroidal regression.}
\begin{tabular}{ccccccc}
    \hline
    $\sigma$ & $n$ & & $\hat{m}_{\mh,0}$ & $\hat{m}_{\mh,1}$ & $\hat{m}_{h,0}^{\rm TW}$ & $\hat{m}_{\mh,1}^{\rm PM}$ \\ \hline \hline
    \multirow{3}*{0.1} & 50 & & 4.672 & \textbf{3.026}  & 4.861 & 30.321 \\ 
    & 100 & & 2.057 & \textbf{1.273} & 2.237 & 15.662 \\ 
    & 200 & & 1.034 & \textbf{0.635} & 1.148 & 6.949 \\  \hline \hline
    \multirow{3}*{0.25} & 50 & & 9.523 & \textbf{8.373} & 8.448 & 33.968 \\ 
    & 100 & & 4.065 & \textbf{3.435} & 3.930 & 18.083 \\ 
    & 200 & & 2.056 & \textbf{1.623} & 2.031 & 8.235 \\  \hline
\end{tabular}
\label{table:simulation}
\end{table}

\section{Real data analysis}\label{section:5}
\setcounter{equation}{0}

We analyzed the New York taxi network data obtained from 
\url{https://www.nyc.gov/site/tlc/about/tlc-trip-record-data.page}. Following \cite{Zhou and Muller (2022)}, we focused on yellow taxi trips within Manhattan, which were aggregated into 13 regions. To investigate the temporal dynamics of taxi demand, we constructed a $13 \times 13$ graph Laplacian for each hour on Sundays from January~1,~2021 to December~31,~2024. 
Specifically, for each hour, we formed a $13 \times 13$ adjacency matrix whose $(k,\ell)$-th entry represents the number of yellow taxi trips between the $k$-th and $\ell$-th regions. Each adjacency matrix was then transformed into a graph Laplacian.

We modeled the predictor space as the two-dimensional torus
$\mathbb{T}^2 = \mathbb{S}^1 \times \mathbb{S}^1$ to capture the dual periodic structure induced by daily and annual cycles.
In particular, the hour $i_1 \in \{0,\dots,23\}$ of the day and the $i_2$-th day of the year were encoded as
\begin{align*}
    \mathbf{X} =
    \Bigg(
    \cos\!\left(\frac{2\pi(i_1+0.5)}{24}\right),
    \sin\!\left(\frac{2\pi(i_1+0.5)}{24}\right),
    \cos\!\left(\frac{2\pi(i_2-0.5)}{D}\right),
    \sin\!\left(\frac{2\pi(i_2-0.5)}{D}\right)
    \Bigg)\tran\in\mathbb{T}^2,
\end{align*}
where $D=365$ for non-leap years and $D=366$ for leap years. 
The response variable $Y\in\mathcal{L}_{13}$ corresponds to the associated graph Laplacian. 

Data from 2021 and 2022 were used for training, while data from 2023 were used for bandwidth validation. We compared the predictive performance of the four estimators $\hat{m}_{\mh,0}$, $\hat{m}_{\mh,1}$, $\hat{m}_{h,0}^{\rm TW}$ and $\hat{m}_{\mh,1}^{\rm PM}$, considered in Section \ref{section:4}, on data from 2024. For bandwidth selection of $\hat{m}_{\mh,0}$, $\hat{m}_{\mh,1}$ and $\hat{m}_{\mh,1}^{\rm PM}$, we again employed a two-stage grid search. In the first stage, we used the grid $\{0.02 \times k : k = 1,\dots,10\} \times \{0.01 \times k : k = 1,\dots,10\}$ for our estimators, and
$\{0.02 \times k : k = 1,\dots,20\} \times \{0.01 \times k : k = 1,\dots,20\}$ for $\hat{m}_{\mh,1}^{\rm PM}$. The second stage was performed in the same manner as described in Section~\ref{section:4}. 
For $\hat{m}_{h,0}^{\rm TW}$, we searched its bandwidth over the grid $\{0.004 \times k : k = 1,\dots,50\}$.

The resulting average squared prediction errors, computed using the Frobenius distance, were $8.218 \times 10^{5}$ for $\hat{m}_{\mh,0}$, $\bf{8.100 \times 10^{5}}$ for $\hat{m}_{\mh,1}$, 
$8.937 \times 10^{5}$ for $\hat{m}_{h,0}^{\rm TW}$ 
and $8.393 \times 10^{5}$ for $\hat{m}_{\mh,1}^{\rm PM}$. 
These results confirm that our local linear approach outperforms the local constant approaches and highlight the importance of respecting the underlying toroidal geometry.

\begin{figure}
    \centering
    \includegraphics[width=\linewidth]{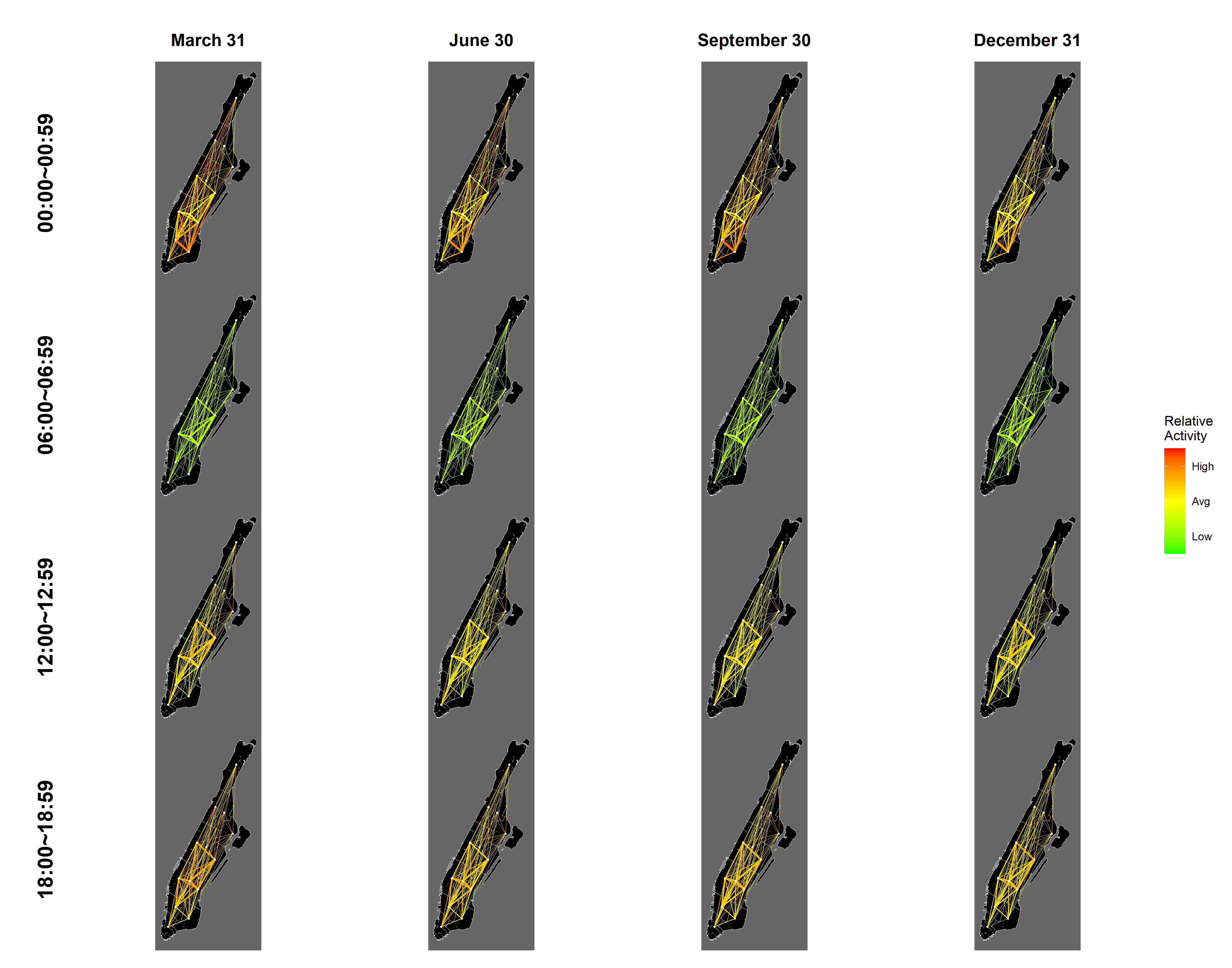}
    \caption{\it Fitted yellow taxi networks in Manhattan at selected hours and calendar dates. Edge thickness is proportional to absolute ridership, while color indicates relative ridership compared with the overall data average.}
    \label{fig:real_data_vis}
\end{figure}

To interpret the effect of $\mX$ on $Y$, we refitted the model to the full dataset using the best-performing estimator $\hat{m}_{\mh,1}$. Figure~\ref{fig:real_data_vis} displays the fitted taxi networks at selected hours and calendar dates when the corresponding day is Sunday. Across all dates, taxi ridership is high during late-night and evening hours and low in the early morning, likely reflecting increased demand during periods of social activity and reduced public transportation service. In addition, compared with other periods, taxi ridership in late December appears to be relatively lower, which may be attributed to colder weather and a greater tendency for people to stay at home during the holiday season. Overall, this figure illustrates that taxi demand is strongly influenced by both the time of day and the calendar date.

\section*{Acknowledgements}

Chang Jun Im was supported by the National Research Foundation of Korea grant funded by the Korea government(MSIT) (No. RS-2025-00515381). Jeong Min Jeon was supported by the National Research Foundation of Korea grant funded by the Korea government(MSIT) (No. RS-2023-00211910).

\vspace{1cm}
\noindent
\textbf{Corresponding Author} \\
Jeong Min Jeon \\
Department of Statistics and School of Transdisciplinary Innovations, \\
Seoul National University, \\
1 Gwanak-ro, Gwanak-gu, Seoul 08826, South Korea \\
\textit{E-mail address:} jeongmin.jeon.stat@gmail.com

\newpage

\begin{appendix}




\centerline{\bf \Large Supplementary Material to}
\smallskip
\centerline{\bf \Large ``Local Fr\'echet regression with toroidal predictors''}
\bigskip
\normalsize
\centerline{\bf Chang Jun Im and Jeong Min Jeon}
\bigskip
\centerline{\bf Seoul National University, South Korea}
\bigskip

The Supplementary Material consists of two parts. In the first part, we provide the proofs of \Cref{thm:consistency} and \Cref{thm:main4.3} for the local constant estimator. In the second part, we prove them for the local linear estimator.

\section{Proofs of \Cref{thm:consistency} and \Cref{thm:main4.3} for  $\hat{m}_{\mh, 0}(\mathbf{x})$} \label{proofs:local constant}

\setcounter{equation}{0}
\setcounter{subsection}{0}

To derive the consistency and error rates for $d_{\mathbb{M}}(\hat{m}_{\mh, 0}(\mathbf{x}), m_{\oplus}(\mathbf{x}))$, we establish the properties of $d_{\mathbb{M}}(\tilde{m}_{\mh, 0}(\mathbf{x}), m_{\oplus}(\mathbf{x}))$ and $ d_{\mathbb{M}}(\hat{m}_{\mh, 0}(\mathbf{x}), \tilde{m}_{\mh, 0}(\mathbf{x}))$. The latter properties are obtained by investigating $\tilde{M}_{\mh, 0}(\mathbf{x}, y) - M_{\oplus}(\mathbf{x}, y)$ and $\hat{M}_{\mh, 0}(\mathbf{x}, y) - \tilde{M}_{\mh, 0}(\mathbf{x}, y)$, respectively.

First, we collect some notations. Let $\mathbb{N}_{\geq 0}$ denote the set of all non-negative integers. We write $\mathbf{j} = (j_1, \dots, j_d)\tran$ for $j_\ell\in\mathbb{N}_{\geq 0}$. For $k\in\{1,2\}$, we define
\begin{align} \begin{split} \label{eq:A.1.4}
    c_{\mh, \mathbf{j}, k}(L) := 2^{d} \left[ \prod_{=1}^{d} \int_{0}^{\pi} L^k \left( \frac{1 - \cos \theta_\ell }{h_\ell^2} \right) 
    \theta_{\ell} ^{j_\ell} \mathrm{d} \theta_\ell \right].
\end{split} \end{align}
To simplify \eqref{eq:A.1.4}, we define a function $\bm{L}_{\mh} :\mathbb{R}^{d} \to [0, \infty)$ as
\begin{align*}
    \bm{L}_{\mh} (\bu) := \prod_{\ell=1}^{d} L \left( \frac{1 - \cos u_\ell }{h_\ell^2} \right),\quad \bu = (u_1, \dots, u_d)\tran.
\end{align*}
Then, we get the expression
\begin{align*}
    c_{\mh, \mathbf{j}, k}(L) = 2^{d} \int_{[0, \pi)^{d}} \bm{L}_{\mh}^k (\bm{\theta})  \bm{\theta}^{\mathbf{j}} \mathrm{d} \bm{\theta}, 
\end{align*}
where $\bu^{\mathbf{j}} = \prod_{\ell=1}^{d} u_\ell^{j_\ell}$. Note that
\begin{align} \begin{split} \label{symmetricproperty}
    \int_{[-\pi, \pi)^{d}} \bm{L}_{\mh}^k (\bm{\theta})  \bm{\theta}^{\mathbf{j}} \mathrm{d} \bm{\theta} = \begin{cases}
            0  &  \text{some~} j_\ell \text{~is an odd number},\\
            c_{\mh, \mathbf{j}, k}(L)  &  \text{otherwise}.
       \end{cases}
\end{split} \end{align}
Also, it holds that
\begin{align}
    \mathcal{L}_{\mx, \mh} \left( \bm{\Phi}_{\mx} \left( \bm{\theta} \right) \right) = \prod_{\ell=1}^{d} L \left( \frac{1 - \left( \mx_\ell \cos \theta_\ell + \left( \mathbf{R} \mx_\ell \right) \sin \theta_\ell \right)\tran \mx_\ell}{h_\ell^2} \right) = \bm{L}_{\mh} (\bm{\theta}) \label{eq:2.7}
\end{align}
for every $\bm{\theta} \in [-\pi, \pi)^{d}$, where $\mathcal{L}_{\mx, \mh}$ and $\bm{\Phi}_{\mx}$ are defined in \eqref{Lx def} and \eqref{tangent-normal}, respectively. We write $\rho(\mh) \coloneqq \prod_{\ell=1}^{d} h_\ell$ and $| \mathbf{j} | \coloneqq \sum_{\ell=1}^{d} j_\ell$. 

\subsection{Lemmas for proof of \Cref{thm:consistency}}

\begin{lemma} \label{lemma:1.2}
    Let $h : \mathbb{T}^{d} \to \mathbb{R}$ be an integrable function. Then, for any $\mx \in \mathbb{T}^{d}$, it holds that
    \begin{align*}
        \int_{\mathbb{T}^{d}} h(\mz) \omega_{1}^{d}(\mathrm{d} \mz) = \int_{[-\pi, \pi)^{d}}  h \left( \bm{\Phi}_{\mx} \left( \bm{\theta} \right) \right) \mathrm{d} \bm{\theta}.
    \end{align*}
\end{lemma}

\begin{proof}
By Lemma 2 of Garc\'{\i}a-Portugu\'{e}s et al. (2013), we get
\begin{align*}
    \int_{\mathbb{S}^{1}} h_{\ell}(\mz_{\ell}) \omega_{1}(\mathrm{d} \mz_\ell) = \int_{-\pi}^{\pi}  h_{\ell}\left( \bm{\Phi}_{\mx, \ell} \left( \theta_{\ell} \right) \right)\mathrm{d} \theta_{\ell}, \quad \ell=1,\ldots,d
\end{align*}
for any integrable function $h_{\ell} : \mathbb{S}^{1} \to \mathbb{R}$. The lemma follows from Fubini's theorem.
\end{proof}

\begin{lemma} \label{lemma:1.4}
    Assume that the condition \ref{con:L1} holds. Then, for any $\mathbf{j} = (j_1, \dots, j_d)$ with $j_\ell\in\mathbb{N}_{\geq 0}$ and $k\in\{1,2\}$, it holds that
    \begin{align*}
        \lim_{\mh \to \bm{0}_d} \frac{c_{\mh, \mathbf{j}, k}(L)}{\mh^{\mathbf{j}} \rho(\mh)} =2^{\frac{3d + \left| \mathbf{j} \right|}{2}} \left[ \prod_{\ell=1}^{d} \int_0^{\infty} L^{k}(r^2) r^{j_\ell} \mathrm{d}r \right].
    \end{align*}
\end{lemma}

\begin{proof} [Proof of \Cref{lemma:1.4}]
By Lemma B.2 of Im et al. (2025), we get
\begin{align}
    \lim_{h_\ell \to 0} 2 \left[ \int_{0}^{\pi} L^k \left( \frac{1 - \cos \theta_\ell }{h_\ell^2} \right) 
    \theta_\ell^{j_\ell}  \mathrm{d} \theta_\ell \right] h_\ell^{-(j_\ell+1)} = 2^{\frac{j_\ell+3}{2}} \int_0^{\infty} L^{k}(r^2) r^{j_\ell} \mathrm{d}r, \quad \ell=1,\ldots,d. \label{eq:A.1.6}
\end{align}
By taking $\prod_{\ell=1}^d$ on both sides of \eqref{eq:A.1.6}, we get the desired result.
\end{proof}

\begin{lemma} \label{lemma:new2.2}
    Assume that the conditions \ref{con:L1} and \ref{con:D2} hold and that $\lim_{n\rightarrow\infty} \| \mh \|_2=0$. Then, for any $k\in\{1,2\}$, it holds that
    \begin{align*}
        \mathbb{E}\left[ \mathcal{L}_{\mx, \mh}^k \left( \mathbf{X} \right) \right] -  c_{\mh, \bm{0}_d, k}(L) f (\mx) = o\left(\rho(\mh)\right).
    \end{align*}
\end{lemma}

\begin{proof} [Proof of Lemma \ref{lemma:new2.2}]
    
    Combining \eqref{eq:2.7} and \Cref{lemma:1.2}, we get
    \begin{align*}
        \mathbb{E}\left[ \mathcal{L}_{\mx, \mh}^k \left( \mathbf{X} \right) \right]  = & \int_{\mathbb{T}^{d}} \mathcal{L}_{\mx, \mh}^k \left( \mz \right) f( \mz ) \omega_d (\mathrm{d} \mz )  \\
        = &\int_{[-\pi, \pi)^{d}}  \mathcal{L}_{\mx, \mh}^k \left( \bm{\Phi}_{\mx} \left( \bm{\theta} \right) \right) f \left( \bm{\Phi}_{\mx} \left( \bm{\theta} \right) \right)  \mathrm{d} \bm{\theta} \\
        = & \int_{[-\pi, \pi)^{d}}  \bm{L}_{\mh}^k (\bm{\theta}) f \left( \bm{\Phi}_{\mx} \left( \bm{\theta} \right) \right)  \mathrm{d} \bm{\theta}. 
    \end{align*}
    Also, \eqref{symmetricproperty} implies that
    \begin{align} \begin{split} \label{eq:new7.35}
          \mathbb{E}\left[ \mathcal{L}_{\mx, \mh}^k \left( \mathbf{X} \right) \right] -  c_{\mh, \bm{0}_d, k}(L) f (\mx) = \int_{[-\pi, \pi)^{d}} \bm{L}_{\mh}^k (\bm{\theta})  \Big( f \left( \bm{\Phi}_{\mx} \left( \bm{\theta} \right) \right) - f(\mx) \Big)  \mathrm{d} \bm{\theta}.
    \end{split} \end{align}
    Let $\epsilon>0$ be any given constant. By the condition \ref{con:D2}, there exists a constant $\delta_{\epsilon} \in (0, \pi)$, depending on $\epsilon$, such that
    \begin{align}
       \left\| \mz - \mx \right\|_2 < \delta_{\epsilon} ~ \Rightarrow ~ | f ( \mz )  - f ( \mx )  | < \epsilon, \quad  \mz \in \mathbb{T}^{d}. \label{eq:new7.36}
    \end{align}
    For any $\bm{\theta} \in [-\pi, \pi)^{d}$, it holds that
    \begin{align}
        \left\| \bm{\Phi}_{\mx} \left( \bm{\theta} \right) - \mx \right\|_2^2 = 2 \sum_{\ell=1}^{d} \left( 1 - \cos \theta_\ell \right) = 4 \sum_{\ell=1}^{d} \sin^2 \left( \frac{\theta_\ell}{2} \right) \leq \| \bm{\theta} \|_2^2. \label{eq:new7.37}
    \end{align}
    We define sets $E_0, E_1, \dots, E_d \subset [-\pi, \pi)^{d}$ as
    \begin{align*}
        E_0 := \left(- \frac{\delta_{\epsilon}}{\sqrt{d}}, \frac{\delta_{\epsilon}}{\sqrt{d}} \right)^{d}, \quad E_m := \left\{ ( z_1, \dots, z_d ) \tran  \in [-\pi, \pi)^{d} : | z_m | \geq \frac{\delta_{\epsilon}}{\sqrt{d}} \right\}, \quad m=1, \dots, d.
    \end{align*}
    Note that $[-\pi, \pi)^{d} = \bigcup_{m=0}^{d} E_m$. By \eqref{eq:new7.35}, we get
    \begin{align} \begin{split} \label{eq:new7.38}
          \left| \mathbb{E}\left[ \mathcal{L}_{\mx, \mh}^k \left( \mathbf{X} \right) \right] -  c_{\mh, \bm{0}_d, k}(L) f (\mx) \right| \le \sum_{m=0}^{d} \int_{E_m} \bm{L}_{\mh}^k (\bm{\theta})  \left| f \left( \bm{\Phi}_{\mx} \left( \bm{\theta} \right) \right) - f(\mx) \right|  \mathrm{d} \bm{\theta}.
    \end{split} \end{align}
    By \eqref{eq:new7.37}, we also get
    \begin{align} \begin{split} \label{eq:new7.39}
          \left\| \bm{\Phi}_{\mx} \left( \bm{\theta} \right) - \mx \right\|_2 \le \| \bm{\theta} \|_2 \le \delta_{\epsilon}
    \end{split} \end{align}
    for any $\bm{\theta} \in E_0$. 
    Combining \eqref{eq:new7.36} and \eqref{eq:new7.39}, we get
    \begin{align} \begin{split} \label{eq:new7.40}
          \left| f \left( \bm{\Phi}_{\mx} \left( \bm{\theta} \right) \right) - f(\mx) \right| \le 
          \begin{cases}
            \epsilon   & \text{if~}\bm{\theta} \in E_0, \\
            2 \sup_{\mz \in \mathbb{T}^{d}} f ( \mz )   &  \text{if~} \bm{\theta} \notin E_0.
           \end{cases}
    \end{split} \end{align}
    Combining \eqref{symmetricproperty}, \eqref{eq:new7.38} and \eqref{eq:new7.40}, we also get
    \begin{align} \begin{split} \label{eq:new7.41}
          \Big| \mathbb{E}\left[ \mathcal{L}_{\mx, \mh}^k \left( \mathbf{X} \right) \right] -  c_{\mh, \bm{0}_d, k}(L) f (\mx) \Big| & \le \epsilon \int_{E_0} \bm{L}_{\mh}^k (\bm{\theta})  \mathrm{d} \bm{\theta} + 2 \sup_{\mz \in \mathbb{T}^{d}} f ( \mz ) \sum_{m=1}^{d} \int_{E_m} \bm{L}_{\mh}^k (\bm{\theta})  \mathrm{d} \bm{\theta} \\
          & \le \epsilon \int_{[-\pi, \pi)^{d}} \bm{L}_{\mh}^k (\bm{\theta})  \mathrm{d} \bm{\theta} + 2 \sup_{\mz \in \mathbb{T}^{d}} f ( \mz ) \sum_{m=1}^{d} \int_{E_m} \bm{L}_{\mh}^k (\bm{\theta})  \mathrm{d} \bm{\theta} \\
          & = \epsilon \cdot c_{\mh, \bm{0}_d, k}(L) + 2 \sup_{\mz \in \mathbb{T}^{d}} f ( \mz ) \sum_{m=1}^{d} \int_{E_m} \bm{L}_{\mh}^k (\bm{\theta})  \mathrm{d} \bm{\theta}.
    \end{split} \end{align}
    We write $\bm{\theta}_{-m} =  \left( \theta_1, \dots, \theta_{m-1}, \theta_{m+1}, \dots \theta_d \right)\tran \in \mathbb{R}^{d-1}$. It holds that
    \begin{align} \begin{split} \label{eq:new7.42}
          &\int_{E_m} \bm{L}_{\mh}^k (\bm{\theta})  \mathrm{d} \bm{\theta} \\
          & = \left( \int_{[-\pi, \pi)^{d-1}} \left[ \prod_{\ell\neq m}^d L^k \left( \frac{1 - \cos \theta_\ell }{h_\ell^2} \right) \right]  \mathrm{d} \bm{\theta}_{-m} \right) \left( \int_{[-\pi, \pi) \setminus \left(-\frac{\delta_{\epsilon}}{\sqrt{d}}, \frac{\delta_{\epsilon}}{\sqrt{d}} \right)} L^k \left( \frac{1 - \cos \theta_m}{h_m^2} \right)  \mathrm{d} \theta_m \right) \\
          & = 2^{d} \left( \int_{[0, \pi)^{d-1}} \left[ \prod_{\ell\neq m}^d L^k \left( \frac{1 - \cos \theta_\ell }{h_\ell^2} \right) \right]  \mathrm{d} \bm{\theta}_{-m} \right) \left( \int_{ \frac{\delta_{\epsilon}}{\sqrt{d}}}^{\pi} L^k \left( \frac{1 - \cos \theta_m}{h_m^2} \right)  \mathrm{d} \theta_m \right).
    \end{split} \end{align}
    By a modification of \Cref{lemma:1.4}, we get
    \begin{align} \label{eq:new7.43}
        2^{d-1} \int_{[0, \pi)^{d-1}} \left[ \prod_{\ell\neq m}^d L^k \left( \frac{1 - \cos \theta_\ell }{h_\ell^2} \right) \right] \mathrm{d} \bm{\theta}_{-m} = O \left( \prod_{1 \leq \ell \neq m \leq d} h_\ell \right).
    \end{align}
    Also, Lemma B.2 of Im et al. (2025) implies that
    \begin{align} \label{eq:new7.44}
        2 \int_{\frac{\delta_{\epsilon}}{\sqrt{d}}}^{\pi} L^k \left( \frac{1 - \cos \theta_m}{h_m^2} \right) \mathrm{d} \theta_m = o \left( h_m \right). 
    \end{align}
    Combining \eqref{eq:new7.42}, \eqref{eq:new7.43} and \eqref{eq:new7.44}, we get
    \begin{align} \label{eq:new7.45}
        \int_{E_m} \bm{L}_{\mh}^k (\bm{\theta})  \mathrm{d} \bm{\theta} = o \left( \rho ( \mh ) \right).
    \end{align}
    Combining \eqref{eq:new7.41}, \eqref{eq:new7.45} and \Cref{lemma:1.4}, we also get
    \begin{align} \label{eq:new7.46}
         \left| \mathbb{E}\left[ \mathcal{L}_{\mx, \mh}^k \left( \mathbf{X} \right) \right] -  c_{\mh, \bm{0}_d, k}(L) f (\mx) \right| & \leq  \epsilon \cdot O \left(\rho(\mh) \right) + 2 \sup_{\mz \in \mathbb{T}^{d}} f ( \mz )  \sum_{m=1}^{d}  o \left( \rho ( \mh ) \right).
    \end{align}
    Since \eqref{eq:new7.46} holds for any $\epsilon \in (0, \infty)$, we get the desired result.
\end{proof}

\begin{lemma} \label{lemma:new3.2}
    Assume that the conditions \ref{con:L1}, \ref{con:D2} and \ref{con:D3} hold and that $\lim_{n\rightarrow\infty} \| \mh \|_2=0$. Then, for any $k\in\{1,2\}$, it holds that
    \begin{align*}
        \sup_{y \in \mathbb{M}} \Big| \mathbb{E}\left[ \mathcal{L}_{\mx, \mh}^k \left( \mathbf{X} \right) g_y \left( \mX \right) \right] -  c_{\mh, \bm{0}_d, k}(L) ( f \cdot g_y ) (\mx)  \Big| = o\left(\rho(\mh)\right).
    \end{align*}
\end{lemma}

\begin{proof} [Proof of Lemma \ref{lemma:new3.2}]
    The lemma follows by arguing as in the proof of \Cref{lemma:new2.2}.
\end{proof}

\begin{lemma} \label{lemma:new3.4}
    Assume that the conditions \ref{con:L1} and \ref{con:D1}--\ref{con:D3} hold and that $\lim_{n\rightarrow\infty} \| \mh \|_2=0$. Then, it holds that
    \begin{align*}
        \sup_{y \in \mathbb{M}} \big|  \tilde{M}_{\mh, 0}(\mathbf{x}, y) - M_{\oplus}(\mathbf{x}, y)   \big| = o(1).
    \end{align*}
\end{lemma}

\begin{proof} [Proof of \Cref{lemma:new3.4}]
Combining \Cref{lemma:new2.2} and \Cref{lemma:new3.2}, we get
\begin{align} \begin{split} \label{eq:new7.54}
    \sup_{y \in \mathbb{M}} & \Big| \mathbb{E}\left[ \mathcal{L}_{\mx, \mh} \left( \mathbf{X} \right) g_y \left( \mX \right) \right] - g_y (\mathbf{x}) \mathbb{E}\left[ \mathcal{L}_{\mx, \mh} \left( \mathbf{X} \right)  \right] \Big| = o\left(\rho(\mh)\right).
\end{split} \end{align}
\Cref{lemma:1.4} and \Cref{lemma:new2.2} imply that
\begin{align} \label{eq:new7.55}
\begin{split}
    \left( c_{\mh, \bm{0}_d, 1}(L) f (\mathbf{x}) \right)^{-1} & = O\left(\rho^{-1} \left( \mh \right)\right), \\
    \quad c_{\mh, \bm{0}_d, 1}(L) f (\mathbf{x}) \left( \mathbb{E}\left[ \mathcal{L}_{\mx, \mh} \left( \mathbf{X} \right)\right] \right)^{-1} & = O(1), \\
    \left( \mathbb{E}\left[ \mathcal{L}_{\mx, \mh} \left( \mathbf{X} \right)\right] \right)^{-1} & = O\left(\rho^{-1} \left( \mh \right)\right).
\end{split}
\end{align}
Combining \eqref{eq:new7.54} and \eqref{eq:new7.55}, we get
    \begin{align} \label{lemma:new3.3}
        \sup_{y \in \mathbb{M}} \left| \frac{\mathbb{E}\left[ \mathcal{L}_{\mx, \mh} \left( \mathbf{X} \right) g_y (\mathbf{X}) \right]}{\mathbb{E}\left[ \mathcal{L}_{\mx, \mh} \left( \mathbf{X} \right) \right]} - g_y(\mathbf{x}) \right| = o(1).
    \end{align}
Note that
    \begin{align} \label{eq:3.21}
    \begin{split}
        M_{\oplus}(\mathbf{x}, y) = & \int_{\mathbb{M}} d_{\mathbb{M}}^2(y, w) \mathrm{d}P_{Y|\mathbf{X}=\mathbf{x}}(w) =  \int_{\mathbb{M}} d_{\mathbb{M}}^2(y, w) g_w(\mathbf{x}) \mathrm{d} P_Y (w),  \\
        \tilde{M}_{\mh, 0}(\mathbf{x}, y) = & \frac{\mathbb{E} \left[ \mathcal{L}_{\mx, \mh} \left( \mathbf{X} \right) M_{\oplus}(\mathbf{X}, y) \right] }{\mathbb{E} \left[ \mathcal{L}_{\mx, \mh} \left( \mathbf{X} \right)  \right] } = \int_{\mathbb{M}} d_{\mathbb{M}}^2(y, w) \frac{\mathbb{E}\left[ \mathcal{L}_{\mx, \mh} \left( \mathbf{X} \right) g_w (\mathbf{X}) \right]}{\mathbb{E}\left[ \mathcal{L}_{\mx, \mh} \left( \mathbf{X} \right) \right]} \mathrm{d} P_Y (w),
    \end{split}
    \end{align}
    where the last equality in \eqref{eq:3.21} follows from Fubini's theorem. Since $d_{\mathbb{M}}^2(y, w) \leq \mathrm{diam}(\mathbb{M})^2 < \infty$, the lemma follows from \eqref{lemma:new3.3} and \eqref{eq:3.21}.
\end{proof}

\begin{lemma} \label{thm:3.1}
   Assume that the conditions \ref{con:L1}, \ref{con:D1}--\ref{con:D3} and \ref{con:M1} hold and that $\lim_{n\rightarrow\infty} \| \mh \|_2=0$. Then, it holds that
    \begin{align*}
        d_{\mathbb{M}}(\tilde{m}_{\mh, 0}(\mathbf{x}), m_{\oplus}(\mathbf{x})) = o(1).
    \end{align*}

\end{lemma}

\begin{proof} [Proof of \Cref{thm:3.1}]
    The lemma follows by arguing as in the proof of Lemma B.7 of Im et al. (2025) and using \Cref{lemma:new3.4}.
\end{proof}

\begin{lemma} \label{lemma:3.6}
    Assume that the conditions \ref{con:L1}, \ref{con:B1}, \ref{con:D1} and \ref{con:D2} hold. Then, it holds that
    \begin{align*}
        \sup_{y \in \mathbb{M}} \big| \hat{M}_{\mh, 0}(\mathbf{x}, y) - \tilde{M}_{\mh, 0}(\mathbf{x}, y) \big| = o_{\mathbb{P}}(1)
    \end{align*}
\end{lemma}

\begin{proof} [Proof of \Cref{lemma:3.6}]
    Note that
    \begin{align} \label{eq:3.31}
    \begin{split}
         & \big| \hat{M}_{\mh, 0}(\mathbf{x}, y) - \tilde{M}_{\mh, 0}(\mathbf{x}, y) \big| \\
         & \leq \dfrac{\mathrm{diam} (\mathbb{M})^2}{ \mathbb{E} \left[\mathcal{L}_{\mx, \mh} \left( \mathbf{X} \right)\right] } \cdot \bigg| n^{-1} \sum_{i=1}^{n} \mathcal{L}_{\mx, \mh} \left( \mathbf{X}^{(i)} \right) - \mathbb{E} \left[ \mathcal{L}_{\mx, \mh} \left( \mathbf{X} \right)  \right] \bigg| \\
         & \quad +  \dfrac{1}{ \mathbb{E} \left[\mathcal{L}_{\mx, \mh} \left( \mathbf{X} \right)\right] } \cdot \bigg| n^{-1} \sum_{i=1}^{n} \mathcal{L}_{\mx, \mh} \left( \mathbf{X}^{(i)} \right) d_{\mathbb{M}}^2(Y^{(i)}, y) - \mathbb{E} \left[\mathcal{L}_{\mx, \mh} \left( \mathbf{X} \right) d_{\mathbb{M}}^2(Y, y) \right] \bigg|.
    \end{split}
    \end{align}
    \Cref{lemma:1.4} and \Cref{lemma:new2.2} imply that
    \begin{align} \label{eq:3.35}
    \begin{split}
        & n^{-1} \sum_{i=1}^{n} \mathcal{L}_{\mx, \mh} \left( \mathbf{X}^{(i)} \right) - \mathbb{E} \left[ \mathcal{L}_{\mx, \mh} \left( \mathbf{X} \right)  \right] \\
        & = O_{\mathbb{P}} \left( \sqrt{ n^{-1} \mathbb{V}\mathrm{ar}\left[  \mathcal{L}_{\mx, \mh} \left( \mathbf{X} \right) \right] } \right) \\
        & = O_{\mathbb{P}}\left(n^{-\frac{1}{2}} \rho^{\frac{1}{2}} \left( \mh \right)\right),  \\
        & n^{-1} \sum_{i=1}^{n} \mathcal{L}_{\mx, \mh} \left( \mathbf{X}^{(i)} \right) d_{\mathbb{M}}^2(Y^{(i)}, y) - \mathbb{E} \left[\mathcal{L}_{\mx, \mh} \left( \mathbf{X} \right) d_{\mathbb{M}}^2(Y, y) \right] \\
        & = O_{\mathbb{P}} \left( \sqrt{ n^{-1} \mathbb{V}\mathrm{ar}\left[  \mathcal{L}_{\mx, \mh} \left( \mathbf{X} \right) d_{\mathbb{M}}^2(Y, y) \right] } \right) \\
        & = O_{\mathbb{P}}\left(n^{-\frac{1}{2}} \rho^{\frac{1}{2}} \left( \mh \right)\right).
    \end{split}
    \end{align}
    Combining \eqref{eq:new7.55}, \eqref{eq:3.31} and \eqref{eq:3.35}, we get
    \begin{align*}
        \big| \hat{M}_{\mh, 0}(\mathbf{x}, y) - \tilde{M}_{\mh, 0}(\mathbf{x}, y) \big| = O_{\mathbb{P}}\left(n^{-\frac{1}{2}} \rho^{-\frac{1}{2}} \left( \mh \right)\right).
    \end{align*}
    The rest of the proof proceeds as in the proof of Lemma B.8 of Im et al. (2025).
\end{proof}

\begin{lemma} \label{thm:3.3}
    Assume that the conditions \ref{con:L1}, \ref{con:B1}, \ref{con:D1}, \ref{con:D2} and \ref{con:M1} hold. Then, it holds that
    \begin{align*}
        d_{\mathbb{M}}(\hat{m}_{\mh, 0}(\mathbf{x}), \tilde{m}_{\mh, 0}(\mathbf{x})) = o_{\mathbb{P}}(1).
    \end{align*}
\end{lemma}

\begin{proof} [Proof of \Cref{thm:3.3}]
   The lemma follows by arguing as in the proof of Lemma B.9 of Im et al. (2025).
\end{proof}

\subsection{Proof of \Cref{thm:consistency} for  $\hat{m}_{\mh, 0}(\mathbf{x})$}

The theorem follows from \Cref{thm:3.1} and \Cref{thm:3.3}.

\subsection{Lemmas for proof of \Cref{thm:main4.3}}

Let $g:\mathbb{T}^d\rightarrow\mathbb{R}$ be a sufficiently smooth function and $\bar{g} : \mathbb{R}^{2d} \setminus\mathcal{N} \to \mathbb{R}$ be its homogeneous extension, as defined in \eqref{homo ext}. Differentiating once more in \eqref{homo eq} with respect to $t_k$ gives
\[
\frac{\partial^2}{\partial t_\ell \partial t_k}
\bar g(\mathbf{D}_{\mz}\mt)
=
(\mathbf{D}_{\mz}\me_\ell)\tran
\nabla^2 \bar g(\mathbf{D}_{\mz}\mt)
\mathbf{D}_{\mz}\me_k=0,
\quad k,\ell = 1,\dots,d,
\]
where $\nabla^2 \bar g(\mathbf{D}_{\mz}\mt)$ denotes the Hessian of $\bar{g}$ at $\mathbf{D}_{\mz}\mt$. Let $\mathbf{O}_p$ denote the $p \times p$ zero matrix. Collecting these terms into matrix form yields
\[
\mathbf{D}_{\mz}\tran \nabla^2 \bar g(\mathbf{D}_{\mz}\mt)\, \mathbf{D}_{\mz}
=
\mathbf{O}_d.
\]
By substituting $\mt = (1, \dots, 1) \tran \in (0, \infty)^d$, we get
\begin{align} \label{constraint2}
\mathbf{D}_{\mz} \tran\nabla^2 \bar{g}(\mz)\mathbf{D}_{\mz}=\mathbf{O}_{d}.
\end{align}

From \eqref{constraint} and \eqref{constraint2}, we get
\begin{align} \begin{split} \label{eq:2.2}
    \mathbf{D}_{\mz} \tran\nabla \bar{f}(\mz)=\bm{0}_d 
    &, \quad \mathbf{D}_{\mz} \tran\nabla^2 \bar{f}(\mz)\mathbf{D}_{\mz}=\mathbf{O}_{d},\\
    \mathbf{D}_{\mz} \tran\nabla \left( \bar{f} \cdot \bar{g}_y \right) (\mz)=\bm{0}_d &, \quad \mathbf{D}_{\mz} \tran\nabla^2 \left( \bar{f} \cdot \bar{g}_y \right) (\mz)\mathbf{D}_{\mz}=\mathbf{O}_{d}, 
\end{split} \end{align}
where $\nabla^2 \bar{f}$ and $\nabla^2 \left( \bar{f} \cdot \bar{g}_y \right)$ denote the Hessians of $\bar{f}$ and $\bar{f} \cdot \bar{g}_y$, respectively. 
\Cref{lemma:2.1} below provides a Taylor-like expansion for the density function $f$ of $\mX$.

\begin{lemma} \label{lemma:2.1}
    Assume that $\bar{f}$ is differentiable on $\mathbb{R}^{2d} \setminus\mathcal{N}$. Then, for any $\bm{\theta} \in [-\pi, \pi)^{d}$, it holds that
    \begin{align} \label{taylororder1}
        \left| f \left( \bm{\Phi}_{\mx} \left( \bm{\theta} \right) \right) - f \left( \mx \right) \right| \leq  \| \bm{\theta} \|_2 \sup_{t \in [0, 1]} \left\| \nabla \bar{f}  \left( \bm{\Phi}_{\mx} (t\bm{\theta}) \right) \right\|_2.
    \end{align}
    If $\bar{f}$ is twice differentiable on $\mathbb{R}^{2d} \setminus\mathcal{N}$, then it holds that
    \begin{align} \label{taylororder2}
        \left| f \left( \bm{\Phi}_{\mx} \left( \bm{\theta} \right) \right) - f \left( \mx \right) - \left( \mathbf{R}_{\mx} \bm{\theta} \right) \tran\nabla \bar{f}(\mx) \right| \leq  \frac{1}{2} \| \bm{\theta} \|_2^2 \sup_{t \in [0, 1]} \left\| \nabla^2 \bar{f}  \left( \bm{\Phi}_{\mx} (t\bm{\theta}) \right) \right\|_2. 
    \end{align}
\end{lemma}

\begin{proof} [Proof of Lemma \ref{lemma:2.1}]
    We first prove \eqref{taylororder1}. Note that the Jacobian matrix of $\bm{\Phi}_{\mx} : [-\pi, \pi)^{d} \to \mathbb{T}^{d}$ is evaluated as
    \begin{align} \label{Jacobian}
        \mathbf{J}_{\bm{\Phi}_{\mx}} (\bm{\theta}) = \begin{bmatrix}
            - \mx_1 \sin \theta_1 + \left( \mathbf{R} \mx_1 \right) \cos \theta_1 & \cdots & \bm{0}_2 \\
            \vdots & \ddots & \vdots \\
            \bm{0}_2 & \cdots & - \mx_d \sin \theta_d + \left( \mathbf{R} \mx_d \right) \cos \theta_d \end{bmatrix} = \mathbf{R}_{\bm{\Phi}_{\mx} (\bm{\theta})}.
    \end{align}
    By \eqref{Jacobian}, the gradient $\nabla \left( \bar{f} \circ \bm{\Phi}_{\mx} \right)$ of $\bar{f} \circ \bm{\Phi}_{\mx}$ is evaluated as
    \begin{align} \label{Gradient}
        \nabla \left( \bar{f} \circ \bm{\Phi}_{\mx} \right) (\bm{\theta}) = \mathbf{J}_{\bm{\Phi}_{\mx}} (\bm{\theta}) \tran \nabla \bar{f} \left( \bm{\Phi}_{\mx} \left( \bm{\theta} \right) \right) = \mathbf{R}_{\bm{\Phi}_{\mx} (\bm{\theta})}\tran \nabla \bar{f} \left( \bm{\Phi}_{\mx} \left( \bm{\theta} \right) \right).
    \end{align}
    Taylor's remainder theorem and \eqref{Gradient} imply that there exists $t_{\bm{\theta}}^{\ast} \in [0, 1]$ such that
    \begin{align} \begin{split} \label{taylororder1_Step1}
        f \left( \bm{\Phi}_{\mx} \left( \bm{\theta} \right) \right) - f \left( \mx \right) = \left( \mathbf{R}_{\bm{\Phi}_{\mx} (t_{\bm{\theta}}^{\ast}\bm{\theta})} \bm{\theta} \right) \tran \nabla \bar{f}  \left( \bm{\Phi}_{\mx} (t_{\bm{\theta}}^{\ast}\bm{\theta}) \right).
    \end{split} \end{align}
    Note that $\left\| \mathbf{R}_{\bm{\Phi}_{\mx} (t_{\bm{\theta}}^{\ast}\bm{\theta})} \bm{\theta}\right\|_2 = \left\| \bm{\theta} \right\|_2$. This and \eqref{taylororder1_Step1} imply that
    \begin{align*} 
        \left| f \left( \bm{\Phi}_{\mx} \left( \bm{\theta} \right) \right) - f \left( \mx \right)  \right| & \leq \left\| \mathbf{R}_{\bm{\Phi}_{\mx} (t_{\bm{\theta}}^{\ast}\bm{\theta})} \bm{\theta}\right\|_2 \left\| \nabla \bar{f}  \left( \bm{\Phi}_{\mx} (t_{\bm{\theta}}^{\ast}\bm{\theta}) \right) \right\|_2 \leq \left\|\bm{\theta}\right\|_{2} \sup_{t \in [0, 1]} \left\| \nabla \bar{f}  \left( \bm{\Phi}_{\mx} (t\bm{\theta}) \right) \right\|_2,
    \end{align*}
    which gives \eqref{taylororder1}.
    
    Now, we prove \eqref{taylororder2}. Since $\bm{\Phi}_{\mx} \left( \bm{0}_d \right) = \mx$, \eqref{Gradient} implies that
    \begin{align*}
        \nabla \left( \bar{f} \circ \bm{\Phi}_{\mx} \right) \left(\bm{0}_d \right) = \mathbf{R}_{\bm{\Phi}_{\mx} (\bm{0}_d)}\tran \nabla \bar{f} \left( \bm{\Phi}_{\mx} \left( \bm{0}_d \right) \right) = \mathbf{R}_{\mx}\tran \nabla \bar{f} \left( \mx \right).
    \end{align*} 
    By the second-order chain rule for vector-valued compositions, the Hessian of $f\circ\bm{\Phi}_{\mx}$ can be written as
\begin{align}
\nabla^{2}(\bar f \circ \bm{\Phi}_{\mx})(\bm{\theta}) =
\mathbf{J}_{\bm{\Phi}_{\mx}}(\bm{\theta})\tran
\nabla^{2}\bar f(\bm{\Phi}_{\mx}(\bm{\theta}))
\mathbf{J}_{\bm{\Phi}_{\mx}}(\bm{\theta}) +
\sum_{k=1}^{2d}
\frac{\partial \bar f}{\partial z_k}(\bm{\Phi}_{\mx}(\bm{\theta}))
\nabla^{2} (\bm{\Phi}_{\mx})_k(\bm{\theta}).
\label{eq:hessian_chain_rule}
\end{align}
We show that the second term in \eqref{eq:hessian_chain_rule} vanishes. Note that
\[
\frac{\partial^2}{\partial \theta_\ell^2}
\bm{\Phi}_{\mx,\ell}(\theta_\ell)
=
-\,\bm{\Phi}_{\mx,\ell}(\theta_\ell),
\quad
\frac{\partial^2}{\partial \theta_\ell\partial \theta_k}
\bm{\Phi}_{\mx}(\bm{\theta})
=
\bm{0}_{2d},\quad \ell\neq k.
\]
Hence, all nonzero second derivatives of $\bm{\Phi}_{\mx}$
lie in the radial directions
$\bm{\Phi}_{\mx,\ell}(\theta_\ell)$. On the other hand, \eqref{eq:2.2} implies that
\[
\mz_\ell^{\tran}\frac{\partial \bar f(\mz)}{\partial \mz_\ell}=0,
\quad \ell=1,\dots,d.
\]
Evaluating at $\mz=\bm{\Phi}_{\mx}(\bm{\theta})$,
each term in the summation of
\eqref{eq:hessian_chain_rule} reduces to a blockwise inner product
between $\bm{\Phi}_{\mx,\ell}(\theta_\ell)$ and
$\partial \bar f/\partial \mz_\ell$,
which vanishes by the above constraint.
Hence, the second term in
\eqref{eq:hessian_chain_rule} is identically zero. Consequently,
\[
\nabla^{2}(\bar f \circ \bm{\Phi}_{\mx})(\bm{\theta})
=
\mathbf{R}_{\bm{\Phi}_{\mx}(\bm{\theta})}\tran
\nabla^{2}\bar f(\bm{\Phi}_{\mx}(\bm{\theta}))
\mathbf{R}_{\bm{\Phi}_{\mx}(\bm{\theta})}.
\]
By Taylor's remainder theorem, there exists $t_{\bm{\theta}}^{\ast \ast} \in [0, 1]$ such that
\begin{align} \begin{split} \label{Taylorexpansion_Step1}
        f \left( \bm{\Phi}_{\mx} \left( \bm{\theta} \right) \right) - f \left( \mx \right) - \left( \mathbf{R}_{\mx} \bm{\theta} \right) \tran\nabla \bar{f}(\mx) = \frac{1}{2} \bm{\theta} \tran \mathbf{R}_{\bm{\Phi}_{\mx} (t_{\bm{\theta}}^{\ast \ast}\bm{\theta})} \tran \nabla^2 \bar{f}  \left( \bm{\Phi}_{\mx} (t_{\bm{\theta}}^{\ast \ast}\bm{\theta}) \right) \mathbf{R}_{\bm{\Phi}_{\mx} (t_{\bm{\theta}}^{\ast \ast}\bm{\theta})} \bm{\theta}.
    \end{split} \end{align}
Note that $\left\| \mathbf{R}_{\bm{\Phi}_{\mx} (t_{\bm{\theta}}^{\ast \ast}\bm{\theta})} \bm{\theta}\right\|_2 = \left\| \bm{\theta} \right\|_2$. This and \eqref{Taylorexpansion_Step1} imply that
    \begin{align*}
        | f \left( \bm{\Phi}_{\mx} \left( \bm{\theta} \right) \right) - f \left( \mx \right) - \left( \mathbf{R}_{\mx} \bm{\theta} \right) \tran\nabla \bar{f}(\mx) | & \leq \frac{1}{2} \left\| \mathbf{R}_{\bm{\Phi}_{\mx} (t_{\bm{\theta}}^{\ast \ast}\bm{\theta})} \bm{\theta}\right\|_2^2 \left\| \nabla^2 \bar{f}  \left( \bm{\Phi}_{\mx} (t_{\bm{\theta}}^{\ast \ast}\bm{\theta}) \right) \right\|_2  \\
        & \leq \frac{1}{2} \left\|\bm{\theta}\right\|_{2}^{2} \sup_{t \in [0, 1]} \left\| \nabla^2 \bar{f}  \left( \bm{\Phi}_{\mx} (t\bm{\theta}) \right) \right\|_2.
    \end{align*}
This completes the proof.
\end{proof}

\begin{lemma} \label{lemma:3.1}
    Assume that $\bar{f}$ and $\bar{g}_y$ are differentiable on $\mathbb{R}^{2d} \setminus\mathcal{N}$. Then, for any $\bm{\theta} \in [-\pi, \pi)^{d}$, it holds that
    \begin{align*} 
        \left| \left( f \cdot g_y \right) \left( \bm{\Phi}_{\mx} \left( \bm{\theta} \right) \right) - \left( f \cdot g_y \right) \left( \mx \right) \right| \leq  \| \bm{\theta} \|_2 \sup_{t \in [0, 1]} \left\| \nabla \left( \bar{f} \cdot \bar{g}_y \right) \left( \bm{\Phi}_{\mx} (t\bm{\theta}) \right) \right\|_2.
    \end{align*}
    If $\bar{f}$ and $\bar{g}_y$ are twice differentiable on $\mathbb{R}^{2d} \setminus\mathcal{N}$, then it holds that
    \begin{align*}
        \left| \left( f \cdot g_y \right) \left( \bm{\Phi}_{\mx} \left( \bm{\theta} \right) \right) - \left( f \cdot g_y \right) \left( \mx \right) - \left( \mathbf{R}_{\mx} \bm{\theta} \right) \tran \nabla  \left( \bar{f} \cdot \bar{g}_y \right) \left( \mx \right)  \right| \nonumber  \leq  \frac{1}{2} \| \bm{\theta} \|_2^2 \sup_{t \in [0, 1]} \left\| \nabla^2 \left( \bar{f} \cdot \bar{g}_y \right)  \left( \bm{\Phi}_{\mx} (t\bm{\theta}) \right) \right\|_2.
    \end{align*}
\end{lemma}

\begin{proof} [Proof of Lemma \ref{lemma:3.1}]
    The lemma follows by arguing as in the proof of \Cref{lemma:2.1}.
\end{proof}

\begin{lemma} \label{lemma:2.2}
    Assume that the conditions \ref{con:L1} and \ref{con:D4} hold and that $\lim_{n\rightarrow\infty} \| \mh \|_2=0$. Then, for any $k\in\{1,2\}$, it holds that
    \begin{align*}
        \mathbb{E}\left[ \mathcal{L}_{\mx, \mh}^k \left( \mathbf{X} \right) \right] -  c_{\mh, \bm{0}_d, k}(L) f (\mx) = O\left(\rho(\mh)\| \mh \|_{2}^{2}\right).
    \end{align*}
\end{lemma}

\begin{proof} [Proof of Lemma \ref{lemma:2.2}]
    For any $\mv = \left( v_1, \dots, v_d \right) \tran \in \mathbb{R}^{d}$, \eqref{symmetricproperty} implies that
    \begin{align} \begin{split} \label{eq:2.10}
        \int_{[-\pi, \pi)^{d}}  \bm{L}_{\mh}^k (\bm{\theta}) \bm{\theta} \tran \mv \mathrm{d} \bm{\theta}  = \sum_{\ell=1}^{d} v_{\ell} \int_{[-\pi, \pi)^{d}}  \bm{L}_{\mh}^k (\bm{\theta}) \theta_{\ell}  \mathrm{d} \bm{\theta} = 0.
    \end{split} \end{align}
    Combining \Cref{lemma:1.4} and \eqref{symmetricproperty}, we get
    \begin{align} \begin{split} \label{eq:2.11}
        \int_{[-\pi, \pi)^{d}} \bm{L}_{\mh}^k (\bm{\theta})   \| \bm{\theta} \|_2^2  \mathrm{d} \bm{\theta} & = \sum_{\ell=1}^{d} \int_{[-\pi, \pi)^{d}} \bm{L}_{\mh}^k (\bm{\theta})   \theta_\ell^2 \mathrm{d} \bm{\theta} = \sum_{\ell=1}^{d} c_{\mh, 2\mathbf{e}_{\ell}, k}(L)  = O\left( \rho(\mh) \| \mh \|_{2}^{2} \right).
    \end{split} \end{align}
    Combining \eqref{eq:new7.35}, \eqref{eq:2.10}, \eqref{eq:2.11} and \Cref{lemma:2.1}, we also get
    \begin{align*}
          &\left| \mathbb{E}\left[ \mathcal{L}_{\mx, \mh}^k \left( \mathbf{X} \right) \right] -  c_{\mh, \bm{0}_d, k}(L) f (\mx) \right| \\
          &= \left| \mathbb{E}\left[ \mathcal{L}_{\mx, \mh}^k \left( \mathbf{X} \right) \right] -  c_{\mh, \bm{0}_d, k}(L) f (\mx) - \int_{[-\pi, \pi)^{d}} \bm{L}_{\mh}^k (\bm{\theta})  \bm{\theta} \tran  \mathbf{R}_{\mx} \tran \nabla \bar{f}(\mx)  \mathrm{d} \bm{\theta} \right| \\
          &= \left| \int_{[-\pi, \pi)^{d}} \bm{L}_{\mh}^k (\bm{\theta})  \Big( f \left( \bm{\Phi}_{\mx} \left( \bm{\theta} \right) \right) - f(\mx) - \left( \mathbf{R}_{\mx} \bm{\theta} \right) \tran\nabla \bar{f}(\mx)\Big)  \mathrm{d} \bm{\theta} \right| \\
          &\leq \frac{1}{2} \sup_{\mz \in \mathbb{T}^{d}} \left\| \nabla^2 \bar{f}  \left( \mz \right) \right\|_2 \int_{[-\pi, \pi)^{d}} \bm{L}_{\mh}^k (\bm{\theta})   \| \bm{\theta} \|_2^2  \mathrm{d} \bm{\theta} \\
          &= \frac{1}{2} \sup_{\mz \in \mathbb{T}^{d}} \left\| \nabla^2 \bar{f}  \left( \mz \right) \right\|_2  O\left( \rho(\mh) \| \mh \|_{2}^{2} \right).
    \end{align*}
    This completes the proof.
\end{proof}

\begin{lemma} \label{lemma:3.2}
    Assume that the conditions \ref{con:L1}, \ref{con:D4} and \ref{con:D5} hold and that $\lim_{n\rightarrow\infty} \| \mh \|_2=0$. Then, for any $k\in\{1,2\}$, it holds that
    \begin{align*}
        \sup_{y \in \mathbb{M}} \Big| \mathbb{E}\left[ \mathcal{L}_{\mx, \mh}^k \left( \mathbf{X} \right) g_y \left( \mX \right) \right] -  c_{\mh, \bm{0}_d, k}(L) ( f \cdot g_y ) (\mx)  \Big| = O\left(\rho(\mh)\| \mh \|_{2}^{2}\right).
    \end{align*}
\end{lemma}

\begin{proof} [Proof of Lemma \ref{lemma:3.2}]
    The lemma follows by arguing as in the proof of \Cref{lemma:2.2} and using \Cref{lemma:3.1}.
\end{proof}

\begin{lemma} \label{thm:3.2}
    Assume that the conditions \ref{con:L1}, \ref{con:D1}, \ref{con:D4}, \ref{con:D5}, \ref{con:M1} and \ref{con:M2} hold and that $\lim_{n\rightarrow\infty} \| \mh \|_2=0$. Then, it holds that
    \begin{align*}
        d_{\mathbb{M}}(\tilde{m}_{\mh, 0}(\mathbf{x}), m_{\oplus}(\mathbf{x}))^{\beta_{\oplus} - 1} = O\left(\| \mh \|_{2}^{2}\right).
    \end{align*}
\end{lemma}

\begin{proof} [Proof of \Cref{thm:3.2}]
    Combining \eqref{eq:3.21} and the fact that $\tilde{M}_{\mh, 0} (\mathbf{x}, m_{\oplus}(\mathbf{x}) ) \ge \tilde{M}_{\mh, 0} (\mathbf{x}, \tilde{m}_{\mh, 0}(\mathbf{x}))$, we get
    \begin{align} \begin{split} \label{eq:3.27}
        &M_{\oplus} \left( \mx, \tilde{m}_{\mh, 0}(\mx) \right) - M_{\oplus} \left( \mx, m_{\oplus}(\mx) \right) \\ 
        &\leq \left( M_{\oplus} \left( \mx, \tilde{m}_{\mh, 0}(\mx) \right) - M_{\oplus} \left( \mx, m_{\oplus}(\mx) \right) \right)  + \left( \tilde{M}_{\mh, 0} (\mx, m_{\oplus}(\mx) ) - \tilde{M}_{\mh, 0} (\mx, \tilde{m}_{\mh, 0}(\mx)) \right) \\
        &=\int_{\mathbb{M}} \left( d_{\mathbb{M}}^2(m_{\oplus}(\mx), w) - d_{\mathbb{M}}^2(\tilde{m}_{\mh, 0}(\mx), w) \right) \left( \frac{\mathbb{E}\left[ \mathcal{L}_{\mx, \mh} \left( \mathbf{X} \right) g_w (\mathbf{X}) \right]}{\mathbb{E}\left[ \mathcal{L}_{\mx, \mh} \left( \mathbf{X} \right) \right]} - g_w (\mx) \right) \mathrm{d} P_Y (w) \\
        &\leq \int_{\mathbb{M}} \left| d_{\mathbb{M}}^2(m_{\oplus}(\mx), w) - d_{\mathbb{M}}^2(\tilde{m}_{\mh, 0}(\mx), w) \right| \left| \frac{\mathbb{E}\left[ \mathcal{L}_{\mx, \mh} \left( \mathbf{X} \right) g_w (\mathbf{X}) \right]}{\mathbb{E}\left[ \mathcal{L}_{\mx, \mh} \left( \mathbf{X} \right) \right]} - g_w (\mx) \right| \mathrm{d} P_Y (w) \\
        & \leq 2 \mathrm{diam} (\mathbb{M})  d_{\mathbb{M}}(\tilde{m}_{\mh, 0}(\mathbf{x}), m_{\oplus}(\mathbf{x})) \sup_{y \in \mathbb{M}} \left| \frac{\mathbb{E}\left[ \mathcal{L}_{\mx, \mh} \left( \mathbf{X} \right) g_y (\mathbf{X}) \right]}{\mathbb{E}\left[ \mathcal{L}_{\mx, \mh} \left( \mathbf{X} \right) \right]} - g_y (\mathbf{x}) \right|.
    \end{split} \end{align}
    \Cref{lemma:2.2} and \Cref{lemma:3.2} imply that
    \begin{align} \label{tau0bound}
        \sup_{y \in \mathbb{M}} \big| \mathbb{E}\left[\mathcal{L}_{\mx, \mh} \left( \mathbf{X} \right) g_y (\mathbf{X}) \right] - g_y (\mathbf{x}) \mathbb{E}\left[\mathcal{L}_{\mx, \mh} \left( \mathbf{X} \right)\right] \big| = O\left(\rho(\mh)\| \mh \|_{2}^{2}\right).
    \end{align}
    Combining \eqref{eq:new7.55} and \eqref{tau0bound}, we get
    \begin{align} \begin{split} \label{eq:3.28}
        \sup_{y \in \mathbb{M}} \left| \frac{\mathbb{E}\left[\mathcal{L}_{\mx, \mh} \left( \mathbf{X} \right) g_y (\mathbf{X}) \right]}{\mathbb{E}\left[\mathcal{L}_{\mx, \mh} \left( \mathbf{X} \right) \right]} - g_y (\mathbf{x}) \right| = O\left(\| \mh \|_{2}^{2}\right).
    \end{split} \end{align}
    By \Cref{thm:3.1}, there exists a constant $N\in\mathbb{N}$ such that $d_{\mathbb{M}}(\tilde{m}_{\mh, 0}(\mathbf{x}), m_{\oplus}(\mathbf{x})) < \eta_{\oplus}$ whenever $n\geq N$, where $\eta_{\oplus}>0$ is the constant defined in the condition \ref{con:M2}.
    By the condition \ref{con:M2}, we get
    \begin{align}
        C_{\oplus} \cdot d_{\mathbb{M}}(\tilde{m}_{\mh, 0}(\mathbf{x}), m_{\oplus}(\mathbf{x}))^{\beta_{\oplus}} \leq M_{\oplus} (\mathbf{x}, \tilde{m}_{\mh, 0}(\mathbf{x})) - M_{\oplus} (\mathbf{x}, m_{\oplus}(\mathbf{x})) \label{eq:3.30}
    \end{align}
    whenever $n\geq N$, where $C_{\oplus} >0$ is the constant defined in the condition \ref{con:M2}. Combining \eqref{eq:3.27}, \eqref{eq:3.28} and \eqref{eq:3.30}, we get the desired result.
\end{proof}

In the proof of \Cref{thm:3.4} below, for any metric space $(T,d_T)$ and constant $\delta>0$, let $N_{[ ]} (\delta, T, d_T )$ and $N (\delta, T, d_T )$ denote the $\delta$-bracketing number and $\delta$-covering number of $(T, d_T)$, respectively.

\begin{lemma} \label{thm:3.4}
    Assume that the conditions \ref{con:L1}, \ref{con:B1}, \ref{con:D1}--\ref{con:D3} and \ref {con:M1}--\ref{con:M3} hold. Then, it holds that
    \begin{align*}
        d_{\mathbb{M}}(\hat{m}_{\mh, 0}(\mathbf{x}), \tilde{m}_{\mh, 0}(\mathbf{x})) ^ {\beta_{\oplus}-\alpha_{\mathbb{M}}} = O_{\mathbb{P}} \left( n^{-\frac{1}{2}} \rho^{-\frac{1}{2}} \left( \mh \right) \right).
    \end{align*}
\end{lemma}

\begin{proof} [Proof of Theorem \ref{thm:3.4}]
     We define functions $\hat{T}_{\mh, 0} : \mathbb{T}^{d} \times \mathbb{M} \to \mathbb{R}$, $U_{\mh, 0} : \mathbb{T}^{d} \times \mathbb{M} \times \mathbb{M} \to \mathbb{R}$ and  $\hat{S}_{\mh,0} : \mathbb{T}^{d} \times \mathbb{M} \to \mathbb{R}$ as
    \begin{align*}
        \hat{T}_{\mh, 0}(\mx, y) &:= \hat{M}_{\mh, 0}(\mx, y) - \tilde{M}_{\mh, 0}(\mx, y), \\
        \quad U_{\mh, 0}(\mz, w , y ) &:=  D_{\mh, 0}(\mx, w , y ) \mathcal{L}_{\mx, \mh} \left( \mz \right), \\
        \hat{S}_{\mh,0}(\mx,y) &:= n^{-1} \sum_{i=1}^{n} U_{\mh, 0}(\mathbf{X}^{(i)}, Y^{(i)}, y) - \mathbb{E} \left[ U_{\mh, 0}(\mathbf{X}, Y, y) \right],
    \end{align*}
    where $D_{\mh, 0}(\mx, w , y ) := d_{\mathbb{M}}^2(w , y ) - d_{\mathbb{M}}^2(w ,  \tilde{m}_{\mh, 0}(\mx))$. Note that
    \begin{align} \begin{split} \label{eq:new3.38}
        & \mathbb{E} \left[ \mathcal{L}_{\mx, \mh} \left( \mathbf{X} \right)  \right] \left| \hat{T}_{\mh, 0}(\mathbf{x}, y) - \hat{T}_{\mh, 0}(\mathbf{x}, \tilde{m}_{\mh, 0}(\mathbf{x})) \right| \\
        & = \mathbb{E} \left[ \mathcal{L}_{\mx, \mh} \left( \mathbf{X} \right)  \right] \left| \frac{n^{-1} \sum_{i=1}^{n} U_{\mh, 0}(\mathbf{X}^{(i)}, Y^{(i)}, y) }{n^{-1} \sum_{i=1}^{n} \mathcal{L}_{\mx, \mh} \left( \mathbf{X}^{(i)} \right)} - \frac{\mathbb{E} \left[ U_{\mh, 0}(\mathbf{X}, Y, y) \right]}{\mathbb{E} \left[ \mathcal{L}_{\mx, \mh} \left( \mathbf{X} \right)  \right]} \right|  \\
        & \leq \mathbb{E} \left[ \mathcal{L}_{\mx, \mh} \left( \mathbf{X} \right)  \right] \left| \frac{n^{-1} \sum_{i=1}^{n} U_{\mh, 0}(\mathbf{X}^{(i)}, Y^{(i)}, y) }{n^{-1} \sum_{i=1}^{n} \mathcal{L}_{\mx, \mh} \left( \mathbf{X}^{(i)} \right)} - \frac{n^{-1} \sum_{i=1}^{n} U_{\mh, 0}(\mathbf{X}^{(i)}, Y^{(i)}, y)}{\mathbb{E} \left[ \mathcal{L}_{\mx, \mh} \left( \mathbf{X} \right)  \right]} \right|  \\
        & \quad + \bigg| n^{-1} \sum_{i=1}^{n} U_{\mh, 0}(\mathbf{X}^{(i)}, Y^{(i)}, y) - \mathbb{E} \left[ U_{\mh, 0}(\mathbf{X}, Y, y) \right] \bigg|  \\
        & = \bigg| n^{-1} \sum_{i=1}^{n} U_{\mh, 0}(\mathbf{X}^{(i)}, Y^{(i)}, y) \bigg| \cdot \left| \frac{\mathbb{E} \left[ \mathcal{L}_{\mx, \mh} \left( \mathbf{X} \right)  \right]}{n^{-1} \sum_{i=1}^{n} \mathcal{L}_{\mx, \mh} \left( \mathbf{X}^{(i)} \right)} - 1 \right| + \left| \hat{S}_{\mh,0}(\mx,y) \right| \\
        & \leq \frac{\sum_{i=1}^{n} \left| U_{\mh, 0}(\mathbf{X}^{(i)}, Y^{(i)}, y) \right|}{ \sum_{i=1}^{n} \mathcal{L}_{\mx, \mh} \left( \mathbf{X}^{(i)} \right)} \cdot \bigg| \mathbb{E} \left[ \mathcal{L}_{\mx, \mh} \left( \mathbf{X} \right)  \right] - n^{-1} \sum_{i=1}^{n} \mathcal{L}_{\mx, \mh} \left( \mathbf{X}^{(i)} \right) \bigg| + \left| \hat{S}_{\mh,0}(\mx,y) \right| \\
        &\leq 2 \mathrm{diam}(\mathbb{M}) \cdot d_{\mathbb{M}}(y, \tilde{m}_{\mh, 0}(\mathbf{x}))  \cdot \bigg| \mathbb{E} \left[ \mathcal{L}_{\mx, \mh} \left( \mathbf{X} \right)  \right] - n^{-1} \sum_{i=1}^{n} \mathcal{L}_{\mx, \mh} \left( \mathbf{X}^{(i)} \right) \bigg| + \left| \hat{S}_{\mh,0}(\mx,y) \right|,
    \end{split} \end{align}
    where the last inequality follows from that
    \begin{align*}
        \left| D_{\mh, 0}(\mx, w , y ) \right| &\leq 2 \mathrm{diam}(\mathbb{M}) d_{\mathbb{M}}(y, \tilde{m}_{\mh, 0}(\mathbf{x})), \\
        \left| U_{\mh, 0}(\mz, w, y) \right| 
        &\leq  2 \mathrm{diam}(\mathbb{M}) d_{\mathbb{M}}(y, \tilde{m}_{\mh, 0}(\mathbf{x})) \mathcal{L}_{\mx, \mh} \left( \mz \right).
    \end{align*}
    
    Regarding the first term in \eqref{eq:new3.38}, we get
    \begin{align*}
         \mathbb{E} \left[ \bigg| n^{-1} \sum_{i=1}^{n} \mathcal{L}_{\mx, \mh} \left( \mathbf{X}^{(i)} \right) - \mathbb{E} \left[ \mathcal{L}_{\mx, \mh} \left( \mathbf{X} \right)  \right] \bigg| \right] \leq \sqrt{ \mathbb{V} \mathrm{ar} \left[ n^{-1} \sum_{i=1}^{n} \mathcal{L}_{\mx, \mh} \left( \mathbf{X}^{(i)} \right) \right] } = O\left(n^{-\frac{1}{2}} \rho^{\frac{1}{2}} \left( \mh \right)\right)
    \end{align*}
    by \eqref{eq:3.35}.
    Hence, there exist constants $L_1 > 0$ and $N_1 \in \mathbb{N}$ such that
    \begin{align} \label{eq:3.39}
        \mathbb{E} \left[ \bigg| n^{-1} \sum_{i=1}^{n} \mathcal{L}_{\mx, \mh} \left( \mathbf{X}^{(i)} \right) - \mathbb{E} \left[ \mathcal{L}_{\mx, \mh} \left( \mathbf{X} \right)  \right] \bigg| \right] \leq L_1 n^{-\frac{1}{2}} \rho^{\frac{1}{2}} \left( \mh \right)
    \end{align}
    whenever $n \geq N_1$. 
    
    Regarding the second term in \eqref{eq:new3.38}, we apply Theorems 2.7.11 and 2.14.2 of van der Vaart and Wellner (1996). To do so, we define a set $\mathcal{H}_{\mh, \delta}$ and a function $H_{\mh, \delta} : \mathbb{T}^{d} \times \mathbb{M} \to \mathbb{R}$ as
    \begin{gather*}
        \mathcal{H}_{\mh, \delta} := \{ U_{\mh, 0}(\cdot, \cdot, y) : y \in B_{\mathbb{M}}\left(\tilde{m}_{\mh, 0} (\mathbf{x}), \delta\right) \}, \quad H_{\mh, \delta} (\mz, w) := 2 \delta \mathrm{diam} (\mathbb{M})\mathcal{L}_{\mx, \mh} \left( \mz \right)
    \end{gather*}
    for $\delta > 0$. By Theorem 2.14.2 of van der Vaart and Wellner (1996), there exists a constant $C>0$ such that
    \begin{align} \begin{split} \label{s0hatbound_step1}
        &\mathbb{E} \left[ \sup_{y \in B_{\mathbb{M}}\left(\tilde{m}_{\mh, 0} (\mathbf{x}), \delta\right)} \left| \hat{S}_{\mh,0}(\mx,y) \right| \right]  \\
        & \leq \frac{C}{\sqrt{n}} \cdot \left\| H_{\mh, \delta} \right\|_{L_2 (P)} \cdot \int_{0}^{1} \sqrt{1+\log N_{[ ]} \left( \left\| H_{\mh, \delta} \right\|_{L_2 (P)} \epsilon, \mathcal{H}_{\mh, \delta}, d_{\mathbb{M}} \right)} \, \mathrm{d} \epsilon,
    \end{split} \end{align}
    where $\left\| H_{\mh, \delta} \right\|_{L_2 (P)} := \left( \mathbb{E}\left[ H_{\mh, \delta}^2 (\mathbf{X}, Y) \right] \right)^{1/2} = 2 \delta \mathrm{diam} (\mathbb{M}) \left( \mathbb{E}\left[ \mathcal{L}_{\mx, \mh}^2 \left( \mathbf{X} \right) \right] \right)^{1/2}$.
    Note that
    \begin{align} \label{eq:temp3.48}
        \left|  U_{\mh, 0}(\mz, w, y_1) -  U_{\mh, 0}(\mz, w, y_2) \right|
        \leq d_{\mathbb{M}}(y_1, y_2) \frac{H_{\mh, \delta} (\mz, w)}{\delta}
    \end{align}
    for any $y_1, y_2 \in B_{\mathbb{M}}\left(\tilde{m}_{\mh, 0} (\mathbf{x}), \delta\right)$. Combining Theorem 2.7.11 of van der Vaart and Wellner (1996), \eqref{s0hatbound_step1} and \eqref{eq:temp3.48}, we get
    \begin{align} \begin{split} \label{s0hatbound_step2}
        &\mathbb{E} \left[ \sup_{y \in B_{\mathbb{M}}\left(\tilde{m}_{\mh, 0} (\mathbf{x}), \delta\right)} \left| \hat{S}_{\mh,0}(\mx,y) \right| \right]  \\
        & \leq \frac{C}{\sqrt{n}} \cdot \left\| H_{\mh, \delta} \right\|_{L_2 (P)} \cdot \int_{0}^{1} \sqrt{1+\log N \left( \frac{\epsilon}{2}, B_{\mathbb{M}}\left(\tilde{m}_{\mh, 0} (\mathbf{x}), \delta\right), d_{\mathbb{M}} \right)} \, \mathrm{d} \epsilon \\
        & = \frac{2C}{\sqrt{n}} \cdot \left\| H_{\mh, \delta} \right\|_{L_2 (P)} \cdot \int_{0}^{\frac{1}{2}} \sqrt{1+\log N \left( \epsilon, B_{\mathbb{M}}\left(\tilde{m}_{\mh, 0} (\mathbf{x}), \delta\right), d_{\mathbb{M}} \right)} \, \mathrm{d} \epsilon.
    \end{split} \end{align}
    The condition \ref{con:M3} implies that there exist constants $L_2 >0$ and $\delta_{\mathbb{M}}\in(0,1)$ such that
    \begin{align} \label{eq:3.41}
        \sup_{y \in \mathbb{M}:d_{\mathbb{M}}(y,m_{\oplus}(\mx))<r_{\mathbb{M}}} \int_{0}^{1/2} \sqrt{1+\log N ( \delta \epsilon, B_{\mathbb{M}}(y, \delta), d_{\mathbb{M}} )} \, \mathrm{d} \epsilon \leq L_2  \delta ^{\alpha_{\mathbb{M}}- 1}
    \end{align}
    whenever $\delta \in (0, \delta_{\mathbb{M}}]$, where $r_{\mathbb{M}}>0$ and $\alpha_{\mathbb{M}} \in (0, 1]$ are the constants defined in the condition \ref{con:M3}. Also, \Cref{thm:3.1} implies that there exists a constant $N_2\in\mathbb{N}$ such that
    \begin{align} \label{eq:3.42}
        d_{\mathbb{M}}(\tilde{m}_{\mh, 0} (\mathbf{x}),m_{\oplus}(\mx)) < r_{\mathbb{M}}
    \end{align}
    whenever $n \geq N_2$. Combining \eqref{eq:3.41} and \eqref{eq:3.42}, we get
    \begin{align} \label{entropyboundmh0}
        \int_{0}^{1/2} \sqrt{1+\log N ( \delta \epsilon, B_{\mathbb{M}}\left( \tilde{m}_{\mh, 0} (\mathbf{x}), \delta \right), d_{\mathbb{M}} )} \, \mathrm{d} \epsilon \leq L_2  \delta ^{\alpha_{\mathbb{M}}- 1}
    \end{align}
    whenever $\delta \in (0, \delta_{\mathbb{M}}]$ and $n \geq N_2$. \Cref{lemma:1.4}, \Cref{lemma:new2.2} and \eqref{eq:new7.55} imply that there exist constants $L_3 > 0$ and $N_3 \in \mathbb{N}$ such that
    \begin{align} \label{eq:temp3.47}
        \mathbb{E}\left[ \mathcal{L}_{\mx, \mh}^2 \left( \mathbf{X} \right) \right] \leq  L_3 \rho ( \mh ), \quad \left( \mathbb{E} \left[ \mathcal{L}_{\mx, \mh} \left( \mathbf{X} \right)  \right] \right)^{-1} \leq L_3 \rho^{-1} \left( \mh \right)
    \end{align}
    whenever $n \geq N_3$. Combining \eqref{s0hatbound_step2}, \eqref{entropyboundmh0} and \eqref{eq:temp3.47}, we get
    \begin{align} \begin{split} \label{s0hatbound_step3}
        \mathbb{E} \left[ \sup_{y \in B_{\mathbb{M}}\left(\tilde{m}_{\mh, 0} (\mathbf{x}), \delta\right)} \left| \hat{S}_{\mh,0}(\mx,y) \right| \right] & \leq \frac{2C}{\sqrt{n}} \cdot \left( 2\delta \mathrm{diam} (\mathbb{M}) L_3^{\frac{1}{2}} \rho^{\frac{1}{2}} \left( \mh \right)\right) \cdot \left( L_2 \delta^{\alpha_{\mathbb{M}}- 1} \right) \\
        & =  4C L_2 L_3^{\frac{1}{2}} \mathrm{diam} (\mathbb{M}) \delta^{\alpha_{\mathbb{M}}} n^{-\frac{1}{2}}\rho^{\frac{1}{2}} \left( \mh \right)
    \end{split} \end{align}
    whenever $\delta \in (0, \delta_{\mathbb{M}}]$ and $n \geq \max \{N_2 , N_3\}$. 
    
    Let $N_4 := \max \{ N_1, N_2, N_3 \}$ and $L_4 := 2 L_3 \mathrm{diam} (\mathbb{M})( L_1 + 2C L_2 L_3^{1/2} )$. Combining \eqref{eq:new3.38}, \eqref{eq:3.39}, \eqref{eq:temp3.47} and \eqref{s0hatbound_step3}, we get
    \begin{align} \begin{split} \label{eq:3.44}
        & \mathbb{E} \left[ \sup_{y \in B_{\mathbb{M}}\left(\tilde{m}_{\mh, 0} (\mathbf{x}), \delta\right)} \left| \hat{T}_{\mh, 0}(\mathbf{x}, y) - \hat{T}_{\mh, 0}(\mathbf{x}, \tilde{m}_{\mh, 0}(\mathbf{x})) \right| \right] \\
        & \leq \left( \mathbb{E} \left[ \mathcal{L}_{\mx, \mh} \left( \mathbf{X} \right)  \right] \right)^{-1} \cdot \left( 2 \mathrm{diam}(\mathbb{M}) \right) \cdot \delta \cdot \bigg| \mathbb{E} \left[ \mathcal{L}_{\mx, \mh} \left( \mathbf{X} \right)  \right] - n^{-1} \sum_{i=1}^{n} \mathcal{L}_{\mx, \mh} \left( \mathbf{X}^{(i)} \right) \bigg| \\
        & \quad + \left( \mathbb{E} \left[ \mathcal{L}_{\mx, \mh} \left( \mathbf{X} \right)  \right] \right)^{-1} \cdot \mathbb{E} \left[ \sup_{y \in B_{\mathbb{M}}\left(\tilde{m}_{\mh, 0} (\mathbf{x}), \delta\right)} \left| \hat{S}_{\mh,0}(\mx,y) \right| \right] \\
        & \leq 2 L_1 L_3 \mathrm{diam} (\mathbb{M}) \delta n^{-\frac{1}{2}} \rho^{-\frac{1}{2}} \left( \mh \right) + 4C L_2 L_3^{\frac{3}{2}} \mathrm{diam} (\mathbb{M}) \delta^{\alpha_{\mathbb{M}}}  n^{-\frac{1}{2}}\rho^{-\frac{1}{2}} \left( \mh \right) \\
        & \leq L_4 \delta^{\alpha_{\mathbb{M}}} n^{-\frac{1}{2}}\rho^{-\frac{1}{2}} \left( \mh \right)
    \end{split} \end{align}
    whenever $\delta \in (0, \delta_{\mathbb{M}}]$ and $n \geq N_4$.
    
    Let $\epsilon>0$ be given. Recall the constants $\eta_{\oplus} >0$ and $C_{\oplus} >0$ defined in the condition \ref{con:M2}. Since $\sum_{k=1}^{\infty} 4^{-k(\beta_{\oplus} - \alpha_{\mathbb{M}}) / \beta_{\oplus} } < \infty$, there exists a constant $M_{\epsilon} \in \mathbb{N}$, depending on $\epsilon$, such that
    \begin{align} \label{eq:3.45}
        \frac{L_4}{C_{\oplus}}\sum_{k=M_{\epsilon}+1}^{\infty} 4^{-\frac{k(\beta_{\oplus} - \alpha_{\mathbb{M}})}{\beta_{\oplus}} }  < \frac{\epsilon}{8}.
    \end{align}
    By \Cref{thm:3.3}, there exists a constant $N_{\epsilon} \in \mathbb{N}$, depending on $\epsilon$, such that
    \begin{align} \label{eq:3.46}
        \mathbb{P} \left( d_{\mathbb{M}}(\hat{m}_{\mh, 0}(\mathbf{x}), \tilde{m}_{\mh, 0}(\mathbf{x})) \geq \eta_{\oplus} \right) < \frac{\epsilon}{2}
    \end{align}
    whenever $n \geq N_{\epsilon}$.
    Let $t_n := \left(n\rho(\mh)\right)^{\beta_{\oplus} / (4 \beta_{\oplus} - 4 \alpha_{\mathbb{M}})}$. We define events $E_{n, \epsilon}$, $A_{n, k}$ and $B_n$ as
    \begin{align*}
        E_{n, \epsilon} :=&  \left\{ d_{\mathbb{M}}(\hat{m}_{\mh, 0}(\mathbf{x}), \tilde{m}_{\mh, 0}(\mathbf{x})) > 2^{\frac{2M_{\epsilon}}{\beta_{\oplus}}} \left(n\rho(\mh)\right)^{\frac{-1}{2\beta_{\oplus} - 2\alpha_{\mathbb{M}} }}\right\} \\
        =&  \left\{ t_n d_{\mathbb{M}}(\hat{m}_{\mh, 0}(\mathbf{x}), \tilde{m}_{\mh, 0}(\mathbf{x}))^{\frac{\beta_{\oplus}}{2}} > 2^{M_{\epsilon}} \right\}, \\
        A_{n, k} :=&  \left\{ 2 ^ {k-1} < t_n d_{\mathbb{M}}(\hat{m}_{\mh, 0}(\mathbf{x}), \tilde{m}_{\mh, 0}(\mathbf{x}))^{\frac{\beta_{\oplus}}{2}} \leq 2^{k} \right\}, \\ 
        B_n :=&  \left\{ d_{\mathbb{M}}(\hat{m}_{\mh, 0}(\mathbf{x}), \tilde{m}_{\mh, 0}(\mathbf{x})) < \eta_{\oplus} \right\}.
    \end{align*}
    Since $E_{n, \epsilon} = \cup_{k=M_{\epsilon}+1}^{\infty} A_{n,k}$, it holds that
    \begin{align} \label{eq:3.47}
        \mathbb{P} (E_{n, \epsilon}) \leq
        \sum_{k=M_{\epsilon}+1}^{\infty} \mathbb{P} (A_{n,k} \cap B_n) + \mathbb{P} (B_n^c).
    \end{align}
    Let $r_{n, k} := \min \{ (4^k t_n^{-2})^{1/\beta_{\oplus}}, \delta_{\mathbb{M}} \}$.
    Since $\hat{M}_{\mh, 0} (\mathbf{x}, \hat{m}_{\mh, 0}(\mathbf{x}))\le \hat{M}_{\mh, 0} (\mathbf{x}, \tilde{m}_{\mh, 0} (\mathbf{x}) )$, it holds that
    \begin{align} \begin{split} \label{eq:3.49}
          \sup_{y \in B_{\mathbb{M}}(\tilde{m}_{\mh, 0}(\mathbf{x}), r_{n, k})} \left| \hat{T}_{\mh, 0}(\mathbf{x}, y) - \hat{T}_{\mh, 0}(\mathbf{x}, \tilde{m}_{\mh, 0}(\mathbf{x})) \right| 
          &\ge - \hat{T}_{\mh, 0}(\mathbf{x}, \hat{m}_{\mh, 0}(\mathbf{x})) + \hat{T}_{\mh, 0}(\mathbf{x}, \tilde{m}_{\mh, 0}(\mathbf{x})) \\
          &\ge \tilde{M}_{\mh, 0} (\mathbf{x}, \hat{m}_{\mh, 0}(\mathbf{x})) - \tilde{M}_{\mh, 0} (\mathbf{x}, \tilde{m}_{\mh, 0} (\mathbf{x}) )  \\
          &\ge C_{\oplus} d_{\mathbb{M}}(\hat{m}_{\mh, 0}(\mathbf{x}), \tilde{m}_{\mh, 0} (\mathbf{x}) )^{\beta_{\oplus}} \\
          &\ge C_{\oplus} \frac{4^{k-1}}{t_n^2}
    \end{split} \end{align} 
    on the event $A_{n,k} \cap B_n$. Using a version of \eqref{eq:3.44} with $\delta$ being replaced by $r_{n, k}$ and applying Markov's inequality to the probability of \eqref{eq:3.49}, we get
    \begin{align} \begin{split} \label{eq:3.48}
        \mathbb{P} (A_{n,k} \cap B_n) & \leq  \mathbb{P} \left( \sup_{y \in B_{\mathbb{M}}(\tilde{m}_{\mh, 0}(\mathbf{x}), r_{n, k})} \left| \hat{T}_{\mh, 0}(\mathbf{x}, y) - \hat{T}_{\mh, 0}(\mathbf{x}, \tilde{m}_{\mh, 0}(\mathbf{x})) \right| \geq C_{\oplus} \frac{4^{k-1}}{t_n^2}\right) \\
        & \leq \frac{t_n^2}{C_{\oplus} 4^{k-1}} \mathbb{E} \left[ \sup_{y \in B_{\mathbb{M}}(\tilde{m}_{\mh, 0}(\mathbf{x}), r_{n, k})} \left| \hat{T}_{\mh, 0}(\mathbf{x}, y) - \hat{T}_{\mh, 0}(\mathbf{x}, \tilde{m}_{\mh, 0}(\mathbf{x})) \right| \right] \\
        &\leq \frac{L_4 t_n^2}{C_{\oplus} 4^{k-1}} \left( r_{n, k} \right)^{\alpha_{\mathbb{M}}} n^{-\frac{1}{2}}\rho^{-\frac{1}{2}} \left( \mh \right) \\
        &\le \frac{4L_4}{C_{\oplus}} 4^{-\frac{k(\beta_{\oplus} - \alpha_{\mathbb{M}})}{\beta_{\oplus}}}
    \end{split} \end{align}
    whenever $n \geq N_4$. Combining \eqref{eq:3.45}, \eqref{eq:3.46}, \eqref{eq:3.47} and \eqref{eq:3.48}, we get
    \begin{align*}
        \mathbb{P} ( E_{n, \epsilon} ) \leq \frac{4L_4}{C_{\oplus}} \sum_{k=M_{\epsilon}+1}^{\infty} 4^{-\frac{k(\beta_{\oplus} - \alpha_{\mathbb{M}})}{\beta_{\oplus}} } + \mathbb{P} (B_n^c) < \epsilon
    \end{align*}
    whenever $n \geq \max \{ N_4 , N_{\epsilon} \}$.
    Since $\epsilon>0$ is arbitrary, we get the desired result.
\end{proof}

\subsection{Proof of Theorem \ref{thm:main4.3} for  $\hat{m}_{\mh, 0}(\mathbf{x})$}

The theorem follows from \Cref{thm:3.2} and \Cref{thm:3.4}.

\section{Proofs of \Cref{thm:consistency} and \Cref{thm:main4.3} for  $\hat{m}_{\mh, 1}(\mathbf{x})$} \label{proofs:local linear}

\setcounter{equation}{0}
\setcounter{subsection}{0}

First, we collect some notations. We define
\begin{align*}
     \tilde{\tau}_{\mh, 0}(\mathbf{x}, y) &:= \mathbb{E} \left[ \mathcal{L}_{\mx, \mh} ( \mathbf{X} ) g_y (\mathbf{X}) \right],\\ 
     \bm{\tilde{\tau}}_{\mh, 1}(\mathbf{x}, y) &:= \mathbb{E} \left[ \mathcal{L}_{\mx, \mh} ( \mathbf{X} ) \bm{\Phi}_{\mx}^{-1} (\mathbf{X}) g_y (\mathbf{X}) \right], \\
     \tilde{\nu}_{\mh, 0}(\mathbf{x}, y) &:= \mathbb{E} \left[ \mathcal{L}_{\mx, \mh} ( \mathbf{X} )  d_{\mathbb{M}}^2(Y, y) \right],\\  
     \bm{\tilde{\nu}}_{\mh, 1}(\mathbf{x}, y) &:= \mathbb{E} \left[ \mathcal{L}_{\mx, \mh} ( \mathbf{X} )  \bm{\Phi}_{\mx}^{-1} (\mathbf{X}) d_{\mathbb{M}}^2(Y, y) \right], \\
    \hat{\nu}_{\mh, 0}(\mathbf{x}, y) &:= n^{-1} \sum_{i=1}^{n} \mathcal{L}_{\mx, \mh} ( \mathbf{X}^{(i)} ) d_{\mathbb{M}}^2(Y^{(i)}, y),\\ 
    \bm{\hat{\nu}}_{\mh, 1}(\mathbf{x}, y) &:= n^{-1} \sum_{i=1}^{n} \mathcal{L}_{\mx, \mh} ( \mathbf{X}^{(i)} ) \bm{\Phi}_{\mx}^{-1} (\mathbf{X}^{(i)}) d_{\mathbb{M}}^2(Y^{(i)}, y).
\end{align*}
Then, $\tilde{M}_{\mh,1}$ and $\hat{M}_{\mh,1}$ can be written as
\begin{align} \begin{gathered} \label{eq:4.80}
    \tilde{M}_{\mh,1} (\mathbf{x}, y) = \frac{1}{\tilde{\sigma}_{\mh}(\mathbf{x})} ( \tilde{\nu}_{\mh, 0}(\mathbf{x}, y) - \bm{\tilde{\mu}}_{\mh, 1}(\mathbf{x})\tran\bm{\tilde{\mu}}_{\mh, 2}(\mathbf{x})^{-1} \bm{\tilde{\nu}}_{\mh, 1}(\mathbf{x}, y)  ),  \\
    \hat{M}_{\mh,1} (\mathbf{x}, y) = \frac{1}{\hat{\sigma}_{\mh}(\mathbf{x})} ( \hat{\nu}_{\mh, 0}(\mathbf{x}, y) - \bm{\hat{\mu}}_{\mh, 1}(\mathbf{x})\tran\bm{\hat{\mu}}_{\mh, 2}(\mathbf{x})^{-1} \bm{\hat{\nu}}_{\mh, 1}(\mathbf{x}, y)  ),
\end{gathered} \end{align}
respectively. We write
\begin{align*}
    a_{j, k}(L) := \int_0^{\infty} L^{k}(r^2) r^{j} \mathrm{d}r
\end{align*}
for $j \in \mathbb{N}_{\geq 0}$ and $k\in\{1,2\}$. We define the diagonal matrices
\begin{align*}
    \mathbf{\Lambda}_{\mh} = \begin{bmatrix}
        h_1 & \cdots & 0 \\
        \vdots & \ddots & \vdots \\
        0 & \cdots & h_d
    \end{bmatrix}, \quad \mathbf{C}_{\mh} = \begin{bmatrix}
        c_{\mh, 2\mathbf{e}_{1}, 1}(L) & \cdots & 0 \\
        \vdots & \ddots & \vdots \\
        0 & \cdots & c_{\mh, 2\mathbf{e}_{d}, 1}(L)
    \end{bmatrix}.
\end{align*}

\subsection{Lemmas for proof of \Cref{thm:consistency}}

We first prove that $\bm{\tilde{\mu}}_{\mh, 2}(\mathbf{x})$ is invertible for sufficiently large $n$ and that $\bm{\hat{\mu}}_{\mh, 2}(\mathbf{x})$ is invertible with probability tending to one. 

\begin{lemma} \label{lemma:new4.4}
    Assume that the conditions \ref{con:L1} and \ref{con:D2} hold and that $\lim_{n\rightarrow\infty} \| \mh \|_2=0$. Then, it holds that
    \begin{align} \begin{gathered}\label{eq:new4.38}
          \left\| \mathbf{\Lambda}_{\mh}^{-1} \bm{\tilde{\mu}}_{\mh, 1}(\mathbf{x}) \right\|_2 = o\left(  \rho ( \mh ) \right), \quad \left\| \mathbf{\Lambda}_{\mh}^{-1} \bm{\tilde{\mu}}_{\mh, 2}(\mathbf{x}) \mathbf{\Lambda}_{\mh}^{-1} - \mathbf{\Lambda}_{\mh}^{-2} \mathbf{C}_{\mh} f(\mx) \right\|_2 = o\left(\rho ( \mh )\right).
    \end{gathered} \end{align}
    Moreover, $\mathbf{C}_{\mh}$ is invertible for sufficiently large $n$ and it holds that
    \begin{align} \label{chinvertible}
        \left\| \frac{1}{c_{\mh, \bm{0}_d, 1}(L)}\mathbf{\Lambda}_{\mh}^{-2} \mathbf{C}_{\mh} - \frac{2a_{2, 1}(L)}{a_{0, 1}(L)}\mathbf{I}_d \right\|_2 = o\left(1\right), \quad \left\| c_{\mh, \bm{0}_d, 1}(L)\mathbf{\Lambda}_{\mh}^{2} \mathbf{C}_{\mh}^{-1} - \frac{a_{0, 1}(L)}{2a_{2, 1}(L)}\mathbf{I}_d \right\|_2 = o\left(1\right).
    \end{align}
    If the condition \ref{con:D1} further holds, then $\bm{\tilde{\mu}}_{\mh, 2}(\mathbf{x})$ is invertible for sufficiently large $n$ and it holds that
    \begin{align} \label{eq:new4.39}
    \begin{split}
        \left\| \mathbf{\Lambda}_{\mh} \bm{\tilde{\mu}}_{\mh, 2}(\mathbf{x})^{-1} \mathbf{\Lambda}_{\mh} - \frac{1}{f(\mathbf{x})} \mathbf{\Lambda}_{\mh}^{2} \mathbf{C}_{\mh}^{-1} \right\|_2 &= o\left(  \rho^{-1} \left( \mh \right) \right),  \\
        \left\| c_{\mh, \bm{0}_d, 1}(L)\mathbf{\Lambda}_{\mh} \bm{\tilde{\mu}}_{\mh, 2}(\mathbf{x})^{-1} \mathbf{\Lambda}_{\mh} - \frac{a_{0, 1}(L)}{2a_{2, 1}(L) f(\mx)}\mathbf{I}_d \right\|_2 &= o\left(  1 \right).
    \end{split}
    \end{align}
\end{lemma}

\begin{proof} [Proof of \Cref{lemma:new4.4}]
    We first prove \eqref{eq:new4.38}. Combining \eqref{symmetricproperty}, \eqref{eq:2.7} and \Cref{lemma:1.2}, we get
    \begin{align} \begin{split} \label{eq:new4.43}
        \int_{\mathbb{T}^{d}} \mathcal{L}_{\mx, \mh} \left( \mz \right) \bm{\Phi}_{\mx}^{-1} (\mz) \omega_{1}^{d}(\mathrm{d}\mz) & = \int_{[-\pi, \pi)^{d}}  \mathcal{L}_{\mx, \mh} \left( \bm{\Phi}_{\mx} \left( \bm{\theta} \right) \right)  \bm{\theta}  \mathrm{d} \bm{\theta} \\
        & = \int_{[-\pi, \pi)^{d}}  \bm{L}_{\mh} (\bm{\theta}) \bm{\theta}  \mathrm{d} \bm{\theta} \\
        & = \bm{0}_{d}, \\
        \int_{\mathbb{T}^{d}} \mathcal{L}_{\mx, \mh} \left( \mz \right) \bm{\Phi}_{\mx}^{-1} (\mz) \bm{\Phi}_{\mx}^{-1} (\mz)\tran \omega_{1}^{d}(\mathrm{d}\mz)  & = \int_{[-\pi, \pi)^{d}}  \mathcal{L}_{\mx, \mh} \left( \bm{\Phi}_{\mx} \left( \bm{\theta} \right) \right)  \bm{\theta} \bm{\theta} \tran \mathrm{d} \bm{\theta} \\
        & = \int_{[-\pi, \pi)^{d}}  \bm{L}_{\mh} (\bm{\theta}) \bm{\theta} \bm{\theta} \tran \mathrm{d} \bm{\theta} \\
        & = \mathbf{C}_{\mh}. \\
    \end{split} \end{align}
    The Cauchy-Schwarz inequality and \eqref{eq:new4.43} imply that
    \begin{align} \begin{split} \label{mu1bound_step1}
        \big\| \mathbf{\Lambda}_{\mh}^{-1} \bm{\tilde{\mu}}_{\mh, 1}(\mathbf{x}) \big\|_{2}^{2} & = \bigg\| \mathbf{\Lambda}_{\mh}^{-1} \int_{\mathbb{T}^{d}} \mathcal{L}_{\mx, \mh} \left( \mz \right) \bm{\Phi}_{\mx}^{-1} (\mz)  \big( f(\mz) - f(\mx) \big) \omega_{1}^{d}(\mathrm{d}\mz)   \bigg\|_{2}^{2}  \\
        & \leq \left( \int_{\mathbb{T}^{d}} \mathcal{L}_{\mx, \mh} \left( \mz \right) \left\| \mathbf{\Lambda}_{\mh}^{-1}\bm{\Phi}_{\mx}^{-1} (\mz) \right\|_2 \big| f(\mz) - f(\mx) \big| \omega_{1}^{d} (\mathrm{d}\mz) \right)^2  \\
        & = \left( \int_{[-\pi, \pi)^{d}}  \bm{L}_{\mh} (\bm{\theta}) \left\| \mathbf{\Lambda}_{\mh}^{-1} \bm{\theta} \right\|_2  \left|  f \left( \bm{\Phi}_{\mx} \left( \bm{\theta} \right) \right) - f(\mx)\right| \mathrm{d} \bm{\theta} \right)^2\\
        & \leq \left( \int_{[-\pi, \pi)^{d}}  \bm{L}_{\mh} (\bm{\theta})    \left|  f \left( \bm{\Phi}_{\mx} \left( \bm{\theta} \right) \right) - f(\mx)\right| \mathrm{d} \bm{\theta} \right) \\
        & \quad \cdot \left( \int_{[-\pi, \pi)^{d}}  \bm{L}_{\mh} (\bm{\theta}) \left\| \mathbf{\Lambda}_{\mh}^{-1} \bm{\theta} \right\|_2^2   \left|  f \left( \bm{\Phi}_{\mx} \left( \bm{\theta} \right) \right) - f(\mx)\right| \mathrm{d} \bm{\theta} \right) \\
        & \leq \left( \int_{[-\pi, \pi)^{d}}  \bm{L}_{\mh} (\bm{\theta})    \left|  f \left( \bm{\Phi}_{\mx} \left( \bm{\theta} \right) \right) - f(\mx)\right| \mathrm{d} \bm{\theta} \right) \cdot 2 \sup_{\mz \in \mathbb{T}^{d}} f ( \mz ) \\
        & \quad \cdot  \left( \int_{[-\pi, \pi)^{d}}  \bm{L}_{\mh} (\bm{\theta}) \left\| \mathbf{\Lambda}_{\mh}^{-1} \bm{\theta} \right\|_2^2  \mathrm{d} \bm{\theta} \right).
    \end{split} \end{align}
    By arguing as in the proof of \Cref{lemma:new2.2}, we get
    \begin{align} \label{mu1bound_step2}
        \int_{[-\pi, \pi)^{d}}  \bm{L}_{\mh} (\bm{\theta})    \left|  f \left( \bm{\Phi}_{\mx} \left( \bm{\theta} \right) \right) - f(\mx)\right| \mathrm{d} \bm{\theta} = o\left(\rho(\mh)\right).
    \end{align}
    Also, \eqref{symmetricproperty} and \Cref{lemma:1.4} imply that
    \begin{align} \begin{split} \label{mu1bound_step3}
        \int_{[-\pi, \pi)^{d}}  \bm{L}_{\mh}^{k} (\bm{\theta}) \| \mathbf{\Lambda}_{\mh}^{-1}\bm{\theta} \|_{2}^{2} \mathrm{d} \bm{\theta} &= \sum_{\ell=1}^{d} \frac{1}{h_\ell^2}\int_{[-\pi, \pi)^{d}} \bm{L}_{\mh}^{k} (\bm{\theta}) \theta_\ell^2 \mathrm{d} \bm{\theta} = \sum_{\ell=1}^{d}\frac{c_{\mh, 2\mathbf{e}_{\ell}, k}(L)}{h_\ell^2} = O\left( \rho(\mh)\right)
    \end{split} \end{align}
    for any $k\in\{1, 2\}$. Combining \eqref{mu1bound_step1}, \eqref{mu1bound_step2} and \eqref{mu1bound_step3}, we get
    \begin{align} \label{mu1bound_step4}
        \big\| \mathbf{\Lambda}_{\mh}^{-1} \bm{\tilde{\mu}}_{\mh, 1}(\mathbf{x}) \big\|_{2}^2 = o\left(\rho^{2} \left( \mh \right) \right).
    \end{align}  
    Similarly, the Cauchy-Schwarz inequality and \eqref{eq:new4.43} imply that
    \begin{align} \begin{split} \label{mu2bound_step1}
        & \left\| \mathbf{\Lambda}_{\mh}^{-1} \bm{\tilde{\mu}}_{\mh, 2}(\mathbf{x}) \mathbf{\Lambda}_{\mh}^{-1} - \mathbf{\Lambda}_{\mh}^{-2} \mathbf{C}_{\mh} f(\mx) \right\|_2^2 \\
        & = \bigg\| \mathbf{\Lambda}_{\mh}^{-1} \left( \int_{\mathbb{T}^{d}} \mathcal{L}_{\mx, \mh} \left( \mz \right) \bm{\Phi}_{\mx}^{-1} (\mz) \bm{\Phi}_{\mx}^{-1} (\mz) \tran  \left( f(\mz) - f(\mathbf{x}) \right) \omega_{1}^{d}(\mathrm{d}\mz)  \right) \mathbf{\Lambda}_{\mh}^{-1} \bigg\|_2^2  \\
        & = \bigg\| \int_{\mathbb{T}^{d}} \mathcal{L}_{\mx, \mh} \left( \mz \right) \mathbf{\Lambda}_{\mh}^{-1}\bm{\Phi}_{\mx}^{-1} (\mz) \bm{\Phi}_{\mx}^{-1} (\mz) \tran \mathbf{\Lambda}_{\mh}^{-1} \left( f(\mz) - f(\mathbf{x}) \right) \omega_{1}^{d}(\mathrm{d}\mz)  \bigg\|_2^2 \\
        & \leq \left( \int_{\mathbb{T}^{d}} \mathcal{L}_{\mx, \mh} \left( \mz \right) \left\| \mathbf{\Lambda}_{\mh}^{-1} \bm{\Phi}_{\mx}^{-1} (\mz) \right\|_2^2 \big| f(\mz) - f(\mx) \big| \omega_{1}^{d} (\mathrm{d}\mz) \right)^2 \\
        & = \left( \int_{[-\pi, \pi)^{d}}  \bm{L}_{\mh} (\bm{\theta}) \left\| \mathbf{\Lambda}_{\mh}^{-1} \bm{\theta} \right\|_2^2   \left|  f \left( \bm{\Phi}_{\mx} \left( \bm{\theta} \right) \right) - f(\mx)\right| \mathrm{d} \bm{\theta} \right)^2\\
        & \leq \left( \int_{[-\pi, \pi)^{d}}  \bm{L}_{\mh} (\bm{\theta})   \left|  f \left( \bm{\Phi}_{\mx} \left( \bm{\theta} \right) \right) - f(\mx)\right| \mathrm{d} \bm{\theta} \right) \\
        & \quad \cdot \left( \int_{[-\pi, \pi)^{d}}  \bm{L}_{\mh} (\bm{\theta}) \left\| \mathbf{\Lambda}_{\mh}^{-1} \bm{\theta} \right\|_2^4   \left|  f \left( \bm{\Phi}_{\mx} \left( \bm{\theta} \right) \right) - f(\mx)\right| \mathrm{d} \bm{\theta} \right) \\
        & \leq \left( \int_{[-\pi, \pi)^{d}}  \bm{L}_{\mh} (\bm{\theta})   \left|  f \left( \bm{\Phi}_{\mx} \left( \bm{\theta} \right) \right) - f(\mx)\right| \mathrm{d} \bm{\theta} \right) \cdot 2 \sup_{\mz \in \mathbb{T}^{d}} f ( \mz ) \\
        & \quad \cdot \left( \int_{[-\pi, \pi)^{d}}  \bm{L}_{\mh} (\bm{\theta}) \left\| \mathbf{\Lambda}_{\mh}^{-1} \bm{\theta} \right\|_2^4    \mathrm{d} \bm{\theta} \right).
    \end{split} \end{align}
    Also, \eqref{symmetricproperty} and \Cref{lemma:1.4} imply that
    \begin{align} \begin{split} \label{mu2bound_step2}
        \int_{[-\pi, \pi)^{d}}  \bm{L}_{\mh}^{k} (\bm{\theta}) \left\| \mathbf{\Lambda}_{\mh}^{-1} \bm{\theta} \right\|_2^4 \mathrm{d} \bm{\theta} & = \sum_{\ell=1}^{d} \sum_{m=1}^{d} \frac{1}{h_\ell^2 h_m^2}\int_{[-\pi, \pi)^{d}} \bm{L}_{\mh}^{k} (\bm{\theta}) \theta_\ell^2\theta_m^2 \mathrm{d} \bm{\theta} \\
        & = \sum_{\ell=1}^{d} \sum_{m=1}^{d} \frac{c_{\mh, 2\left( \mathbf{e}_{\ell} + \mathbf{e}_{m} \right), k}(L)}{h_\ell^2 h_m^2} \\
        & = O\left( \rho(\mh)\right)
    \end{split} \end{align}
    for any $k\in\{1, 2\}$. Combining \eqref{mu1bound_step2}, \eqref{mu2bound_step1} and \eqref{mu2bound_step2}, we get
    \begin{align} \label{mu2bound_step3}
        \left\| \mathbf{\Lambda}_{\mh}^{-1} \bm{\tilde{\mu}}_{\mh, 2}(\mathbf{x}) \mathbf{\Lambda}_{\mh}^{-1} - \mathbf{\Lambda}_{\mh}^{-2} \mathbf{C}_{\mh} f(\mx) \right\|_2^2 = o\left(\rho^{2} \left( \mh \right) \right).
    \end{align}  
    Then, \eqref{eq:new4.38} follows from \eqref{mu1bound_step4} and \eqref{mu2bound_step3}.

    Now, we prove the invertibility of $\mathbf{C}_{\mh}$ and \eqref{chinvertible}. Note that
    \begin{align} \label{chinvertible_step1}
        \frac{1}{c_{\mh, \bm{0}_d, 1}(L)}\mathbf{\Lambda}_{\mh}^{-2} \mathbf{C}_{\mh} = \mathrm{diag} \left( \frac{c_{\mh, 2\mathbf{e}_{1}, 1}(L)}{c_{\mh, \bm{0}_d, 1}(L) h_1^2} ,\dots, \frac{c_{\mh, 2\mathbf{e}_{d}, 1}(L)}{c_{\mh, \bm{0}_d, 1}(L) h_d^2} \right).
    \end{align}
    By \Cref{lemma:1.4}, we get
    \begin{align} \label{chinvertible_step2}
        \lim_{n \to \infty} \frac{c_{\mh, 2\mathbf{e}_{\ell}, 1}(L)}{c_{\mh, \bm{0}_d, 1}(L) h_\ell^2} = \frac{2a_{2, 1}(L)}{a_{0, 1}(L)} > 0,\quad \ell=1, \dots, d.
    \end{align}
    Combining \eqref{chinvertible_step1} and \eqref{chinvertible_step2}, $\mathbf{C}_{\mh}$ is invertible for sufficiently large $n$ and it holds that
    \begin{align} \label{chinvertible_step3}
        \left\| \frac{1}{c_{\mh, \bm{0}_d, 1}(L)}\mathbf{\Lambda}_{\mh}^{-2} \mathbf{C}_{\mh} - \frac{2a_{2, 1}(L)}{a_{0, 1}(L)}\mathbf{I}_d \right\|_2 = o\left(1\right).
    \end{align}
    Combining \eqref{chinvertible_step2} and \eqref{chinvertible_step3}, we get
    \begin{align*} 
        &\left\| c_{\mh, \bm{0}_d, 1}(L)\mathbf{\Lambda}_{\mh}^{2} \mathbf{C}_{\mh}^{-1} - \frac{a_{0, 1}(L)}{2a_{2, 1}(L)}\mathbf{I}_d \right\|_2 \\
        & = \left\| \left( \frac{2a_{2, 1}(L)}{a_{0, 1}(L)}\mathbf{I}_d - \frac{1}{c_{\mh, \bm{0}_d, 1}(L)}\mathbf{\Lambda}_{\mh}^{-2} \mathbf{C}_{\mh} \right) \cdot \left( \frac{a_{0, 1}(L) c_{\mh, \bm{0}_d, 1}(L)}{2a_{2, 1}(L)}\mathbf{\Lambda}_{\mh}^{2} \mathbf{C}_{\mh}^{-1} \right) \right\|_2\\
        & \leq \left\| \frac{2a_{2, 1}(L)}{a_{0, 1}(L)}\mathbf{I}_d - \frac{1}{c_{\mh, \bm{0}_d, 1}(L)}\mathbf{\Lambda}_{\mh}^{-2} \mathbf{C}_{\mh}  \cdot  \right\|_2 \frac{a_{0, 1}(L)}{2a_{2, 1}(L)}  \left\| c_{\mh, \bm{0}_d, 1}(L) \mathbf{\Lambda}_{\mh}^{2} \mathbf{C}_{\mh}^{-1} \right\|_2  \\ 
        & = o(1).
    \end{align*}
    
    Now, we prove the invertibility of $\bm{\tilde{\mu}}_{\mh, 2}(\mathbf{x})$ and \eqref{eq:new4.39}. Note that
    \begin{align} \begin{split} \label{eq:new7.95}
        &\inf_{\bu \in \mathbb{S}^{d}} \bu\tran\mathbf{\Lambda}_{\mh}^{-1} \bm{\tilde{\mu}}_{\mh, 2}(\mathbf{x}) \mathbf{\Lambda}_{\mh}^{-1} \bu \\
        & \geq  f(\mx) \left( \inf_{\bu \in \mathbb{S}^{d}} \bu\tran \mathbf{\Lambda}_{\mh}^{-2} \mathbf{C}_{\mh}  \bu \right) - \sup_{\bu \in \mathbb{S}^{d}} \bu\tran \left( \mathbf{\Lambda}_{\mh}^{-2} \mathbf{C}_{\mh} f(\mx) - \mathbf{\Lambda}_{\mh}^{-1} \bm{\tilde{\mu}}_{\mh, 2}(\mathbf{x}) \mathbf{\Lambda}_{\mh}^{-1} \right) \bu \\
        & = f(\mx) \min \left\{ \frac{c_{\mh, 2\mathbf{e}_{1}, 1}(L)}{h_1^2} ,\dots, \frac{c_{\mh, 2\mathbf{e}_{d}, 1}(L)}{h_d^2} \right\} - \left\| \mathbf{\Lambda}_{\mh}^{-1} \bm{\tilde{\mu}}_{\mh, 2}(\mathbf{x}) \mathbf{\Lambda}_{\mh}^{-1} - \mathbf{\Lambda}_{\mh}^{-2} \mathbf{C}_{\mh} f(\mx) \right\|_2.
    \end{split} \end{align}
    Combining \Cref{lemma:1.4}, \eqref{eq:new4.38} and \eqref{eq:new7.95}, we get
    \begin{align} \begin{split} \label{eq:new7.96}
        &\liminf_{n \to \infty}\rho^{-1} \left( \mh \right)  \left( \inf_{\bu \in \mathbb{S}^{d}} \bu\tran\mathbf{\Lambda}_{\mh}^{-1} \bm{\tilde{\mu}}_{\mh, 2}(\mathbf{x}) \mathbf{\Lambda}_{\mh}^{-1} \bu \right) \\
        &\geq f(\mx) \liminf_{n \to \infty} \left( \min \left\{ \frac{c_{\mh, 2\mathbf{e}_{1}, 1}(L)}{\rho ( \mh ) h_1^2} ,\dots, \frac{c_{\mh, 2\mathbf{e}_{d}, 1}(L)}{\rho ( \mh ) h_d^2} \right\} \right)\\
        &\quad - \limsup_{n \to \infty} \rho^{-1} \left( \mh \right) \left\| \mathbf{\Lambda}_{\mh}^{-1} \bm{\tilde{\mu}}_{\mh, 2}(\mathbf{x}) \mathbf{\Lambda}_{\mh}^{-1} - \mathbf{\Lambda}_{\mh}^{-2} \mathbf{C}_{\mh} f(\mx) \right\|_2 \\ 
        &= 2^{\frac{3d+2}{2}} a_{0, 1}^{d-1}(L) a_{2, 1}(L) f(\mx) \\
        &>0.
    \end{split} \end{align}
    This implies that $\mathbf{\Lambda}_{\mh}^{-1} \bm{\tilde{\mu}}_{\mh, 2}(\mathbf{x}) \mathbf{\Lambda}_{\mh}^{-1}$ is positive-definite for sufficiently large $n$, and thus $\bm{\tilde{\mu}}_{\mh, 2}(\mathbf{x})$ is invertible for sufficiently large $n$. By \eqref{eq:new7.96}, we get
    \begin{align} \begin{split} \label{eq:new7.97}
        & \limsup_{n \to \infty} \rho ( \mh )  \left\| \left( \mathbf{\Lambda}_{\mh}^{-1} \bm{\tilde{\mu}}_{\mh, 2}(\mathbf{x}) \mathbf{\Lambda}_{\mh}^{-1} \right)^{-1} \right\|_2 \\
        & = \limsup_{n \to \infty} \rho ( \mh ) \left( \inf_{\bu \in \mathbb{S}^{d}} \bu\tran\mathbf{\Lambda}_{\mh}^{-1} \bm{\tilde{\mu}}_{\mh, 2}(\mathbf{x}) \mathbf{\Lambda}_{\mh}^{-1} \bu\right)^{-1} \\
        & = \left( \liminf_{n \to \infty} \rho^{-1} \left( \mh \right)  \left( \inf_{\bu \in \mathbb{S}^{d}} \bu\tran\mathbf{\Lambda}_{\mh}^{-1} \bm{\tilde{\mu}}_{\mh, 2}(\mathbf{x}) \mathbf{\Lambda}_{\mh}^{-1} \bu \right)\right)^{-1} \\
        & \leq 2^{-\frac{3d+2}{2}} \frac{1}{ a_{0, 1}^{d-1}(L) a_{2, 1}(L) f(\mx)}.
    \end{split} \end{align}
    Combining \Cref{lemma:1.4}, \eqref{eq:new4.38}, \eqref{chinvertible_step2} and \eqref{eq:new7.97}, we get
    \begin{align*}
        & \left\|\mathbf{\Lambda}_{\mh} \bm{\tilde{\mu}}_{\mh, 2}(\mathbf{x})^{-1} \mathbf{\Lambda}_{\mh} - \frac{1}{f(\mathbf{x})} \mathbf{\Lambda}_{\mh}^{2} \mathbf{C}_{\mh}^{-1} \right\|_2  \\
        & = \left\| \left( \mathbf{\Lambda}_{\mh} \bm{\tilde{\mu}}_{\mh, 2}(\mathbf{x})^{-1} \mathbf{\Lambda}_{\mh} \right)\cdot \left( \mathbf{\Lambda}_{\mh}^{-2} \mathbf{C}_{\mh} f(\mx) - \mathbf{\Lambda}_{\mh}^{-1} \bm{\tilde{\mu}}_{\mh, 2}(\mathbf{x}) \mathbf{\Lambda}_{\mh}^{-1} \right) \cdot \left(  \frac{1}{f(\mathbf{x})} \mathbf{\Lambda}_{\mh}^{2} \mathbf{C}_{\mh}^{-1}\right) \right\|_2 \\
        & \leq \left\| \left( \mathbf{\Lambda}_{\mh}^{-1} \bm{\tilde{\mu}}_{\mh, 2}(\mathbf{x}) \mathbf{\Lambda}_{\mh}^{-1} \right)^{-1} \right\|_2 \cdot \left\| \mathbf{\Lambda}_{\mh}^{-2} \mathbf{C}_{\mh} f(\mx) - \mathbf{\Lambda}_{\mh}^{-1} \bm{\tilde{\mu}}_{\mh, 2}(\mathbf{x}) \mathbf{\Lambda}_{\mh}^{-1} \right\|_2 \cdot \left\| \frac{1}{f(\mathbf{x})} \mathbf{\Lambda}_{\mh}^{2} \mathbf{C}_{\mh}^{-1} \right\|_2 \\
        & = \left( \rho ( \mh )\left\| \left( \mathbf{\Lambda}_{\mh}^{-1} \bm{\tilde{\mu}}_{\mh, 2}(\mathbf{x}) \mathbf{\Lambda}_{\mh}^{-1} \right)^{-1} \right\|_2 \right) \cdot \left\| \mathbf{\Lambda}_{\mh}^{-2} \mathbf{C}_{\mh} f(\mx) - \mathbf{\Lambda}_{\mh}^{-1} \bm{\tilde{\mu}}_{\mh, 2}(\mathbf{x}) \mathbf{\Lambda}_{\mh}^{-1} \right\|_2  \\
        & \quad \cdot \left\| \frac{\rho ( \mh )}{f(\mathbf{x})} \mathbf{\Lambda}_{\mh}^{2} \mathbf{C}_{\mh}^{-1} \right\|_2 \cdot \rho^{-2} \left( \mh \right)\\
        & = \left( \rho ( \mh )\left\| \left( \mathbf{\Lambda}_{\mh}^{-1} \bm{\tilde{\mu}}_{\mh, 2}(\mathbf{x}) \mathbf{\Lambda}_{\mh}^{-1} \right)^{-1} \right\|_2 \right) \cdot \left\| \mathbf{\Lambda}_{\mh}^{-2} \mathbf{C}_{\mh} f(\mx) - \mathbf{\Lambda}_{\mh}^{-1} \bm{\tilde{\mu}}_{\mh, 2}(\mathbf{x}) \mathbf{\Lambda}_{\mh}^{-1} \right\|_2  \\
        & \quad \cdot\max \left\{ \frac{\rho ( \mh ) h_1^2}{f(\mathbf{x}) c_{\mh, 2\mathbf{e}_{1}, 1}(L)} ,\dots, \frac{\rho ( \mh ) h_d^2}{f(\mathbf{x}) c_{\mh, 2\mathbf{e}_{d}, 1}(L)} \right\} \cdot \rho^{-2} \left( \mh \right) \\
        & \leq O(1) \cdot o\left(\rho ( \mh )\right) \cdot O(1) \cdot \rho^{-2} \left( \mh \right) \\
        & = o\left(  \rho^{-1} \left( \mh \right) \right).
    \end{align*}
    This with \Cref{lemma:1.4} and \eqref{chinvertible} gives \eqref{eq:new4.39}.
\end{proof}

\begin{lemma} \label{lemma:temp4.6}
    Assume that the conditions \ref{con:L1}, \ref{con:B1} and \ref{con:D2} hold. Then, it holds that
    \begin{align} \label{eq:temp4.55}
    \begin{split}
        \left\| \mathbf{\Lambda}_{\mh}^{-1} \bm{\hat{\mu}}_{\mh, 1}(\mathbf{x}) - \mathbf{\Lambda}_{\mh}^{-1} \bm{\tilde{\mu}}_{\mh, 1}(\mathbf{x})  \right\|_2 = O_{\mathbb{P}} \left(n^{-\frac{1}{2}}\rho^{\frac{1}{2}} \left( \mh \right) \right),  \\
        \left\| \mathbf{\Lambda}_{\mh}^{-1} \bm{\hat{\mu}}_{\mh, 2}(\mathbf{x}) \mathbf{\Lambda}_{\mh}^{-1} - \mathbf{\Lambda}_{\mh}^{-1} \bm{\tilde{\mu}}_{\mh, 2}(\mathbf{x}) \mathbf{\Lambda}_{\mh}^{-1} \right\|_2 = O_{\mathbb{P}}\left(n^{-\frac{1}{2}}\rho^{\frac{1}{2}} \left( \mh \right)\right).
    \end{split}
    \end{align}
    If the condition \ref{con:D1} further holds, then $\bm{\hat{\mu}}_{\mh, 2}(\mathbf{x})$ is invertible with probability tending to one and it holds that
    \begin{gather}
        \left\| \left(\mathbf{\Lambda}_{\mh}^{-1} \bm{\hat{\mu}}_{\mh, 2}(\mathbf{x}) \mathbf{\Lambda}_{\mh}^{-1}\right)^{-1} - \left(\mathbf{\Lambda}_{\mh}^{-1} \bm{\tilde{\mu}}_{\mh, 2}(\mathbf{x}) \mathbf{\Lambda}_{\mh}^{-1}\right)^{-1} \right\|_2 = O_{\mathbb{P}} \left( n^{-\frac{1}{2}}\rho^{-\frac{3}{2}} \left( \mh \right) \right). \label{eq:temp4.56}
    \end{gather}
\end{lemma}

\begin{proof} [Proof of \Cref{lemma:temp4.6}]
    We first prove \eqref{eq:temp4.55}. By \Cref{lemma:1.2}, we get
    \begin{align} \begin{split} \label{eq:temp4.58}
        & \mathbb{E} \left[ \mathcal{L}_{\mx, \mh}^2 \left( \mathbf{X} \right) \left( \mathbf{\Lambda}_{\mh}^{-1}  \bm{\Phi}_{\mx}^{-1} \left( \mathbf{X} \right)\right) \tran \left( \mathbf{\Lambda}_{\mh}^{-1}  \bm{\Phi}_{\mx}^{-1} \left( \mathbf{X} \right)\right) \right]  \\
        & = \int_{\mathbb{T}^{d}} \mathcal{L}_{\mx, \mh}^2 \left( \mz \right) \left( \mathbf{\Lambda}_{\mh}^{-1}  \bm{\Phi}_{\mx}^{-1} \left( \mz \right)\right) \tran \left( \mathbf{\Lambda}_{\mh}^{-1}  \bm{\Phi}_{\mx}^{-1} \left( \mz \right)\right) f( \mz ) \omega_d (\mathrm{d} \mz ) \\
        & = \int_{[-\pi, \pi)^{d}}  \bm{L}_{\mh}^2 (\bm{\theta}) \left\| \mathbf{\Lambda}_{\mh}^{-1} \bm{\theta} \right\|_2^2 f \left( \bm{\Phi}_{\mx} \left( \bm{\theta} \right) \right)  \mathrm{d} \bm{\theta}\\ 
        & \leq \sup_{\mz \in \mathbb{T}^{d}} f ( \mz ) \cdot \int_{[-\pi, \pi)^{d}}  \bm{L}_{\mh}^2 (\bm{\theta}) \left\| \mathbf{\Lambda}_{\mh}^{-1} \bm{\theta} \right\|_2^2  \mathrm{d} \bm{\theta}, \\
        & \left\| \mathbb{E} \left[ \mathcal{L}_{\mx, \mh}^2 \left( \mathbf{X} \right) \left( \mathbf{\Lambda}_{\mh}^{-1}  \bm{\Phi}_{\mx}^{-1} \left( \mathbf{X} \right) \bm{\Phi}_{\mx}^{-1} \left( \mathbf{X} \right) \tran \mathbf{\Lambda}_{\mh}^{-1}\right)^{2} \right] \right\|_2 \\
        & = \left\| \int_{\mathbb{T}^{d}} \mathcal{L}_{\mx, \mh}^2 \left( \mz \right)  \left( \mathbf{\Lambda}_{\mh}^{-1}  \bm{\Phi}_{\mx}^{-1} \left( \mz \right) \bm{\Phi}_{\mx}^{-1} \left( \mz \right) \tran \mathbf{\Lambda}_{\mh}^{-1}\right)^{2} f( \mz ) \omega_d (\mathrm{d} \mz ) \right\|_2 \\
        & = \left\| \int_{[-\pi, \pi)^{d}}  \bm{L}_{\mh}^2 (\bm{\theta})  \left( \mathbf{\Lambda}_{\mh}^{-1}  \bm{\theta} \bm{\theta} \tran \mathbf{\Lambda}_{\mh}^{-1}\right)^{2} f \left( \bm{\Phi}_{\mx} \left( \bm{\theta} \right) \right) \mathrm{d} \bm{\theta} \right\|_2 \\
        & \leq \int_{[-\pi, \pi)^{d}}  \bm{L}_{\mh}^2 (\bm{\theta})  \left\| \mathbf{\Lambda}_{\mh}^{-1}  \bm{\theta} \bm{\theta} \tran \mathbf{\Lambda}_{\mh}^{-1} \right\|_2^2 f \left( \bm{\Phi}_{\mx} \left( \bm{\theta} \right) \right) \mathrm{d} \bm{\theta} \\
        & \leq \sup_{\mz \in \mathbb{T}^{d}} f ( \mz ) \cdot \int_{[-\pi, \pi)^{d}}  \bm{L}_{\mh}^2 (\bm{\theta})  \left\| \mathbf{\Lambda}_{\mh}^{-1}  \bm{\theta} \right\|_2^4 \mathrm{d} \bm{\theta}.
    \end{split} \end{align}
    Combining \eqref{mu1bound_step3}, \eqref{mu2bound_step2}, \eqref{eq:temp4.58} and Lemma C.2 of Im et al. (2025), we get \eqref{eq:temp4.55}.

    Now, we prove the invertibility of $\bm{\hat{\mu}}_{\mh, 2}(\mathbf{x})$. By \eqref{eq:temp4.55}, we get
    \begin{align} \label{eq:temp4.63}
        \left\| \mathbf{\Lambda}_{\mh}^{-1} \bm{\hat{\mu}}_{\mh, 2}(\mathbf{x}) \mathbf{\Lambda}_{\mh}^{-1} - \mathbf{\Lambda}_{\mh}^{-1} \bm{\tilde{\mu}}_{\mh, 2}(\mathbf{x}) \mathbf{\Lambda}_{\mh}^{-1} \right\|_2 = o_{\mathbb{P}}\left(\rho ( \mh )\right).
    \end{align}
    We define an event $F_n$ as
    \begin{align*}
        F_n := \bigg\{ \left\| \mathbf{\Lambda}_{\mh}^{-1} \bm{\hat{\mu}}_{\mh, 2}(\mathbf{x}) \mathbf{\Lambda}_{\mh}^{-1} - \mathbf{\Lambda}_{\mh}^{-1} \bm{\tilde{\mu}}_{\mh, 2}(\mathbf{x}) \mathbf{\Lambda}_{\mh}^{-1} \right\|_2 \leq  \rho ( \mh ) 2^{\frac{3d-2}{2}} a_{0, 1}^{d-1}(L) a_{2, 1}(L) f(\mx) \bigg\}.
    \end{align*}
    Let $\epsilon>0$ be any given constant. By \eqref{eq:temp4.63}, we can choose $N_{1,\epsilon} \in \mathbb{N}$ such that $\mathbb{P}( F_n )>1-\epsilon$ whenever $n \ge N_{1,\epsilon}$. By \eqref{eq:new7.96}, we can also choose $N_{2,\epsilon} \in \mathbb{N}$ such that
    \begin{align*}
        \rho^{-1} \left( \mh \right)  \left( \inf_{\bu \in \mathbb{S}^{d}} \bu\tran\mathbf{\Lambda}_{\mh}^{-1} \bm{\tilde{\mu}}_{\mh, 2}(\mathbf{x}) \mathbf{\Lambda}_{\mh}^{-1} \bu \right) \geq 2^{\frac{3d}{2}} a_{0, 1}^{d-1}(L) a_{2, 1}(L) f(\mx)
    \end{align*}
    whenever $n \geq N_{2,\epsilon}$. On the event $F_n$ with $n \geq \max \{N_{1,\epsilon}, N_{2,\epsilon} \}$, it holds that
    \begin{align} \begin{split} \label{eq:temp4.65}
        & \inf_{\bu \in \mathbb{S}^{d}} \bu\tran\mathbf{\Lambda}_{\mh}^{-1} \bm{\hat{\mu}}_{\mh, 2}(\mathbf{x}) \mathbf{\Lambda}_{\mh}^{-1} \bu \\
        & \geq \inf_{\bu \in \mathbb{S}^{d}} \bu\tran\mathbf{\Lambda}_{\mh}^{-1} \bm{\tilde{\mu}}_{\mh, 2}(\mathbf{x}) \mathbf{\Lambda}_{\mh}^{-1} \bu - \sup_{\bu \in \mathbb{S}^{d}} \bu\tran\left( \mathbf{\Lambda}_{\mh}^{-1} \bm{\tilde{\mu}}_{\mh, 2}(\mathbf{x}) \mathbf{\Lambda}_{\mh}^{-1} - \mathbf{\Lambda}_{\mh}^{-1} \bm{\hat{\mu}}_{\mh, 2}(\mathbf{x}) \mathbf{\Lambda}_{\mh}^{-1} \right) \bu \\
        & = \inf_{\bu \in \mathbb{S}^{d}} \bu\tran\mathbf{\Lambda}_{\mh}^{-1} \bm{\tilde{\mu}}_{\mh, 2}(\mathbf{x}) \mathbf{\Lambda}_{\mh}^{-1} \bu - \left\| \mathbf{\Lambda}_{\mh}^{-1} \bm{\hat{\mu}}_{\mh, 2}(\mathbf{x}) \mathbf{\Lambda}_{\mh}^{-1} - \mathbf{\Lambda}_{\mh}^{-1} \bm{\tilde{\mu}}_{\mh, 2}(\mathbf{x}) \mathbf{\Lambda}_{\mh}^{-1} \right\|_2 \\
        & \geq \rho ( \mh ) 2^{\frac{3d}{2}} a_{0, 1}^{d-1}(L) a_{2, 1}(L) f(\mx) - \rho ( \mh ) 2^{\frac{3d-2}{2}} a_{0, 1}^{d-1}(L) a_{2, 1}(L) f(\mx) \\
        & = \rho ( \mh ) 2^{\frac{3d-2}{2}} a_{0, 1}^{d-1}(L) a_{2, 1}(L) f(\mx) \\
        & > 0.
    \end{split} \end{align}
    Hence, for $n \geq \max \{ N_{1,\epsilon}, N_{2,\epsilon}\}$, it holds that
    \begin{align*}
        \mathbb{P}\left(\inf_{\bu \in \mathbb{S}^{d}} \bu\tran\mathbf{\Lambda}_{\mh}^{-1} \bm{\hat{\mu}}_{\mh, 2}(\mathbf{x}) \mathbf{\Lambda}_{\mh}^{-1} \bu > 0\right)\geq \mathbb{P}(F_n) >1 - \epsilon.
    \end{align*}
    Since $\epsilon>0$ is arbitrary, this shows that $\mathbf{\Lambda}_{\mh}^{-1} \bm{\hat{\mu}}_{\mh, 2}(\mathbf{x}) \mathbf{\Lambda}_{\mh}^{-1}$ is positive-definite with probability tending to one. Thus, $\bm{\hat{\mu}}_{\mh, 2}(\mathbf{x})$ is invertible with probability tending to one.
    
    Now, we prove \eqref{eq:temp4.56}. By \eqref{eq:new7.97}, we get
    \begin{align} \label{mu2tildeinverse}
        \left\| \left( \mathbf{\Lambda}_{\mh}^{-1} \bm{\tilde{\mu}}_{\mh, 2}(\mathbf{x}) \mathbf{\Lambda}_{\mh}^{-1} \right)^{-1} \right\|_2 = O \left( \rho^{-1} \left( \mh \right) \right).
    \end{align}
    By \eqref{eq:temp4.65}, we also get
    \begin{align*}
        \rho ( \mh )  \left\| \left( \mathbf{\Lambda}_{\mh}^{-1} \bm{\hat{\mu}}_{\mh, 2}(\mathbf{x}) \mathbf{\Lambda}_{\mh}^{-1} \right)^{-1} \right\|_2
        & = \rho ( \mh ) \left( \inf_{\bu \in \mathbb{S}^{d}} \bu\tran\mathbf{\Lambda}_{\mh}^{-1} \bm{\hat{\mu}}_{\mh, 2}(\mathbf{x}) \mathbf{\Lambda}_{\mh}^{-1} \bu\right)^{-1} \\
        & = \left( \rho^{-1} \left( \mh \right)  \left( \inf_{\bu \in \mathbb{S}^{d}} \bu\tran\mathbf{\Lambda}_{\mh}^{-1} \bm{\hat{\mu}}_{\mh, 2}(\mathbf{x}) \mathbf{\Lambda}_{\mh}^{-1} \bu \right)\right)^{-1} \\
        & \leq 2^{-\frac{3d-2}{2}} \frac{1}{ a_{0, 1}^{d-1}(L) a_{2, 1}(L) f(\mx)}
    \end{align*}
    on the event $F_n$ with $n \geq \max \{N_{1,\epsilon}, N_{2,\epsilon} \}$.
    This implies that
    \begin{align*}
        \liminf_{n \to \infty} \mathbb{P}\left( \rho ( \mh )\left\| \left( \mathbf{\Lambda}_{\mh}^{-1} \bm{\hat{\mu}}_{\mh, 2}(\mathbf{x}) \mathbf{\Lambda}_{\mh}^{-1} \right)^{-1} \right\|_2 \le  2^{-\frac{3d-2}{2}} \frac{1}{ a_{0, 1}^{d-1}(L) a_{2, 1}(L) f(\mx)}\right)\geq \liminf_{n \to \infty} \mathbb{P}(F_n) \ge 1 - \epsilon.
    \end{align*}
    Since $\epsilon>0$ is arbitrary, this shows that
    \begin{align} \label{mu2hatinverse}
        \left\| \left( \mathbf{\Lambda}_{\mh}^{-1} \bm{\hat{\mu}}_{\mh, 2}(\mathbf{x}) \mathbf{\Lambda}_{\mh}^{-1} \right)^{-1} \right\|_2 = O_{\mathbb{P}} \left( \rho^{-1} \left( \mh \right) \right).
    \end{align}
    Combining \eqref{eq:temp4.55}, \eqref{mu2tildeinverse} and \eqref{mu2hatinverse}, we get
    \begin{align*}
        & \left\| \left(\mathbf{\Lambda}_{\mh}^{-1} \bm{\hat{\mu}}_{\mh, 2}(\mathbf{x}) \mathbf{\Lambda}_{\mh}^{-1}\right)^{-1} - \left(\mathbf{\Lambda}_{\mh}^{-1} \bm{\tilde{\mu}}_{\mh, 2}(\mathbf{x}) \mathbf{\Lambda}_{\mh}^{-1}\right)^{-1} \right\|_2  \\
        & \leq  \left\| \left( \mathbf{\Lambda}_{\mh}^{-1} \bm{\tilde{\mu}}_{\mh, 2}(\mathbf{x}) \mathbf{\Lambda}_{\mh}^{-1} \right)^{-1} \right\|_2 \cdot \left\| \mathbf{\Lambda}_{\mh}^{-1} \bm{\hat{\mu}}_{\mh, 2}(\mathbf{x}) \mathbf{\Lambda}_{\mh}^{-1} - \mathbf{\Lambda}_{\mh}^{-1} \bm{\tilde{\mu}}_{\mh, 2}(\mathbf{x}) \mathbf{\Lambda}_{\mh}^{-1} \right\|_2 \cdot \left\| \left( \mathbf{\Lambda}_{\mh}^{-1} \bm{\hat{\mu}}_{\mh, 2}(\mathbf{x}) \mathbf{\Lambda}_{\mh}^{-1} \right)^{-1} \right\|_2  \\
        & \leq O \left( \rho^{-1} \left( \mh \right) \right) \cdot O_{\mathbb{P}}\left(n^{-\frac{1}{2}}\rho^{\frac{1}{2}} \left( \mh \right)\right) \cdot O_{\mathbb{P}} \left( \rho^{-1} \left( \mh \right) \right) \\
        & = O_{\mathbb{P}}\left(n^{-\frac{1}{2}}\rho^{-\frac{3}{2}} \left( \mh \right)\right).
    \end{align*}
    This completes the proof.
\end{proof}

Hereafter, we assume that $\bm{\tilde{\mu}}_{\mh, 2}(\mathbf{x})$ and $\bm{\hat{\mu}}_{\mh, 2}(\mathbf{x})$ are invertible without loss of generality. Below, we prove that $\tilde{\sigma}_{\mh}(\mathbf{x})>0$ for sufficiently large $n$ and that $\hat{\sigma}_{\mh}(\mathbf{x})>0$ with probability tending to one.

\begin{lemma} \label{lemma:temp4.7}
    Assume that the conditions \ref{con:L1}, \ref{con:D1} and \ref{con:D2} hold and that $\lim_{n\rightarrow\infty} \| \mh \|_2=0$. Then, 
    \begin{gather}
        \frac{1}{c_{\mh, \bm{0}_d, 1}(L)} \tilde{\sigma}_{\mh}(\mathbf{x}) - f(\mathbf{x})  = o \left(1 \right). \label{eq:temp4.68}
    \end{gather}
    If the condition $\lim_{n\rightarrow\infty}n\rho ( \mh )=\infty$ further holds, then it holds that
    \begin{gather}
        \frac{1}{c_{\mh, \bm{0}_d, 1}(L)} \left( \hat{\sigma}_{\mh}(\mathbf{x}) - \tilde{\sigma}_{\mh}(\mathbf{x}) \right) = O_{\mathbb{P}} \left(n^{-\frac{1}{2}}\rho^{-\frac{1}{2}} \left( \mh \right) \right). \label{eq:temp4.69}
    \end{gather}
\end{lemma}

\begin{proof} [Proof of \Cref{lemma:temp4.7}]
    \Cref{lemma:1.4} and \Cref{lemma:new2.2} imply that
    \begin{align} \label{eq:4.51}
        \frac{1}{c_{\mh, \bm{0}_d, 1}(L)} \tilde{\mu}_{\mh, 0}(\mathbf{x}) - f(\mathbf{x})  = o(1). 
    \end{align}
    Combining \Cref{lemma:1.4}, \Cref{lemma:new4.4} and \eqref{mu2tildeinverse}, we get
    \begin{align} \begin{split} \label{eq:4.53}
        & \left| \frac{1}{c_{\mh, \bm{0}_d, 1}(L)} \bm{\tilde{\mu}}_{\mh, 1}(\mathbf{x})\tran  \bm{\tilde{\mu}}_{\mh, 2}(\mathbf{x})^{-1} \bm{\tilde{\mu}}_{\mh, 1}(\mathbf{x}) \right| \\
        & = \left| \frac{1}{c_{\mh, \bm{0}_d, 1}(L)} \left( \mathbf{\Lambda}_{\mh}^{-1} \bm{\tilde{\mu}}_{\mh, 1}(\mathbf{x}) \right) \tran \left(\mathbf{\Lambda}_{\mh}^{-1} \bm{\tilde{\mu}}_{\mh, 2}(\mathbf{x}) \mathbf{\Lambda}_{\mh}^{-1}\right)^{-1}\left( \mathbf{\Lambda}_{\mh}^{-1} \bm{\tilde{\mu}}_{\mh, 1}(\mathbf{x})\right) \right| \\
        & \leq \frac{1}{c_{\mh, \bm{0}_d, 1}(L)} \cdot \left\| \mathbf{\Lambda}_{\mh}^{-1} \bm{\tilde{\mu}}_{\mh, 1}(\mathbf{x}) \right\|_2^2 \cdot \left\| \left( \mathbf{\Lambda}_{\mh}^{-1} \bm{\tilde{\mu}}_{\mh, 2}(\mathbf{x}) \mathbf{\Lambda}_{\mh}^{-1} \right)^{-1} \right\|_2  \\
        & = O \left( \rho^{-1} \left( \mh \right) \right) \cdot  o\left(  \rho ( \mh ) \right)^2 \cdot O \left( \rho^{-1} \left( \mh \right) \right) \\
        & = o(1).
    \end{split} \end{align}
    Combining \eqref{eq:4.51} and \eqref{eq:4.53}, we get \eqref{eq:temp4.68}.

    Now, we prove \eqref{eq:temp4.69}. Combining \Cref{lemma:1.4} and \eqref{eq:3.35}, we get
    \begin{align} \label{eq:temp4.73}
        \bigg| \frac{1}{c_{\mh, \bm{0}_d, 1}(L)} \left( \hat{\mu}_{\mh, 0}(\mathbf{x}) - \tilde{\mu}_{\mh, 0}(\mathbf{x}) \right) \bigg| = O_{\mathbb{P}}(n^{-\frac{1}{2}}h^{-\frac{d}{2}}). 
    \end{align}
    Note that
    \begin{align} \begin{split} \label{eq:4.54}
        & \bm{\hat{\mu}}_{\mh, 1}(\mathbf{x})\tran  \bm{\hat{\mu}}_{\mh, 2}(\mathbf{x})^{-1} \bm{\hat{\mu}}_{\mh, 1}(\mathbf{x}) - \bm{\tilde{\mu}}_{\mh, 1}(\mathbf{x})\tran  \bm{\tilde{\mu}}_{\mh, 2}(\mathbf{x})^{-1} \bm{\tilde{\mu}}_{\mh, 1}(\mathbf{x}) \\
        & = \left( \mathbf{\Lambda}_{\mh}^{-1} \bm{\hat{\mu}}_{\mh, 1}(\mathbf{x}) \right) \tran \left(\mathbf{\Lambda}_{\mh}^{-1} \bm{\hat{\mu}}_{\mh, 2}(\mathbf{x}) \mathbf{\Lambda}_{\mh}^{-1}\right)^{-1}\left( \mathbf{\Lambda}_{\mh}^{-1} \bm{\hat{\mu}}_{\mh, 1}(\mathbf{x})\right) \\
        & \quad- \left( \mathbf{\Lambda}_{\mh}^{-1} \bm{\tilde{\mu}}_{\mh, 1}(\mathbf{x}) \right) \tran \left(\mathbf{\Lambda}_{\mh}^{-1} \bm{\tilde{\mu}}_{\mh, 2}(\mathbf{x}) \mathbf{\Lambda}_{\mh}^{-1}\right)^{-1}\left( \mathbf{\Lambda}_{\mh}^{-1} \bm{\tilde{\mu}}_{\mh, 1}(\mathbf{x})\right) \\
        & = \left( \mathbf{\Lambda}_{\mh}^{-1} \bm{\hat{\mu}}_{\mh, 1}(\mathbf{x}) \right) \tran \left( \left(\mathbf{\Lambda}_{\mh}^{-1} \bm{\hat{\mu}}_{\mh, 2}(\mathbf{x}) \mathbf{\Lambda}_{\mh}^{-1}\right)^{-1} - \left(\mathbf{\Lambda}_{\mh}^{-1} \bm{\tilde{\mu}}_{\mh, 2}(\mathbf{x}) \mathbf{\Lambda}_{\mh}^{-1}\right)^{-1} \right) \left( \mathbf{\Lambda}_{\mh}^{-1} \bm{\hat{\mu}}_{\mh, 1}(\mathbf{x})\right) \\
        & \quad + \left( \mathbf{\Lambda}_{\mh}^{-1} \bm{\hat{\mu}}_{\mh, 1}(\mathbf{x}) - \mathbf{\Lambda}_{\mh}^{-1} \bm{\tilde{\mu}}_{\mh, 1}(\mathbf{x}) \right) \tran \left(\mathbf{\Lambda}_{\mh}^{-1} \bm{\tilde{\mu}}_{\mh, 2}(\mathbf{x}) \mathbf{\Lambda}_{\mh}^{-1}\right)^{-1}\left( \mathbf{\Lambda}_{\mh}^{-1} \bm{\hat{\mu}}_{\mh, 1}(\mathbf{x}) + \mathbf{\Lambda}_{\mh}^{-1} \bm{\tilde{\mu}}_{\mh, 1}(\mathbf{x})\right).
    \end{split} \end{align}    
    \Cref{lemma:new4.4} and \Cref{lemma:temp4.6} imply that
    \begin{align} \begin{split} \label{mu1hat}
        \left\| \mathbf{\Lambda}_{\mh}^{-1} \bm{\hat{\mu}}_{\mh, 1}(\mathbf{x})  \right\|_2 &\leq \left\| \mathbf{\Lambda}_{\mh}^{-1} \bm{\tilde{\mu}}_{\mh, 1}(\mathbf{x})  \right\|_2 + \left\| \mathbf{\Lambda}_{\mh}^{-1} \bm{\hat{\mu}}_{\mh, 1}(\mathbf{x}) - \mathbf{\Lambda}_{\mh}^{-1} \bm{\tilde{\mu}}_{\mh, 1}(\mathbf{x})  \right\|_2 \\
        & = o\left(\rho ( \mh ) \right) + O_{\mathbb{P}} \left(n^{-\frac{1}{2}}\rho^{\frac{1}{2}} \left( \mh \right) \right) \\
        & = o_{\mathbb{P}} \left( \rho ( \mh ) \right).
    \end{split}\end{align}
    Combining \Cref{lemma:new4.4}, \Cref{lemma:temp4.6}, \eqref{mu2tildeinverse}, \eqref{eq:4.54} and \eqref{mu1hat}, we get
    \begin{align*}
        & \left| \bm{\hat{\mu}}_{\mh, 1}(\mathbf{x})\tran  \bm{\hat{\mu}}_{\mh, 2}(\mathbf{x})^{-1} \bm{\hat{\mu}}_{\mh, 1}(\mathbf{x}) - \bm{\tilde{\mu}}_{\mh, 1}(\mathbf{x})\tran  \bm{\tilde{\mu}}_{\mh, 2}(\mathbf{x})^{-1} \bm{\tilde{\mu}}_{\mh, 1}(\mathbf{x}) \right| \\
        & \leq \left\| \mathbf{\Lambda}_{\mh}^{-1} \bm{\hat{\mu}}_{\mh, 1}(\mathbf{x})  \right\|_2^2 \cdot \left\| \left(\mathbf{\Lambda}_{\mh}^{-1} \bm{\hat{\mu}}_{\mh, 2}(\mathbf{x}) \mathbf{\Lambda}_{\mh}^{-1}\right)^{-1} - \left(\mathbf{\Lambda}_{\mh}^{-1} \bm{\tilde{\mu}}_{\mh, 2}(\mathbf{x}) \mathbf{\Lambda}_{\mh}^{-1}\right)^{-1} \right\|_2 \\
        & \quad + \left\| \mathbf{\Lambda}_{\mh}^{-1} \bm{\hat{\mu}}_{\mh, 1}(\mathbf{x}) - \mathbf{\Lambda}_{\mh}^{-1} \bm{\tilde{\mu}}_{\mh, 1}(\mathbf{x})  \right\|_2 \cdot \left\| \left( \mathbf{\Lambda}_{\mh}^{-1} \bm{\tilde{\mu}}_{\mh, 2}(\mathbf{x}) \mathbf{\Lambda}_{\mh}^{-1} \right)^{-1} \right\|_2 \\ 
        & \qquad \cdot \left( \left\| \mathbf{\Lambda}_{\mh}^{-1} \bm{\hat{\mu}}_{\mh, 1}(\mathbf{x}) \right\|_2 + \left\| \mathbf{\Lambda}_{\mh}^{-1} \bm{\tilde{\mu}}_{\mh, 1}(\mathbf{x})  \right\|_2\right) \\
        & = o_{\mathbb{P}} \left( \rho ( \mh ) \right)^2 \cdot O_{\mathbb{P}} \left( n^{-\frac{1}{2}}\rho^{-\frac{3}{2}} \left( \mh \right) \right) + O_{\mathbb{P}} \left( n^{-\frac{1}{2}}\rho^{\frac{1}{2}} \left( \mh \right) \right) \cdot O \left( \rho^{-1} \left( \mh \right) \right) \cdot \left(  o_{\mathbb{P}} \left( \rho ( \mh ) \right) +  o \left( \rho ( \mh ) \right) \right) \\
        & = o_{\mathbb{P}} \left( n^{-\frac{1}{2}}\rho^{\frac{1}{2}} \left( \mh \right) \right).
    \end{align*}
   This with \Cref{lemma:1.4} and \eqref{eq:temp4.73} gives \eqref{eq:temp4.69}.
\end{proof}

Hereafter, we assume that $\tilde{\sigma}_{\mh}(\mathbf{x})>0$ and $\hat{\sigma}_{\mh}(\mathbf{x})>0$ without loss of generality.

\begin{lemma} \label{lemma:new4.8}
    Assume that the conditions \ref{con:L1} and \ref{con:D1}--\ref{con:D3} hold and that $\lim_{n\rightarrow\infty} \| \mh \|_2=0$. Then, it holds that
    \begin{align*}
        \sup_{y \in \mathbb{M}} \left| \frac{\tilde{\tau}_{\mh, 0}(\mathbf{x}, y) - \bm{\tilde{\mu}}_{\mh, 1}(\mathbf{x})\tran\bm{\tilde{\mu}}_{\mh, 2}(\mathbf{x})^{-1} \bm{\tilde{\tau}}_{\mh, 1}(\mathbf{x}, y)}{\tilde{\sigma}_{\mh}(\mathbf{x})} - g_y(\mathbf{x})  \right| = o(1).
    \end{align*}
\end{lemma}

\begin{proof} [Proof of Lemma \ref{lemma:new4.8}]
    We first prove that
    \begin{align} \label{tau1bound}
        \sup_{y \in \mathbb{M}} \left\| \mathbf{\Lambda}_{\mh}^{-1} \bm{\tilde{\tau}}_{\mh, 1}(\mathbf{x}, y) \right\|_2 = o\left( \rho ( \mh ) \right).
    \end{align}
    The Cauchy-Schwarz inequality and \eqref{eq:new4.43} imply that
    \begin{align} \begin{split} \label{tau1bound_step1}
        \big\| \mathbf{\Lambda}_{\mh}^{-1} \bm{\tilde{\tau}}_{\mh, 1}(\mathbf{x}, y) \big\|_{2}^{2} & = \bigg\| \mathbf{\Lambda}_{\mh}^{-1} \int_{\mathbb{T}^{d}} \mathcal{L}_{\mx, \mh} \left( \mz \right) \bm{\Phi}_{\mx}^{-1} (\mz)  \big( (f \cdot g_y ) (\mz) - (f \cdot g_y ) (\mx) \big) \omega_{1}^{d}(\mathrm{d}\mz)   \bigg\|_{2}^{2}  \\
        & \leq \left( \int_{[-\pi, \pi)^{d}}  \bm{L}_{\mh} (\bm{\theta})    \left|  (f \cdot g_y ) \left( \bm{\Phi}_{\mx} \left( \bm{\theta} \right) \right) - (f \cdot g_y )(\mx)\right| \mathrm{d} \bm{\theta} \right) \\
        & \quad \cdot 2 \sup_{\mz \in \mathbb{T}^{d}} (f \cdot g_y ) ( \mz ) \cdot  \left( \int_{[-\pi, \pi)^{d}}  \bm{L}_{\mh} (\bm{\theta}) \left\| \mathbf{\Lambda}_{\mh}^{-1} \bm{\theta} \right\|_2  \mathrm{d} \bm{\theta} \right)
    \end{split} \end{align}
    for any $y \in \mathbb{M}$. By arguing as in the proofs of \Cref{lemma:new2.2} and \Cref{lemma:new3.2}, we get
    \begin{align} \label{tau1bound_step2}
        \int_{[-\pi, \pi)^{d}}  \bm{L}_{\mh} (\bm{\theta})    \sup_{y \in \mathbb{M}}   \left|  (f \cdot g_y) \left( \bm{\Phi}_{\mx} \left( \bm{\theta} \right) \right) - (f \cdot g_y) (\mx)\right| \mathrm{d} \bm{\theta} = o\left(\rho(\mh)\right).
    \end{align}
    Combining \eqref{mu1bound_step3}, \eqref{tau1bound_step1} and \eqref{tau1bound_step2}, we get \eqref{tau1bound}. Note that
    \begin{align} \begin{split} \label{eq:new7.124}
        & \tilde{\tau}_{\mh, 0}(\mathbf{x}, y) - \bm{\tilde{\mu}}_{\mh, 1}(\mathbf{x})\tran\bm{\tilde{\mu}}_{\mh, 2}(\mathbf{x})^{-1} \bm{\tilde{\tau}}_{\mh, 1}(\mathbf{x}, y) - g_y(\mathbf{x}) \tilde{\sigma}_{\mh}(\mathbf{x}) \\
        & = \left( \tilde{\tau}_{\mh, 0}(\mathbf{x}, y) - g_y(\mathbf{x}) \tilde{\mu}_{\mh, 0}(\mathbf{x}) \right) - \bm{\tilde{\mu}}_{\mh, 1}(\mathbf{x})\tran\bm{\tilde{\mu}}_{\mh, 2}(\mathbf{x})^{-1} \left( \bm{\tilde{\tau}}_{\mh, 1}(\mathbf{x}, y) - g_y(\mathbf{x}) \bm{\tilde{\mu}}_{\mh, 1}(\mathbf{x}) \right) \\
        & = \left( \tilde{\tau}_{\mh, 0}(\mathbf{x}, y) - g_y(\mathbf{x}) \tilde{\mu}_{\mh, 0}(\mathbf{x}) \right) \\
        & \quad - \left( \mathbf{\Lambda}_{\mh}^{-1} \bm{\tilde{\mu}}_{\mh, 1}(\mathbf{x}) \right) \tran \left(\mathbf{\Lambda}_{\mh}^{-1} \bm{\tilde{\mu}}_{\mh, 2}(\mathbf{x}) \mathbf{\Lambda}_{\mh}^{-1}\right)^{-1}\left( \mathbf{\Lambda}_{\mh}^{-1} \bm{\tilde{\mu}}_{\mh, 1}(\mathbf{x}) - g_y(\mathbf{x}) \mathbf{\Lambda}_{\mh}^{-1} \bm{\tilde{\tau}}_{\mh, 1}(\mathbf{x}, y) \right)
    \end{split} \end{align}  
    for any $y \in \mathbb{M}$. Combining \Cref{lemma:new4.4}, \eqref{eq:new7.54}, \eqref{mu2tildeinverse}, \eqref{tau1bound} and \eqref{eq:new7.124}, we get
    \begin{align} \begin{split} \label{eq:new7.125}
        & \sup_{y \in \mathbb{M}} \left| \tilde{\tau}_{\mh, 0}(\mathbf{x}, y) - \bm{\tilde{\mu}}_{\mh, 1}(\mathbf{x})\tran\bm{\tilde{\mu}}_{\mh, 2}(\mathbf{x})^{-1} \bm{\tilde{\tau}}_{\mh, 1}(\mathbf{x}, y) - g_y(\mathbf{x}) \tilde{\sigma}_{\mh}(\mathbf{x}) \right| \\
        & \leq \sup_{y \in \mathbb{M}} \left| \tilde{\tau}_{\mh, 0}(\mathbf{x}, y) - g_y(\mathbf{x}) \tilde{\mu}_{\mh, 0}(\mathbf{x}) \right| + \left\| \mathbf{\Lambda}_{\mh}^{-1} \bm{\tilde{\mu}}_{\mh, 1}(\mathbf{x}) \right\|_2 \cdot \left\| \left( \mathbf{\Lambda}_{\mh}^{-1} \bm{\tilde{\mu}}_{\mh, 2}(\mathbf{x}) \mathbf{\Lambda}_{\mh}^{-1} \right)^{-1} \right\|_2 \\
        & \quad \cdot \left( \left\| \mathbf{\Lambda}_{\mh}^{-1} \bm{\tilde{\mu}}_{\mh, 1}(\mathbf{x}) \right\|_2 + \sup_{y \in \mathbb{M}}g_y(\mathbf{x}) \cdot \sup_{y \in \mathbb{M}} \left\| \mathbf{\Lambda}_{\mh}^{-1} \bm{\tilde{\tau}}_{\mh, 1}(\mathbf{x}, y) \right\|_{2}\right) \\
        & = o\left( \rho ( \mh ) \right) + o\left( \rho ( \mh ) \right) \cdot O\left( \rho^{-1} \left( \mh \right) \right) \cdot \left( o\left( \rho ( \mh ) \right) + o\left( \rho ( \mh ) \right) \right) \\
        & = o\left( \rho ( \mh ) \right).
    \end{split} \end{align}
    Also, \Cref{lemma:1.4} and \eqref{eq:temp4.68} imply that
    \begin{align} \begin{split} \label{sigmatildeinverse}
        \lim_{n \to \infty} \frac{\rho ( \mh )}{\tilde{\sigma}_{\mh}(\mathbf{x})} = 2^{-\frac{3d}{2}} \frac{1}{ a_{0, 1}^{d}(L) f(\mx)} > 0.
    \end{split} \end{align}
    Combining \eqref{eq:new7.125} and \eqref{sigmatildeinverse}, we get the desired result.
\end{proof}

\begin{lemma} \label{lemma:new4.9}
    Assume that the conditions \ref{con:L1} and \ref{con:D1}--\ref{con:D3} hold and that $\lim_{n\rightarrow\infty} \| \mh \|_2=0$. Then, it holds that
    \begin{align*}
        \sup_{y \in \mathbb{M}} \big|  \tilde{M}_{\mh, 1}(\mathbf{x}, y) - M_{\oplus}(\mathbf{x}, y)   \big| = o(1).
    \end{align*}
\end{lemma}

\begin{proof} [Proof of \Cref{lemma:new4.9}]
    The lemma follows by arguing as in the proof of Lemma C.6 of Im et al. (2025).
\end{proof}

\begin{lemma} \label{thm:4.1}
   Assume that the conditions \ref{con:L1}, \ref{con:D1}--\ref{con:D3} and \ref{con:M1} hold and that $\lim_{n\rightarrow\infty} \| \mh \|_2=0$. Then, it holds that
    \begin{align*}
        d_{\mathbb{M}}(\tilde{m}_{\mh, 1}(\mathbf{x}), m_{\oplus}(\mathbf{x})) = o(1).
    \end{align*}
\end{lemma}

\begin{proof} [Proof of \Cref{thm:4.1}]
    The lemma follows by arguing as in the proof of Lemma C.7 of Im et al. (2025) and using \Cref{lemma:new4.9}.
\end{proof}

\begin{lemma} \label{lemma:4.10}
    Assume that the conditions \ref{con:L1}, \ref{con:B1}, \ref{con:D1} and \ref{con:D2} hold. Then, it holds that
    \begin{align*}
        \sup_{y \in \mathbb{M}} \big| \hat{M}_{\mh, 1}(\mathbf{x}, y) - \tilde{M}_{\mh, 1}(\mathbf{x}, y) \big| = o_{\mathbb{P}}(1).
    \end{align*}
\end{lemma}

\begin{proof} [Proof of \Cref{lemma:4.10}]
    We first prove that
    \begin{align} \label{lemma:new3.5}
        \left| \hat{M}_{\mh, 1}(\mathbf{x}, y) - \tilde{M}_{\mh, 1}(\mathbf{x}, y) \right| = o_{\mathbb{P}}(1)
    \end{align}
    for any $y \in \mathbb{M}$. By \eqref{eq:4.80}, we get
    \begin{align} \begin{split} \label{eq:new7.132}
        & \left| \hat{M}_{\mh, 1}(\mathbf{x}, y) - \tilde{M}_{\mh, 1}(\mathbf{x}, y) \right| \\
        & \le \left| \frac{1}{\hat{\sigma}_{\mh}(\mathbf{x})} \right| \cdot \left| \hat{\nu}_{\mh, 0}(\mathbf{x}, y) - \tilde{\nu}_{\mh, 0}(\mathbf{x}, y)\right| \\
        & \quad + \left| \frac{1}{\hat{\sigma}_{\mh}(\mathbf{x})} \right| \cdot \left| \bm{\hat{\mu}}_{\mh, 1}(\mathbf{x})\tran\bm{\hat{\mu}}_{\mh, 2}(\mathbf{x})^{-1} \bm{\hat{\nu}}_{\mh, 1}(\mathbf{x}, y) - \bm{\tilde{\mu}}_{\mh, 1}(\mathbf{x})\tran\bm{\tilde{\mu}}_{\mh, 2}(\mathbf{x})^{-1} \bm{\tilde{\nu}}_{\mh, 1}(\mathbf{x}, y)\right| \\
        & \quad + \left| \frac{1}{\tilde{\sigma}_{\mh}(\mathbf{x})} \right| \cdot \left| \frac{1}{\hat{\sigma}_{\mh}(\mathbf{x})} \right| \cdot \left| \hat{\sigma}_{\mh}(\mathbf{x}) -\tilde{\sigma}_{\mh}(\mathbf{x}) \right|  \cdot \left| \tilde{\nu}_{\mh, 0}(\mathbf{x}, y) - \bm{\tilde{\mu}}_{\mh, 1}(\mathbf{x})\tran\bm{\tilde{\mu}}_{\mh, 2}(\mathbf{x})^{-1} \bm{\tilde{\nu}}_{\mh, 1}(\mathbf{x}, y) \right|.
    \end{split} \end{align}
    By \eqref{eq:3.35}, we also get
    \begin{align}
        \big| \hat{\nu}_{\mh, 0}(\mathbf{x}, y) - \tilde{\nu}_{\mh, 0}(\mathbf{x}, y) \big| = O_{\mathbb{P}}\left(n^{-\frac{1}{2}} \rho^{\frac{1}{2}} \left( \mh \right)\right). \label{eq:4.109}
    \end{align}
    Similarly to \eqref{eq:temp4.58}, we get
    \begin{align} \begin{split} \label{eq:4.110}
         \left\| \mathbf{\Lambda}_{\mh}^{-1} \bm{\tilde{\nu}}_{\mh, 1}(\mathbf{x}, y)  \right\|_2^2 & \leq  \left\| \mathbb{E} \left[ \mathcal{L}_{\mx, \mh} \left( \mathbf{X} \right) \mathbf{\Lambda}_{\mh}^{-1} \bm{\Phi}_{\mx}^{-1} (\mathbf{X}) d_{\mathbb{M}}^2(Y, y) \right] \right\|_2^2 \\
         & \leq  \mathbb{E} \left[ \mathcal{L}_{\mx, \mh}^2 \left( \mathbf{X} \right) \left\| \mathbf{\Lambda}_{\mh}^{-1} \bm{\Phi}_{\mx}^{-1} (\mathbf{X}) \right\|_2^2 d_{\mathbb{M}}^4(Y, y) \right]  \\
        & \leq  \mathrm{diam} ( \mathbb{M} )^4 \cdot \mathbb{E} \left[ \mathcal{L}_{\mx, \mh}^2 \left( \mathbf{X} \right) \left\| \mathbf{\Lambda}_{\mh}^{-1} \bm{\Phi}_{\mx}^{-1} (\mathbf{X}) \right\|_2^2 \right] \\
        & \leq \mathrm{diam} ( \mathbb{M} )^4 \cdot \sup_{\mz \in \mathbb{T}^{d}} f ( \mz ) \cdot \int_{[-\pi, \pi)^{d}}  \bm{L}_{\mh}^2 (\bm{\theta})  \left\| \mathbf{\Lambda}_{\mh}^{-1}  \bm{\theta} \right\|_2^2 \mathrm{d} \bm{\theta}.
    \end{split} \end{align}
    Combining Lemma C.2 of Im et al. (2025), \eqref{mu1bound_step3} and \eqref{eq:4.110}, we get
    \begin{align} \label{eq:4.111}
    \begin{split}
        \left\| \mathbf{\Lambda}_{\mh}^{-1} \bm{\tilde{\nu}}_{\mh, 1}(\mathbf{x}, y)  \right\|_2 &= O\left(\rho ( \mh )\right), \\
        \| \mathbf{\Lambda}_{\mh}^{-1} \bm{\hat{\nu}}_{\mh, 1}(\mathbf{x}, y) - \mathbf{\Lambda}_{\mh}^{-1} \bm{\tilde{\nu}}_{\mh, 1}(\mathbf{x}, y)  \|_2 &= O_{\mathbb{P}}\left(n^{-\frac{1}{2}} \rho^{\frac{1}{2}} \left( \mh \right)\right).     
    \end{split}
    \end{align}
    Note that
    \begin{align} \begin{split} \label{eq:4.112}
        & \bm{\hat{\mu}}_{\mh, 1}(\mathbf{x})\tran\bm{\hat{\mu}}_{\mh, 2}(\mathbf{x})^{-1} \bm{\hat{\nu}}_{\mh, 1}(\mathbf{x}, y) - \bm{\tilde{\mu}}_{\mh, 1}(\mathbf{x})\tran\bm{\tilde{\mu}}_{\mh, 2}(\mathbf{x})^{-1} \bm{\tilde{\nu}}_{\mh, 1}(\mathbf{x}, y)  \\
        & = \left( \mathbf{\Lambda}_{\mh}^{-1} \bm{\hat{\mu}}_{\mh, 1}(\mathbf{x}) \right) \tran \left(\mathbf{\Lambda}_{\mh}^{-1} \bm{\hat{\mu}}_{\mh, 2}(\mathbf{x}) \mathbf{\Lambda}_{\mh}^{-1}\right)^{-1}\left(  \mathbf{\Lambda}_{\mh}^{-1} \bm{\hat{\nu}}_{\mh, 1}(\mathbf{x}, y) \right) \\
        & \quad - \left( \mathbf{\Lambda}_{\mh}^{-1} \bm{\tilde{\mu}}_{\mh, 1}(\mathbf{x}) \right) \tran \left(\mathbf{\Lambda}_{\mh}^{-1} \bm{\tilde{\mu}}_{\mh, 2}(\mathbf{x}) \mathbf{\Lambda}_{\mh}^{-1}\right)^{-1}\left(  \mathbf{\Lambda}_{\mh}^{-1} \bm{\tilde{\nu}}_{\mh, 1}(\mathbf{x}, y) \right)\\
        & = \left( \mathbf{\Lambda}_{\mh}^{-1} \bm{\hat{\mu}}_{\mh, 1}(\mathbf{x}) \right) \tran \left(\mathbf{\Lambda}_{\mh}^{-1} \bm{\hat{\mu}}_{\mh, 2}(\mathbf{x}) \mathbf{\Lambda}_{\mh}^{-1}\right)^{-1}\left(  \mathbf{\Lambda}_{\mh}^{-1} \bm{\hat{\nu}}_{\mh, 1}(\mathbf{x}, y) - \mathbf{\Lambda}_{\mh}^{-1} \bm{\tilde{\nu}}_{\mh, 1}(\mathbf{x}, y) \right) \\
        & \quad + \left( \mathbf{\Lambda}_{\mh}^{-1} \bm{\hat{\mu}}_{\mh, 1}(\mathbf{x}) \right) \tran \left( \left(\mathbf{\Lambda}_{\mh}^{-1} \bm{\hat{\mu}}_{\mh, 2}(\mathbf{x}) \mathbf{\Lambda}_{\mh}^{-1}\right)^{-1} - \left(\mathbf{\Lambda}_{\mh}^{-1} \bm{\tilde{\mu}}_{\mh, 2}(\mathbf{x}) \mathbf{\Lambda}_{\mh}^{-1}\right)^{-1} \right) \left( \mathbf{\Lambda}_{\mh}^{-1} \bm{\tilde{\nu}}_{\mh, 1}(\mathbf{x}, y) \right) \\
        & \quad + \left( \mathbf{\Lambda}_{\mh}^{-1} \bm{\hat{\mu}}_{\mh, 1}(\mathbf{x}) -  \mathbf{\Lambda}_{\mh}^{-1} \bm{\tilde{\mu}}_{\mh, 1}(\mathbf{x}) \right) \tran \left(\mathbf{\Lambda}_{\mh}^{-1} \bm{\tilde{\mu}}_{\mh, 2}(\mathbf{x}) \mathbf{\Lambda}_{\mh}^{-1}\right)^{-1}\left(  \mathbf{\Lambda}_{\mh}^{-1} \bm{\tilde{\nu}}_{\mh, 1}(\mathbf{x}, y) \right).
    \end{split} \end{align}
    Combining \eqref{eq:temp4.55}, \eqref{eq:temp4.56}, \eqref{mu2tildeinverse}, \eqref{mu2hatinverse}, \eqref{mu1hat}, \eqref{eq:4.111} and \eqref{eq:4.112}, we get
    \begin{align} \begin{split} \label{eq:4.113}
        & \left|  \bm{\hat{\mu}}_{\mh, 1}(\mathbf{x})\tran\bm{\hat{\mu}}_{\mh, 2}(\mathbf{x})^{-1} \bm{\hat{\nu}}_{\mh, 1}(\mathbf{x}, y) - \bm{\tilde{\mu}}_{\mh, 1}(\mathbf{x})\tran\bm{\tilde{\mu}}_{\mh, 2}(\mathbf{x})^{-1} \bm{\tilde{\nu}}_{\mh, 1}(\mathbf{x}, y) \right|  \\
        & \leq \left\| \mathbf{\Lambda}_{\mh}^{-1} \bm{\hat{\mu}}_{\mh, 1}(\mathbf{x}) \right\|_2 \cdot \left\| \left(\mathbf{\Lambda}_{\mh}^{-1} \bm{\hat{\mu}}_{\mh, 2}(\mathbf{x}) \mathbf{\Lambda}_{\mh}^{-1}\right)^{-1} \right\|_2 \cdot \left\|  \mathbf{\Lambda}_{\mh}^{-1} \bm{\hat{\nu}}_{\mh, 1}(\mathbf{x}, y) - \mathbf{\Lambda}_{\mh}^{-1} \bm{\tilde{\nu}}_{\mh, 1}(\mathbf{x}, y) \right\|_2 \\
        & \quad + \left\| \mathbf{\Lambda}_{\mh}^{-1} \bm{\hat{\mu}}_{\mh, 1}(\mathbf{x}) \right\|_2 \cdot \left\|  \left(\mathbf{\Lambda}_{\mh}^{-1} \bm{\hat{\mu}}_{\mh, 2}(\mathbf{x}) \mathbf{\Lambda}_{\mh}^{-1}\right)^{-1} - \left(\mathbf{\Lambda}_{\mh}^{-1} \bm{\tilde{\mu}}_{\mh, 2}(\mathbf{x}) \mathbf{\Lambda}_{\mh}^{-1}\right)^{-1}  \right\|_2 \cdot \left\| \mathbf{\Lambda}_{\mh}^{-1} \bm{\tilde{\nu}}_{\mh, 1}(\mathbf{x}, y) \right\|_2 \\
        & \quad + \left\| \mathbf{\Lambda}_{\mh}^{-1} \bm{\hat{\mu}}_{\mh, 1}(\mathbf{x}) -  \mathbf{\Lambda}_{\mh}^{-1} \bm{\tilde{\mu}}_{\mh, 1}(\mathbf{x}) \right\|_2 \cdot \left\| \left(\mathbf{\Lambda}_{\mh}^{-1} \bm{\tilde{\mu}}_{\mh, 2}(\mathbf{x}) \mathbf{\Lambda}_{\mh}^{-1}\right)^{-1} \right\|_2 \cdot \left\|  \mathbf{\Lambda}_{\mh}^{-1} \bm{\tilde{\nu}}_{\mh, 1}(\mathbf{x}, y) \right\|_2 \\
        & = o_{\mathbb{P}}\left(\rho ( \mh )\right) \cdot O_{\mathbb{P}}\left(\rho^{-1} \left( \mh \right)\right) \cdot O_{\mathbb{P}}\left(n^{-\frac{1}{2}} \rho^{\frac{1}{2}} \left( \mh \right)\right) + o_{\mathbb{P}}\left(\rho ( \mh )\right) \cdot O_{\mathbb{P}}\left(n^{-\frac{1}{2}} \rho^{-\frac{3}{2}} \left( \mh \right)\right) \cdot O\left(\rho ( \mh )\right) \\
        & \quad + O_{\mathbb{P}}\left(n^{-\frac{1}{2}} \rho^{\frac{1}{2}} \left( \mh \right)\right) \cdot O \left(\rho^{-1} \left( \mh \right)\right) \cdot O\left(\rho ( \mh )\right) \\
        & = O_{\mathbb{P}}\left(n^{-\frac{1}{2}} \rho^{\frac{1}{2}} \left( \mh \right)\right).
    \end{split} \end{align}
    \Cref{lemma:1.4} and \Cref{lemma:new2.2} imply that
    \begin{align}
        0 \leq \tilde{\nu}_{\mh, 0}(\mathbf{x}, y) \leq \mathrm{diam} ( \mathbb{M} )^2 \mathbb{E}\left[ \mathcal{L}_{\mx, \mh} \left( \mathbf{X} \right) \right] = O\left(\rho(\mh)\right). \label{eq:4.118}
    \end{align}
    Combining \eqref{eq:new4.38}, \eqref{mu2tildeinverse}, \eqref{eq:4.111} and \eqref{eq:4.118}, we get
    \begin{align} \begin{split} \label{eq:4.119}
        & \left| \tilde{\nu}_{\mh, 0}(\mathbf{x}, y) - \bm{\tilde{\mu}}_{\mh, 1}(\mathbf{x})\tran\bm{\tilde{\mu}}_{\mh, 2}(\mathbf{x})^{-1} \bm{\tilde{\nu}}_{\mh, 1}(\mathbf{x}, y) \right| \\
        & = \left| \tilde{\nu}_{\mh, 0}(\mathbf{x}, y) - \left( \mathbf{\Lambda}_{\mh}^{-1} \bm{\tilde{\mu}}_{\mh, 1}(\mathbf{x}) \right) \tran \left(\mathbf{\Lambda}_{\mh}^{-1} \bm{\tilde{\mu}}_{\mh, 2}(\mathbf{x}) \mathbf{\Lambda}_{\mh}^{-1}\right)^{-1}\left(  \mathbf{\Lambda}_{\mh}^{-1} \bm{\tilde{\nu}}_{\mh, 1}(\mathbf{x}, y) \right) \right|  \\
        & \leq |\tilde{\nu}_{\mh, 0}(\mathbf{x}, y)| + \left\|  \mathbf{\Lambda}_{\mh}^{-1} \bm{\tilde{\mu}}_{\mh, 1}(\mathbf{x}) \right\|_2 \cdot \left\| \left(\mathbf{\Lambda}_{\mh}^{-1} \bm{\tilde{\mu}}_{\mh, 2}(\mathbf{x}) \mathbf{\Lambda}_{\mh}^{-1}\right)^{-1} \right\|_2 \cdot \left\|  \mathbf{\Lambda}_{\mh}^{-1} \bm{\tilde{\nu}}_{\mh, 1}(\mathbf{x}, y) \right\|_2  \\
        & = O\left(\rho(\mh)\right) + o\left(\rho(\mh)\right) \cdot O\left(\rho^{-1}(\mh)\right) \cdot O\left(\rho(\mh)\right) \\
        & = O\left(\rho(\mh)\right).
    \end{split} \end{align}
    \Cref{lemma:temp4.7} implies that
    \begin{align} \begin{split} \label{eq:4.115}
        \bigg| \frac{1}{c_{\mh, \bm{0}_d, 1}(L)} \hat{\sigma}_{\mh}(\mathbf{x}) - f(\mathbf{x}) \bigg| & \leq \bigg| \frac{1}{c_{\mh, \bm{0}_d, 1}(L)} \tilde{\sigma}_{\mh}(\mathbf{x}) - f(\mathbf{x}) \bigg| + \bigg| \frac{1}{c_{\mh, \bm{0}_d, 1}(L)} \left( \hat{\sigma}_{\mh}(\mathbf{x}) -\tilde{\sigma}_{\mh}(\mathbf{x}) \right) \bigg|  \\
        & = o(1) +  O_{\mathbb{P}}\left(n^{-\frac{1}{2}} \rho^{-\frac{1}{2}} \left( \mh \right)\right) \\
        & = o_{\mathbb{P}}(1).
    \end{split} \end{align}
    Combining \Cref{lemma:1.4} and \eqref{eq:4.115}, we get
    \begin{align} \begin{gathered} \label{sigmahatinverse}
        \left| \frac{1}{\hat{\sigma}_{\mh}(\mathbf{x})} \right| = \frac{1}{c_{\mh, \bm{0}_d, 1}(L)} \bigg| \frac{1}{c_{\mh, \bm{0}_d, 1}(L)} \hat{\sigma}_{\mh}(\mathbf{x})\bigg|^{-1} =  O\left( \rho^{-1} \left( \mh \right) \right) O_{\mathbb{P}}(1) = O_{\mathbb{P}} \left( \rho^{-1} \left( \mh \right) \right).
    \end{gathered} \end{align}
    Combining \Cref{lemma:1.4}, \Cref{lemma:temp4.7}, \eqref{sigmatildeinverse}, \eqref{eq:new7.132}, \eqref{eq:4.109}, \eqref{eq:4.113}, \eqref{eq:4.119} and \eqref{sigmahatinverse}, we also get
    \begin{align*}
        & \left| \hat{M}_{\mh, 1}(\mathbf{x}, y) - \tilde{M}_{\mh, 1}(\mathbf{x}, y) \right| \\
        & = O_{\mathbb{P}} \left( \rho^{-1} \left( \mh \right) \right) \cdot O_{\mathbb{P}}\left(n^{-\frac{1}{2}} \rho^{\frac{1}{2}} \left( \mh \right)\right) + O_{\mathbb{P}} \left( \rho^{-1} \left( \mh \right) \right) \cdot O_{\mathbb{P}}\left(n^{-\frac{1}{2}} \rho^{\frac{1}{2}} \left( \mh \right)\right) \\
        & \quad + O \left( \rho^{-1} \left( \mh \right) \right) \cdot O_{\mathbb{P}} \left( \rho^{-1} \left( \mh \right) \right) \cdot O_{\mathbb{P}}\left(n^{-\frac{1}{2}} \rho^{\frac{1}{2}} \left( \mh \right)\right) \cdot O \left( \rho ( \mh ) \right) \\
        & = O_{\mathbb{P}}\left(n^{-\frac{1}{2}} \rho^{-\frac{1}{2}} \left( \mh \right)\right),
    \end{align*}
    which gives \eqref{lemma:new3.5}.

    Now, we prove the uniform version of \eqref{lemma:new3.5} over $y\in\mathbb{M}$. First, we define functions $\tilde{R}_{\mh}, \hat{R}_{\mh} : \mathbb{T}^{d} \to \mathbb{R}$ as
    \begin{align*}
        \tilde{R}_{\mh}(\mx) &:= \tilde{\mu}_{\mh, 0}(\mathbf{x}) + \left\| \mathbf{\Lambda}_{\mh}^{-1} \bm{\tilde{\mu}}_{\mh, 1}(\mathbf{x}) \right\|_2 \cdot
        \left\| \left(\mathbf{\Lambda}_{\mh}^{-1}\bm{\tilde{\mu}}_{\mh, 2}(\mathbf{x}) \mathbf{\Lambda}_{\mh}^{-1}\right)^{-1} \right\|_2 \cdot \mathbb{E} \left[ \mathcal{L}_{\mx, \mh} \left( \mathbf{X} \right)  \left\| \mathbf{\Lambda}_{\mh}^{-1} \bm{\Phi}_{\mx}^{-1} (\mathbf{X}) \right\|_2  \right], \\
        \hat{R}_{\mh}(\mx) &:= \hat{\mu}_{\mh, 0}(\mathbf{x}) + \left\| \mathbf{\Lambda}_{\mh}^{-1} \bm{\hat{\mu}}_{\mh, 1}(\mathbf{x}) \right\|_2 \cdot \left\| (\mathbf{\Lambda}_{\mh}^{-1}\bm{\hat{\mu}}_{\mh, 2}(\mathbf{x}) \mathbf{\Lambda}_{\mh}^{-1})^{-1} \right\|_2
        \cdot n^{-1} \sum_{i=1}^n \mathcal{L}_{\mx, \mh} \left(\mathbf{X}^{(i)}\right)  \left\| \mathbf{\Lambda}_{\mh}^{-1} \bm{\theta}^{(i)} \right\|_2.
    \end{align*}
    Note that
    \begin{align} \begin{split} \label{nutildebound}
        \left| \tilde{\nu}_{\mh, 0}(\mathbf{x}, y_1) - \tilde{\nu}_{\mh, 0}(\mathbf{x}, y_2) \right| & = \left| \mathbb{E} \left[ \mathcal{L}_{\mx, \mh} \left( \mathbf{X} \right)  \left( d_{\mathbb{M}}^2(Y, y_1) - d_{\mathbb{M}}^2(Y, y_2) \right) \right] \right| \\
        & \leq \mathbb{E} \left[ \mathcal{L}_{\mx, \mh} \left( \mathbf{X} \right)  \left| d_{\mathbb{M}}^2(Y, y_1) - d_{\mathbb{M}}^2(Y, y_2) \right| \right] \\
        & \leq 2 \mathrm{diam} (\mathbb{M}) d_{\mathbb{M}}(y_1, y_2) \tilde{\mu}_{\mh, 0}(\mathbf{x}), \\
        \left\| \mathbf{\Lambda}_{\mh}^{-1} \left( \tilde{\nu}_{\mh, 1}(\mathbf{x}, y_1) - \tilde{\nu}_{\mh, 1}(\mathbf{x}, y_2) \right) \right\|_2 & = \left\| \mathbb{E} \left[ \mathcal{L}_{\mx, \mh} \left( \mathbf{X} \right)  \mathbf{\Lambda}_{\mh}^{-1} \bm{\Phi}_{\mx}^{-1} (\mathbf{X}) \left( d_{\mathbb{M}}^2(Y, y_1) - d_{\mathbb{M}}^2(Y, y_2) \right) \right] \right\|_2 \\
        & \leq \mathbb{E} \left[ \mathcal{L}_{\mx, \mh} \left( \mathbf{X} \right)  \| \mathbf{\Lambda}_{\mh}^{-1} \bm{\Phi}_{\mx}^{-1} (\mathbf{X}) \|_2 \left| d_{\mathbb{M}}^2(Y, y_1) - d_{\mathbb{M}}^2(Y, y_2) \right| \right] \\
        & \leq 2 \mathrm{diam} (\mathbb{M}) d_{\mathbb{M}}(y_1, y_2) \mathbb{E} \left[ \mathcal{L}_{\mx, \mh} \left( \mathbf{X} \right)  \| \mathbf{\Lambda}_{\mh}^{-1} \bm{\Phi}_{\mx}^{-1} (\mathbf{X}) \|_2  \right]
    \end{split} \end{align}
    for any $y_1, y_2 \in \mathbb{M}$. Combining \eqref{eq:4.80} and \eqref{nutildebound}, we get
    \begin{align*}
        & \left| \tilde{\sigma}_{\mh} (\mathbf{x}) \right| \cdot \left| \tilde{M}_{\mh, 1}(\mathbf{x}, y_1) - \tilde{M}_{\mh, 1}(\mathbf{x}, y_2) \right|  \\
        & \leq \left|  \tilde{\nu}_{\mh, 0}(\mathbf{x}, y_1) - \tilde{\nu}_{\mh, 0}(\mathbf{x}, y_2) \right| + \left| \bm{\tilde{\mu}}_{\mh, 1}(\mathbf{x})\tran\bm{\tilde{\mu}}_{\mh, 2}(\mathbf{x})^{-1} \left(\bm{\tilde{\nu}}_{\mh, 1}(\mathbf{x}, y_1) - \bm{\tilde{\nu}}_{\mh, 1}(\mathbf{x}, y_2) \right) \right| \\
        & = \left|  \tilde{\nu}_{\mh, 0}(\mathbf{x}, y_1) - \tilde{\nu}_{\mh, 0}(\mathbf{x}, y_2) \right| \\
        & \quad + \left| \left( \mathbf{\Lambda}_{\mh}^{-1} \bm{\tilde{\mu}}_{\mh, 1}(\mathbf{x}) \right) \tran \left (\mathbf{\Lambda}_{\mh}^{-1}\bm{\tilde{\mu}}_{\mh, 2}(\mathbf{x}) \mathbf{\Lambda}_{\mh}^{-1} \right)^{-1} \left( \mathbf{\Lambda}_{\mh}^{-1} \bm{\tilde{\nu}}_{\mh, 1}(\mathbf{x}, y_1) - \mathbf{\Lambda}_{\mh}^{-1} \bm{\tilde{\nu}}_{\mh, 1}(\mathbf{x}, y_2) \right)\right|  \\
        & \leq \left|  \tilde{\nu}_{\mh, 0}(\mathbf{x}, y_1) - \tilde{\nu}_{\mh, 0}(\mathbf{x}, y_2) \right| \\
        & \quad + \| \mathbf{\Lambda}_{\mh}^{-1} \bm{\tilde{\mu}}_{\mh, 1}(\mathbf{x}) \|_2 \cdot \| (\mathbf{\Lambda}_{\mh}^{-1}\bm{\tilde{\mu}}_{\mh, 2}(\mathbf{x}) \mathbf{\Lambda}_{\mh}^{-1})^{-1} \|_2 \cdot \| \mathbf{\Lambda}_{\mh}^{-1} \bm{\tilde{\nu}}_{\mh, 1}(\mathbf{x}, y_1) - \mathbf{\Lambda}_{\mh}^{-1} \bm{\tilde{\nu}}_{\mh, 1}(\mathbf{x}, y_2) \|_2  \\
        & \leq   2 \mathrm{diam} (\mathbb{M}) d_{\mathbb{M}}(y_1, y_2) \tilde{R}_{\mh}(\mx).
    \end{align*}
    Similarly, we get
    \begin{align} \begin{split} \label{nuhatbound}
        \left| \hat{\nu}_{\mh, 0}(\mathbf{x}, y_1) - \hat{\nu}_{\mh, 0}(\mathbf{x}, y_2) \right| & \leq 2 \mathrm{diam} (\mathbb{M}) \cdot d_{\mathbb{M}}(y_1, y_2) \cdot \hat{\mu}_{\mh, 0}(\mathbf{x}), \\
        \left\| \mathbf{\Lambda}_{\mh}^{-1} \left( \hat{\nu}_{\mh, 1}(\mathbf{x}, y_1) - \hat{\nu}_{\mh, 1}(\mathbf{x}, y_2) \right) \right\|_2 & \leq 2 \mathrm{diam} (\mathbb{M}) \cdot d_{\mathbb{M}}(y_1, y_2) \cdot n^{-1} \sum_{i=1}^n \mathcal{L}_{\mx, \mh} \left(\mathbf{X}^{(i)}\right)  \left\| \mathbf{\Lambda}_{\mh}^{-1} \bm{\theta}^{(i)} \right\|_2. 
    \end{split} \end{align}
    Also, \eqref{eq:4.80} and \eqref{nuhatbound} imply that 
    \begin{align*}
        \left| \hat{\sigma}_{\mh} (\mathbf{x}) \right| \cdot \left| \hat{M}_{\mh, 1}(\mathbf{x}, y_1) - \hat{M}_{\mh, 1}(\mathbf{x}, y_2) \right| \leq   2 \mathrm{diam} (\mathbb{M}) d_{\mathbb{M}}(y_1, y_2) \hat{R}_{\mh}(\mx).
    \end{align*}
    Similarly to \eqref{eq:4.110}, we get
    \begin{align} \begin{split} \label{Rhbound_step1}
         \left( \mathbb{E} \left[ \mathcal{L}_{\mx, \mh} \left( \mathbf{X} \right)  \left\| \mathbf{\Lambda}_{\mh}^{-1} \bm{\Phi}_{\mx}^{-1} (\mathbf{X}) \right\|_2  \right] \right)^2 & \leq \mathbb{E} \left[ \mathcal{L}_{\mx, \mh}^2 \left( \mathbf{X} \right)  \left\| \mathbf{\Lambda}_{\mh}^{-1} \bm{\Phi}_{\mx}^{-1} (\mathbf{X}) \right\|_2^2  \right] \\
         & \leq \sup_{\mz \in \mathbb{T}^{d}} f ( \mz ) \cdot \int_{[-\pi, \pi)^{d}}  \bm{L}_{\mh}^2 (\bm{\theta})  \left\| \mathbf{\Lambda}_{\mh}^{-1}  \bm{\theta} \right\|_2^2 \mathrm{d} \bm{\theta}. 
    \end{split} \end{align}
    Combining Lemma C.2 of Im et al. (2025), \eqref{mu1bound_step3} and \eqref{Rhbound_step1}, we get
    \begin{align} \label{Rhbound_step2}
    \begin{split}
        \mathbb{E} \left[ \mathcal{L}_{\mx, \mh} \left( \mathbf{X} \right)  \left\| \mathbf{\Lambda}_{\mh}^{-1} \bm{\Phi}_{\mx}^{-1} (\mathbf{X}) \right\|_2  \right] &= O\left(\rho ( \mh )\right), \\
        n^{-1} \sum_{i=1}^n \mathcal{L}_{\mx, \mh} \left(\mathbf{X}^{(i)}\right)  \left\| \mathbf{\Lambda}_{\mh}^{-1} \bm{\theta}^{(i)} \right\|_2 - \mathbb{E} \left[ \mathcal{L}_{\mx, \mh} \left( \mathbf{X} \right)  \left\| \mathbf{\Lambda}_{\mh}^{-1} \bm{\Phi}_{\mx}^{-1} (\mathbf{X}) \right\|_2  \right] &= O_{\mathbb{P}}\left(n^{-\frac{1}{2}} \rho^{\frac{1}{2}} \left( \mh \right)\right). 
    \end{split}
    \end{align}
    Combining \Cref{lemma:1.4}, \eqref{eq:4.51}, \eqref{eq:temp4.73} and \eqref{Rhbound_step2}, we also get
    \begin{align} \begin{split} \label{Rhbound_step3}
        \hat{\mu}_{\mh, 0}(\mathbf{x}) &= O\left(\rho ( \mh )\right) + O_{\mathbb{P}}\left(n^{-\frac{1}{2}} \rho^{\frac{1}{2}} \left( \mh \right)\right) = O_{\mathbb{P}}\left(\rho ( \mh )\right), \\
        n^{-1} \sum_{i=1}^n \mathcal{L}_{\mx, \mh} \left(\mathbf{X}^{(i)}\right) \left\| \mathbf{\Lambda}_{\mh}^{-1} \bm{\theta}^{(i)} \right\|_2 &= O\left(\rho ( \mh )\right) + O_{\mathbb{P}}\left(n^{-\frac{1}{2}} \rho^{\frac{1}{2}} \left( \mh \right)\right) = O_{\mathbb{P}}\left(\rho ( \mh )\right).
    \end{split} \end{align} 
    Combining \eqref{eq:new4.38}, \eqref{mu2tildeinverse}, \eqref{mu2hatinverse}, \eqref{mu1hat}, \eqref{sigmatildeinverse}, \eqref{sigmahatinverse}, \eqref{Rhbound_step2} and \eqref{Rhbound_step3}, we get
    \begin{align} \begin{split} \label{eq:new7.144}
        \bigg| \frac{1}{\tilde{\sigma}_{\mh} (\mathbf{x})} \bigg| \cdot \tilde{R}_{\mh}(\mx) & = O\left(\rho^{-1} \left( \mh \right)\right) \cdot \left( O\left(\rho ( \mh )\right)+O\left(\rho ( \mh )\right) \cdot O\left(\rho^{-1} \left( \mh \right)\right) \cdot O\left(\rho ( \mh )\right) \right) \\
        & = O(1) , \\
        \bigg| \frac{1}{ \hat{\sigma}_{\mh} (\mathbf{x})} \bigg| \cdot \hat{R}_{\mh}(\mx) & = O_{\mathbb{P}}\left(\rho^{-1} \left( \mh \right)\right) \cdot \left( O_{\mathbb{P}}\left(\rho ( \mh )\right)+O_{\mathbb{P}}\left(\rho ( \mh )\right) \cdot O_{\mathbb{P}}\left(\rho^{-1} \left( \mh \right)\right) \cdot O_{\mathbb{P}}\left(\rho ( \mh )\right) \right) \\
        & =  O_{\mathbb{P}}(1).
    \end{split}  \end{align}
    The rest of the proof proceeds as in the proof of Lemma C.8 of Im et al. (2025).
\end{proof}

\begin{lemma} \label{thm:4.3}
    Assume that the conditions \ref{con:L1}, \ref{con:B1}, \ref{con:D1}, \ref{con:D2} and \ref{con:M1} hold. Then, it holds that
    \begin{align*}
        d_{\mathbb{M}}(\hat{m}_{\mh, 1}(\mathbf{x}), \tilde{m}_{\mh, 1}(\mathbf{x})) = o_{\mathbb{P}}(1).
    \end{align*}
\end{lemma}

\begin{proof} [Proof of \Cref{thm:4.3}]
    The lemma follows by arguing as in the proof of \Cref{thm:3.3} and using \Cref{lemma:4.10}.
\end{proof}

\subsection{Proof of Theorem \ref{thm:consistency} for $\hat{m}_{\mh, 1}(\mathbf{x})$}

The theorem follows from \Cref{thm:4.1} and \Cref{thm:4.3}.

\subsection{Lemmas for proof of \Cref{thm:main4.3}}

\begin{lemma} \label{lemma:4.4}
    Assume that the conditions \ref{con:L1} and \ref{con:D4} hold and that $\lim_{n\rightarrow\infty} \| \mh \|_2=0$. Then, it holds that
    \begin{align} \begin{split}\label{eq:4.38}
          \left\| \mathbf{\Lambda}_{\mh}^{-1} \bm{\tilde{\mu}}_{\mh, 1}(\mathbf{x}) \right\|_2 &= O\left( \rho ( \mh ) \| \mh \|_2  \right).
    \end{split} \end{align}
    If the condition \ref{con:D5} further holds, then it holds that
    \begin{align} \begin{split}\label{eq:4.39}
          \sup_{y \in \mathbb{M}} \left\| \mathbf{\Lambda}_{\mh}^{-1} \bm{\tilde{\tau}}_{\mh, 1}(\mathbf{x}, y) \right\|_2 &= O\left( \rho ( \mh ) \| \mh \|_2  \right).
    \end{split} \end{align}
\end{lemma}

\begin{proof} [Proof of \Cref{lemma:4.4}]
    We first prove \eqref{eq:4.38}. Combining \Cref{lemma:2.1}, \eqref{eq:2.11}, \eqref{mu1bound_step1}, \eqref{mu1bound_step3} and the Cauchy-Schwarz inequality, we get
    \begin{align*}
        \left\| \mathbf{\Lambda}_{\mh}^{-1} \bm{\tilde{\mu}}_{\mh, 1}(\mathbf{x}) \right\|_{2}^{2} & \leq \left( \int_{[-\pi, \pi)^{d}}  \bm{L}_{\mh} (\bm{\theta}) \left\| \mathbf{\Lambda}_{\mh}^{-1} \bm{\theta} \right\|_2  \left|  f \left( \bm{\Phi}_{\mx} \left( \bm{\theta} \right) \right) - f(\mx)\right| \mathrm{d} \bm{\theta} \right)^2\\
        & \leq \left( \int_{[-\pi, \pi)^{d}}  \bm{L}_{\mh} (\bm{\theta}) \left\| \mathbf{\Lambda}_{\mh}^{-1} \bm{\theta} \right\|_2 \| \bm{\theta} \|_2 \left( \sup_{\mz \in \mathbb{T}^{d}} \left\| \nabla \bar{f}  \left( \mz \right) \right\|_2 \right) \mathrm{d} \bm{\theta} \right)^2 \\
        & \leq  \left( \sup_{\mz \in \mathbb{T}^{d}} \left\| \nabla \bar{f}  \left( \mz \right) \right\|_2 \right)^2 \cdot \left( \int_{[-\pi, \pi)^{d}} \bm{L}_{\mh} (\bm{\theta})   \| \bm{\theta} \|_2^2  \mathrm{d} \bm{\theta} \right) \\
        & \quad \cdot \left( \int_{[-\pi, \pi)^{d}} \bm{L}_{\mh} (\bm{\theta})   \left\| \mathbf{\Lambda}_{\mh}^{-1} \bm{\theta} \right\|_2^2  \mathrm{d} \bm{\theta} \right) \\
        & \leq O\left( \rho ( \mh ) \| \mh \|_2^2  \right) \cdot O\left( \rho ( \mh ) \right) \\
        & = O\left( \rho^2 ( \mh ) \| \mh \|_2^2  \right).
    \end{align*}
    Using \Cref{lemma:3.1}, we similarly get
    \begin{align*}
        \sup_{y \in \mathbb{M}} \left\| \mathbf{\Lambda}_{\mh}^{-1} \bm{\tilde{\tau}}_{\mh, 1}(\mathbf{x}, y) \right\|_2 & \leq \sup_{y \in \mathbb{M}} \int_{[-\pi, \pi)^{d}}  \bm{L}_{\mh} (\bm{\theta}) \left\| \mathbf{\Lambda}_{\mh}^{-1} \bm{\theta} \right\|_2  \left|  (f \cdot g_y) \left( \bm{\Phi}_{\mx} \left( \bm{\theta} \right) \right) - (f \cdot g_y)(\mx)\right| \mathrm{d} \bm{\theta} \\
        & \leq \sup_{y \in \mathbb{M}} \sup_{\mz \in \mathbb{T}^{d}} \left\| \nabla \left( \bar{f} \cdot \bar{g}_y \right)  \left( \mz \right) \right\|_2 \cdot \int_{[-\pi, \pi)^{d}}  \bm{L}_{\mh} (\bm{\theta}) \left\| \mathbf{\Lambda}_{\mh}^{-1} \bm{\theta} \right\|_2 \| \bm{\theta} \|_2 \mathrm{d} \bm{\theta} \\
        & = O\left( \rho ( \mh ) \| \mh \|_2  \right).
    \end{align*}
    This completes the proof.
\end{proof}

\begin{lemma} \label{thm:4.2}
    Assume that the conditions \ref{con:L1}, \ref{con:D1}, \ref{con:D4}, \ref{con:D5}, \ref{con:M1} and \ref{con:M2} hold and that $\lim_{n\rightarrow\infty} \| \mh \|_2=0$. Then, it holds that
    \begin{align*}
        d_{\mathbb{M}}(\tilde{m}_{\mh, 1}(\mathbf{x}), m_{\oplus}(\mathbf{x}))^{\beta_{\oplus} - 1} = O\left(\| \mh \|_{2}^{2}\right).
    \end{align*}
\end{lemma}

\begin{proof} [Proof of \Cref{thm:4.2}]
    By arguing as in the proof of Lemma C.14 of Im et al. (2025), we get
    \begin{align} \begin{split} \label{Fubini_moplus}
         & M_{\oplus} (\mathbf{x}, \tilde{m}_{\mh, 1}(\mathbf{x})) - M_{\oplus} (\mathbf{x}, m_{\oplus}(\mathbf{x}))  \\
         & \leq 2 \mathrm{diam} (\mathbb{M}) \cdot d_{\mathbb{M}}(\tilde{m}_{\mh, 1}(\mathbf{x}), m_{\oplus}(\mathbf{x})) \\
         & \quad \cdot \sup_{y \in \mathbb{M}}
         \bigg|  \frac{\tilde{\tau}_{\mh, 0}(\mathbf{x}, y) - \bm{\tilde{\mu}}_{\mh, 1}(\mathbf{x})\tran\bm{\tilde{\mu}}_{\mh, 2}(\mathbf{x})^{-1} \bm{\tilde{\tau}}_{\mh, 1}(\mathbf{x}, y)}{\tilde{\sigma}_{\mh}(\mathbf{x})} - g_y(\mathbf{x}) \bigg|.
    \end{split} \end{align}
    Combining \eqref{tau0bound}, \eqref{mu2tildeinverse}, \eqref{eq:4.38} and \eqref{eq:4.39}, we get
    \begin{align} \begin{split}\label{uniform tau}
        & \sup_{y \in \mathbb{M}} \left| \tilde{\tau}_{\mh, 0}(\mathbf{x}, y) - \bm{\tilde{\mu}}_{\mh, 1}(\mathbf{x})\tran\bm{\tilde{\mu}}_{\mh, 2}(\mathbf{x})^{-1} \bm{\tilde{\tau}}_{\mh, 1}(\mathbf{x}, y) -  g_y(\mathbf{x}) \tilde{\sigma}_{\mh}(\mathbf{x}) \right| \\
        & \leq \sup_{y \in \mathbb{M}} \left| \tilde{\tau}_{\mh, 0}(\mathbf{x}, y) -  g_y(\mathbf{x}) \tilde{\mu}_{\mh, 0}(\mathbf{x})  \right| \\
        & \quad + \sup_{y \in \mathbb{M}} \left| \bm{\tilde{\mu}}_{\mh, 1}(\mathbf{x})\tran\bm{\tilde{\mu}}_{\mh, 2}(\mathbf{x})^{-1} \left( \bm{\tilde{\tau}}_{\mh, 1}(\mathbf{x}, y) - g_y(\mathbf{x}) \bm{\tilde{\mu}}_{\mh, 1}(\mathbf{x}) \right) \right| \\
        & = \sup_{y \in \mathbb{M}} \left| \tilde{\tau}_{\mh, 0}(\mathbf{x}, y) -  g_y(\mathbf{x}) \tilde{\mu}_{\mh, 0}(\mathbf{x})  \right| \\
        & \quad + \sup_{y \in \mathbb{M}} \left| \left( \mathbf{\Lambda}_{\mh}^{-1} \bm{\tilde{\mu}}_{\mh, 1}(\mathbf{x}) \right)\tran\left(\mathbf{\Lambda}_{\mh}^{-1}\bm{\tilde{\mu}}_{\mh, 2}(\mathbf{x}) \mathbf{\Lambda}_{\mh}^{-1}\right)^{-1} \left( \mathbf{\Lambda}_{\mh}^{-1} \bm{\tilde{\tau}}_{\mh, 1}(\mathbf{x}, y) - g_y(\mathbf{x}) \mathbf{\Lambda}_{\mh}^{-1} \bm{\tilde{\mu}}_{\mh, 1}(\mathbf{x}) \right) \right| \\
        & \leq \sup_{y \in \mathbb{M}} \left| \tilde{\tau}_{\mh, 0}(\mathbf{x}, y) -  g_y(\mathbf{x}) \tilde{\mu}_{\mh, 0}(\mathbf{x})  \right| \\
        & \quad + \sup_{y \in \mathbb{M}} g_y(\mathbf{x}) \cdot \left\| \mathbf{\Lambda}_{\mh}^{-1} \bm{\tilde{\mu}}_{\mh, 1}(\mathbf{x}) \right\|_2^2 \cdot \left\| \left(\mathbf{\Lambda}_{\mh}^{-1}\bm{\tilde{\mu}}_{\mh, 2}(\mathbf{x}) \mathbf{\Lambda}_{\mh}^{-1}\right)^{-1} \right\|_2  \\
        & \quad + \left\| \mathbf{\Lambda}_{\mh}^{-1} \bm{\tilde{\mu}}_{\mh, 1}(\mathbf{x}) \right\|_2 \cdot \left\| \left(\mathbf{\Lambda}_{\mh}^{-1}\bm{\tilde{\mu}}_{\mh, 2}(\mathbf{x}) \mathbf{\Lambda}_{\mh}^{-1}\right)^{-1} \right\|_2  \cdot \sup_{y \in \mathbb{M}} \left\| \mathbf{\Lambda}_{\mh}^{-1} \bm{\tilde{\tau}}_{\mh, 1}(\mathbf{x}, y) \right\|_2 \\
        & = O\left( \rho ( \mh ) \| \mh \|_2^2  \right) + O\left( \rho^2 ( \mh ) \| \mh \|_2^2  \right) \cdot O\left( \rho^{-1} \left( \mh \right) \right) \\
        & \quad + O\left( \rho ( \mh ) \| \mh \|_2  \right) \cdot O\left( \rho^{-1} \left( \mh \right) \right) \cdot O\left( \rho ( \mh ) \| \mh \|_2  \right) \\
        & = O\left( \rho ( \mh ) \| \mh \|_2^2  \right).
    \end{split} \end{align}
    Combining \eqref{sigmatildeinverse} and  \eqref{uniform tau}, we also get
    \begin{align} \label{eq:new7.147}
        \sup_{y \in \mathbb{M}}
         \bigg|  \frac{\tilde{\tau}_{\mh, 0}(\mathbf{x}, y) - \bm{\tilde{\mu}}_{\mh, 1}(\mathbf{x})\tran\bm{\tilde{\mu}}_{\mh, 2}(\mathbf{x})^{-1} \bm{\tilde{\tau}}_{\mh, 1}(\mathbf{x}, y)}{\tilde{\sigma}_{\mh}(\mathbf{x})} - g_y(\mathbf{x}) \bigg| = O\left( \| \mh \|_2^2 \right).
    \end{align}
    \Cref{thm:4.1} and the condition \ref{con:M2} imply that there exists $N\in\mathbb{N}$ such that
    \begin{align}
        C_{\oplus} \cdot d_{\mathbb{M}}(\tilde{m}_{\mh, 1}(\mathbf{x}), m_{\oplus}(\mathbf{x}))^{\beta_{\oplus}} \leq M_{\oplus} (\mathbf{x}, \tilde{m}_{\mh, 1}(\mathbf{x})) - M_{\oplus} (\mathbf{x}, m_{\oplus}(\mathbf{x})) \label{eq:new7.148}
    \end{align}
    whenever $n\geq N$, where $C_{\oplus} >0$ is the constant defined in the condition \ref{con:M2}. Combining \eqref{Fubini_moplus}, \eqref{eq:new7.147} and \eqref{eq:new7.148}, we get the desired result.
\end{proof}

\begin{lemma} \label{thm:4.4}
    Assume that the conditions \ref{con:L1}, \ref{con:B1}, \ref{con:D1}--\ref{con:D3} and \ref{con:M1}--\ref{con:M3} hold. Then, it holds that
    \begin{align*}
        d_{\mathbb{M}}(\hat{m}_{\mh, 1}(\mathbf{x}), \tilde{m}_{\mh, 1}(\mathbf{x})) ^ {\beta_{\oplus}-\alpha_{\mathbb{M}}} = O_{\mathbb{P}} \left( n^{-\frac{1}{2}} \rho^{-\frac{1}{2}} \left( \mh \right) \right).
    \end{align*}
\end{lemma}

\begin{proof} [Proof of Theorem \ref{thm:4.4}]
    We define functions $\hat{T}_{\mh, 1} : \mathbb{T}^{d} \times \mathbb{M} \to \mathbb{R}$, $U_{\mh, 1} : \mathbb{T}^{d} \times \mathbb{M} \times \mathbb{M} \to \mathbb{R}$ and $\hat{S}_{\mh,1} : \mathbb{T}^{d} \times \mathbb{M} \to \mathbb{R}$ as
    \begin{align*}
        \hat{T}_{\mh, 1}(\mx, y) &:= \hat{M}_{\mh, 1}(\mx, y) - \tilde{M}_{\mh, 1}(\mx, y), \\
        U_{\mh, 1}(\mz, w , y ) &:=  D_{\mh, 1}(\mx, w , y ) \tilde{W}_{\mathbf{x}, \mh} (\mz ), \\
        \hat{S}_{\mh,1}(\mx,y) &:= n^{-1} \sum_{i=1}^{n} U_{\mh, 1}(\mathbf{X}^{(i)}, Y^{(i)}, y) - \mathbb{E} \left[ U_{\mh, 1}(\mathbf{X}, Y, y) \right],
    \end{align*}
    where $D_{\mh, 1}(\mx, w , y ) := d_{\mathbb{M}}^2(w , y ) - d_{\mathbb{M}}^2(w ,  \tilde{m}_{\mh, 1}(\mx))$.
    By arguing as in \eqref{eq:new3.38}, we get
    \begin{align} \begin{split} \label{T1hatbound}
        &\left| \hat{T}_{\mh, 1}(\mathbf{x}, y) - \hat{T}_{\mh, 1}(\mathbf{x}, \tilde{m}_{\mh, 1}(\mathbf{x})) \right| \\
        & \leq 2 \mathrm{diam}(\mathbb{M}) \cdot d_{\mathbb{M}}(y, \tilde{m}_{\mh, 1}(\mathbf{x}))  \cdot n^{-1} \sum_{i=1}^{n} \left| \hat{W}_{\mathbf{x}, \mh} (\mathbf{X}^{(i)} ) - \tilde{W}_{\mathbf{x}, \mh} (\mathbf{X}^{(i)} ) \right| + \left| \hat{S}_{\mh,1}(\mx,y) \right|.
    \end{split} \end{align}
    
    Regarding the first term in \eqref{T1hatbound}, note that
    \begin{align*}
        & \tilde{\sigma}_{\mh} (\mathbf{x}) \hat{\sigma}_{\mh} (\mathbf{x}) \left(\hat{W}_{\mathbf{x}, \mh} (\mathbf{X}^{(i)} ) - \tilde{W}_{\mathbf{x}, \mh} (\mathbf{X}^{(i)} ) \right) \\
        & = \left( \tilde{\sigma}_{\mh} (\mathbf{x}) - \hat{\sigma}_{\mh} (\mathbf{x}) \right) \mathcal{L}_{\mx, \mh} \left( \mathbf{X}^{(i)} \right) \\
        & \quad - \left( \tilde{\sigma}_{\mh} (\mathbf{x}) \bm{\hat{\mu}}_{\mh, 1}(\mathbf{x})\tran \bm{\hat{\mu}}_{\mh, 2}(\mathbf{x})^{-1} - \hat{\sigma}_{\mh} (\mathbf{x}) \bm{\tilde{\mu}}_{\mh, 1}(\mathbf{x})\tran \bm{\tilde{\mu}}_{\mh, 2}(\mathbf{x})^{-1} \right) \bm{\theta}^{(i)} \mathcal{L}_{\mx, \mh} \left( \mathbf{X}^{(i)} \right) \\
        & = \left( \tilde{\sigma}_{\mh} (\mathbf{x}) - \hat{\sigma}_{\mh} (\mathbf{x}) \right) \mathcal{L}_{\mx, \mh} \left( \mathbf{X}^{(i)} \right) \\
        & \quad - \left( \tilde{\sigma}_{\mh} (\mathbf{x}) \bm{\hat{\mu}}_{\mh, 1}(\mathbf{x})\tran \bm{\hat{\mu}}_{\mh, 2}(\mathbf{x})^{-1} \mathbf{\Lambda}_{\mh}- \hat{\sigma}_{\mh} (\mathbf{x}) \bm{\tilde{\mu}}_{\mh, 1}(\mathbf{x})\tran \bm{\tilde{\mu}}_{\mh, 2}(\mathbf{x})^{-1} \mathbf{\Lambda}_{\mh} \right) \mathbf{\Lambda}_{\mh}^{-1} \bm{\theta}^{(i)} \mathcal{L}_{\mx, \mh} \left( \mathbf{X}^{(i)} \right).
    \end{align*}
    This implies that
    \begin{align} \begin{split} \label{Whatbound_Step2}
        & \left| \tilde{\sigma}_{\mh} (\mathbf{x}) \right| \cdot \left| \hat{\sigma}_{\mh} (\mathbf{x}) \right| \cdot n^{-1} \sum_{i=1}^{n} \left| \hat{W}_{\mathbf{x}, \mh} (\mathbf{X}^{(i)} ) - \tilde{W}_{\mathbf{x}, \mh} (\mathbf{X}^{(i)} )  \right|\\
        & \leq \left| \tilde{\sigma}_{\mh} (\mathbf{x}) - \hat{\sigma}_{\mh} (\mathbf{x}) \right| \cdot \hat{\mu}_{\mh, 0} (\mathbf{x}) \\
        & \quad + \left\| \tilde{\sigma}_{\mh} (\mathbf{x}) \bm{\hat{\mu}}_{\mh, 1}(\mathbf{x})\tran \bm{\hat{\mu}}_{\mh, 2}(\mathbf{x})^{-1} \mathbf{\Lambda}_{\mh}- \hat{\sigma}_{\mh} (\mathbf{x}) \bm{\tilde{\mu}}_{\mh, 1}(\mathbf{x})\tran \bm{\tilde{\mu}}_{\mh, 2}(\mathbf{x})^{-1} \mathbf{\Lambda}_{\mh} \right\|_2 \\
        & \quad \quad \cdot n^{-1} \sum_{i=1}^{n} \left\| \mathbf{\Lambda}_{\mh}^{-1} \bm{\theta}^{(i)} \right\|_2 \mathcal{L}_{\mx, \mh} \left( \mathbf{X}^{(i)} \right).
    \end{split} \end{align}
    Note that
    \begin{align} \begin{split} \label{Whatbound_Step3}
        & \tilde{\sigma}_{\mh} (\mathbf{x}) \bm{\hat{\mu}}_{\mh, 1}(\mathbf{x})\tran \bm{\hat{\mu}}_{\mh, 2}(\mathbf{x})^{-1} \mathbf{\Lambda}_{\mh} - \hat{\sigma}_{\mh} (\mathbf{x}) \bm{\tilde{\mu}}_{\mh, 1}(\mathbf{x})\tran \bm{\tilde{\mu}}_{\mh, 2}(\mathbf{x})^{-1} \mathbf{\Lambda}_{\mh} \\
        & = \tilde{\sigma}_{\mh} (\mathbf{x}) \left( \mathbf{\Lambda}_{\mh}^{-1} \bm{\hat{\mu}}_{\mh, 1}(\mathbf{x}) \right) \tran \left(\mathbf{\Lambda}_{\mh}^{-1}\bm{\hat{\mu}}_{\mh, 2}(\mathbf{x}) \mathbf{\Lambda}_{\mh}^{-1}\right)^{-1} - \hat{\sigma}_{\mh} (\mathbf{x}) \left( \mathbf{\Lambda}_{\mh}^{-1} \bm{\tilde{\mu}}_{\mh, 1}(\mathbf{x}) \right) \tran \left(\mathbf{\Lambda}_{\mh}^{-1}\bm{\tilde{\mu}}_{\mh, 2}(\mathbf{x}) \mathbf{\Lambda}_{\mh}^{-1}\right)^{-1} \\
        & = \tilde{\sigma}_{\mh} (\mathbf{x}) \left( \mathbf{\Lambda}_{\mh}^{-1} \bm{\hat{\mu}}_{\mh, 1}(\mathbf{x}) \right) \tran \left( \left(\mathbf{\Lambda}_{\mh}^{-1}\bm{\hat{\mu}}_{\mh, 2}(\mathbf{x}) \mathbf{\Lambda}_{\mh}^{-1}\right)^{-1} - \left(\mathbf{\Lambda}_{\mh}^{-1}\bm{\tilde{\mu}}_{\mh, 2}(\mathbf{x}) \mathbf{\Lambda}_{\mh}^{-1}\right)^{-1} \right) \\
        & \quad +  \tilde{\sigma}_{\mh} (\mathbf{x}) \left( \mathbf{\Lambda}_{\mh}^{-1} \bm{\hat{\mu}}_{\mh, 1}(\mathbf{x}) - \mathbf{\Lambda}_{\mh}^{-1} \bm{\tilde{\mu}}_{\mh, 1}(\mathbf{x}) \right) \tran \left(\mathbf{\Lambda}_{\mh}^{-1}\bm{\tilde{\mu}}_{\mh, 2}(\mathbf{x}) \mathbf{\Lambda}_{\mh}^{-1}\right)^{-1} \\
        & \quad + \left( \tilde{\sigma}_{\mh} (\mathbf{x})- \hat{\sigma}_{\mh} (\mathbf{x}) \right)  \left( \mathbf{\Lambda}_{\mh}^{-1} \bm{\tilde{\mu}}_{\mh, 1}(\mathbf{x}) \right) \tran \left(\mathbf{\Lambda}_{\mh}^{-1}\bm{\tilde{\mu}}_{\mh, 2}(\mathbf{x}) \mathbf{\Lambda}_{\mh}^{-1}\right)^{-1}.
    \end{split} \end{align}
    Combining \Cref{lemma:1.4}, \Cref{lemma:new4.4}, \Cref{lemma:temp4.6}, \Cref{lemma:temp4.7}, \eqref{mu2tildeinverse}, \eqref{mu1hat} and \eqref{Whatbound_Step3}, we get
    \begin{align} \begin{split} \label{Whatbound_Step4}
        & \left\| \tilde{\sigma}_{\mh} (\mathbf{x}) \bm{\hat{\mu}}_{\mh, 1}(\mathbf{x})\tran \bm{\hat{\mu}}_{\mh, 2}(\mathbf{x})^{-1} \mathbf{\Lambda}_{\mh} - \hat{\sigma}_{\mh} (\mathbf{x}) \bm{\tilde{\mu}}_{\mh, 1}(\mathbf{x})\tran \bm{\tilde{\mu}}_{\mh, 2}(\mathbf{x})^{-1} \mathbf{\Lambda}_{\mh} \right\|_2 \\
        & \leq \left| \tilde{\sigma}_{\mh} (\mathbf{x}) \right| \cdot \left\| \mathbf{\Lambda}_{\mh}^{-1} \bm{\hat{\mu}}_{\mh, 1}(\mathbf{x}) \right\|_2 \cdot \left\| \left( \mathbf{\Lambda}_{\mh}^{-1}\bm{\hat{\mu}}_{\mh, 2}(\mathbf{x}) \mathbf{\Lambda}_{\mh}^{-1}\right)^{-1} - \left(\mathbf{\Lambda}_{\mh}^{-1}\bm{\tilde{\mu}}_{\mh, 2}(\mathbf{x}) \mathbf{\Lambda}_{\mh}^{-1}\right)^{-1} \right\|_2 \\
        & \quad +  \left| \tilde{\sigma}_{\mh} (\mathbf{x}) \right| \cdot \left\| \mathbf{\Lambda}_{\mh}^{-1} \bm{\hat{\mu}}_{\mh, 1}(\mathbf{x}) - \mathbf{\Lambda}_{\mh}^{-1} \bm{\tilde{\mu}}_{\mh, 1}(\mathbf{x}) \right\|_2 \cdot \left\| \left(\mathbf{\Lambda}_{\mh}^{-1}\bm{\tilde{\mu}}_{\mh, 2}(\mathbf{x}) \mathbf{\Lambda}_{\mh}^{-1}\right)^{-1} \right\|_2 \\
        & \quad + \left| \tilde{\sigma}_{\mh} (\mathbf{x})- \hat{\sigma}_{\mh} (\mathbf{x}) \right| \cdot  \left\| \mathbf{\Lambda}_{\mh}^{-1} \bm{\tilde{\mu}}_{\mh, 1}(\mathbf{x}) \right\|_2 \cdot \left\| \left(\mathbf{\Lambda}_{\mh}^{-1}\bm{\tilde{\mu}}_{\mh, 2}(\mathbf{x}) \mathbf{\Lambda}_{\mh}^{-1}\right)^{-1} \right\|_2 \\
        & \leq O \left( \rho ( \mh ) \right) \cdot o_{\mathbb{P}} \left( \rho ( \mh ) \right) \cdot O_{\mathbb{P}} \left( n^{-\frac{1}{2}}\rho^{-\frac{3}{2}} \left( \mh \right) \right) + O \left( \rho ( \mh ) \right) \cdot O_{\mathbb{P}} \left( n^{-\frac{1}{2}}\rho^{\frac{1}{2}} \left( \mh \right) \right) \cdot O \left( \rho^{-1} \left( \mh \right) \right) \\
        & \quad + O_{\mathbb{P}} \left( n^{-\frac{1}{2}}\rho^{\frac{1}{2}} \left( \mh \right) \right) \cdot o \left( \rho ( \mh ) \right) \cdot O \left( \rho^{-1} \left( \mh \right) \right) \\
        & = O_{\mathbb{P}} \left( n^{-\frac{1}{2}}\rho^{\frac{1}{2}} \left( \mh \right) \right).
    \end{split} \end{align}
    Combining \Cref{lemma:1.4}, \Cref{lemma:temp4.7}, \eqref{sigmatildeinverse}, \eqref{sigmahatinverse}, \eqref{Rhbound_step3}, \eqref{Whatbound_Step2} and \eqref{Whatbound_Step4}, we also get
    \begin{align} \begin{split} \label{Whatbound_Step5}
        & n^{-1} \sum_{i=1}^{n} \left| \hat{W}_{\mathbf{x}, \mh} (\mathbf{X}^{(i)} ) - \tilde{W}_{\mathbf{x}, \mh} (\mathbf{X}^{(i)} )  \right| \\
        & = O \left( \rho^{-1} \left( \mh \right) \right) \cdot O_{\mathbb{P}} \left( \rho^{-1} \left( \mh \right) \right) \\
        & \quad \cdot \left( O_{\mathbb{P}} \left( n^{-\frac{1}{2}}\rho^{\frac{1}{2}} \left( \mh \right) \right) \cdot O_{\mathbb{P}} \left( \rho ( \mh ) \right) + O_{\mathbb{P}} \left( n^{-\frac{1}{2}}\rho^{\frac{1}{2}} \left( \mh \right) \right) \cdot O_{\mathbb{P}} \left( \rho ( \mh ) \right)\right)\\
        & = O_{\mathbb{P}} \left( n^{-\frac{1}{2}}\rho^{-\frac{1}{2}} \left( \mh \right) \right).
    \end{split} \end{align}
    
    Regarding the second term in \eqref{T1hatbound}, by arguing as in the proof of \eqref{eq:3.41}, we can show that there exist constants $L_2 >0$, $N_2\in\mathbb{N}$ and $\delta_{\mathbb{M}}\in(0,1)$ such that
    \begin{align*}
        \int_{0}^{1/2} \sqrt{1+\log N ( \delta \epsilon, B_{\mathbb{M}}\left( \tilde{m}_{\mh, 1} (\mathbf{x}), \delta \right), d_{\mathbb{M}} )} \, \mathrm{d} \epsilon \leq L_2  \delta ^{\alpha_{\mathbb{M}}- 1}
    \end{align*}
    whenever $\delta \in (0, \delta_{\mathbb{M}}]$ and $n \geq N_2$, where $r_{\mathbb{M}}>0$ and $\alpha_{\mathbb{M}} \in (0, 1]$ are constants defined in the condition \ref{con:M3}. The Cauchy-Schwarz inequality implies that
    \begin{align} \begin{split} \label{Wtildebound_step1}
        & \left| \tilde{\sigma}_{\mh} (\mathbf{x}) \right|^2 \cdot \mathbb{E}\left[ \left| \tilde{W}_{\mathbf{x}, h} \left( \mathbf{X} \right) \right|^2 \right] \\
        & = \mathbb{E}\left[ \mathcal{L}_{\mx, \mh}^2 \left( \mathbf{X} \right) \left( 1 - \bm{\tilde{\mu}}_{\mh, 1}(\mathbf{x})\tran \bm{\tilde{\mu}}_{\mh, 2}(\mathbf{x})^{-1} \bm{\Phi}_{\mx}^{-1} (\mX) \right)^2 \right] \\
        & \leq 2 \mathbb{E}\left[ \mathcal{L}_{\mx, \mh}^2 \left( \mathbf{X} \right) \right] + 2 \mathbb{E}\left[ \mathcal{L}_{\mx, \mh}^2 \left( \mathbf{X} \right) \left( \bm{\tilde{\mu}}_{\mh, 1}(\mathbf{x})\tran \bm{\tilde{\mu}}_{\mh, 2}(\mathbf{x})^{-1} \bm{\Phi}_{\mx}^{-1} (\mX) \right)^2 \right] \\
        & = 2 \mathbb{E}\left[ \mathcal{L}_{\mx, \mh}^2 \left( \mathbf{X} \right) \right] + 2 \mathbb{E}\left[ \mathcal{L}_{\mx, \mh}^2 \left( \mathbf{X} \right) \left( \left( \mathbf{\Lambda}_{\mh}^{-1} \bm{\tilde{\mu}}_{\mh, 1}(\mathbf{x}) \right) \tran \left(\mathbf{\Lambda}_{\mh}^{-1}\bm{\tilde{\mu}}_{\mh, 2}(\mathbf{x}) \mathbf{\Lambda}_{\mh}^{-1}\right)^{-1} \left( \mathbf{\Lambda}_{\mh}^{-1} \bm{\Phi}_{\mx}^{-1} (\mX) \right) \right)^2 \right] \\
        & \leq 2 \mathbb{E}\left[ \mathcal{L}_{\mx, \mh}^2 \left( \mathbf{X} \right) \right] \\
        & \quad + 2\left\| \mathbf{\Lambda}_{\mh}^{-1} \bm{\tilde{\mu}}_{\mh, 1}(\mathbf{x}) \right\|_2^2 \cdot \left\| \left(\mathbf{\Lambda}_{\mh}^{-1}\bm{\tilde{\mu}}_{\mh, 2}(\mathbf{x}) \mathbf{\Lambda}_{\mh}^{-1}\right)^{-1} \right\|_2^2 \cdot \mathbb{E}\left[ \mathcal{L}_{\mx, \mh}^2 \left( \mathbf{X} \right) \left\| \mathbf{\Lambda}_{\mh}^{-1} \bm{\Phi}_{\mx}^{-1} (\mX) \right\|_2^2 \right].
    \end{split} \end{align}
    Combining \Cref{lemma:1.4}, \Cref{lemma:new2.2}, \Cref{lemma:new4.4}, \eqref{mu1bound_step3}, \eqref{eq:temp4.58}, \eqref{mu2tildeinverse}, \eqref{sigmatildeinverse} and \eqref{Wtildebound_step1}, we get
    \begin{align*}
        \mathbb{E}\left[ \left| \tilde{W}_{\mathbf{x}, h} \left( \mathbf{X} \right) \right|^2 \right] &= O \left( \rho^{-1} \left( \mh \right) \right)^2 \cdot \left( O \left( \rho ( \mh ) \right) + o \left( \rho ( \mh ) \right)^2 \cdot O \left( \rho^{-1} \left( \mh \right) \right)^2 \cdot O \left( \rho ( \mh ) \right) \right) \\
        & = O \left( \rho^{-1} \left( \mh \right) \right),
    \end{align*}
    which implies that there exist constants $L_3 > 0$ and $N_3 \in \mathbb{N}$ such that
    \begin{align*}
        \mathbb{E}\left[ \left| \tilde{W}_{\mathbf{x}, h} \left( \mathbf{X} \right) \right|^2 \right] \leq  L_3 \rho^{-1} \left( \mh \right)
    \end{align*}
    whenever $n \geq N_3$. By arguing as in the proof of \eqref{s0hatbound_step3}, we get
    \begin{align} \begin{split} \label{s1hatbound_step1}
        \mathbb{E} \left[ \sup_{y \in B_{\mathbb{M}}\left(\tilde{m}_{\mh, 1} (\mathbf{x}), \delta\right)} \left| \hat{S}_{\mh,1}(\mx,y) \right| \right] & \leq   4C L_2 L_3^{\frac{1}{2}} \mathrm{diam} (\mathbb{M}) \delta^{\alpha_{\mathbb{M}}} n^{-\frac{1}{2}}\rho^{-\frac{1}{2}} \left( \mh \right)
    \end{split} \end{align}
    whenever $\delta \in (0, \delta_{\mathbb{M}}]$ and $n \geq \max \{N_2 , N_3\}$.

    Let $\epsilon >0$ be any given constant. By \eqref{Whatbound_Step5}, there exist constants $L_{1,\epsilon} > 0$ and $N_{1,\epsilon} \in \mathbb{N}$, depending on $\epsilon$, such that
    \begin{align*}
        \mathbb{P} \left( n^{-1} \sum_{i=1}^{n} \left| \hat{W}_{\mathbf{x}, \mh} (\mathbf{X}^{(i)} ) - \tilde{W}_{\mathbf{x}, \mh} (\mathbf{X}^{(i)} )  \right| > L_{1,\epsilon} n^{-\frac{1}{2}} \rho^{-\frac{1}{2}} \left( \mh \right) \right) < \frac{\epsilon}{4}
    \end{align*}
    whenever $n \geq N_{1,\epsilon}$. We define an event $F_{n, \epsilon}$ as
    \begin{gather*}
        F_{n, \epsilon} :=  \left\{ n^{-1} \sum_{i=1}^{n} \left| \hat{W}_{\mathbf{x}, \mh} (\mathbf{X}^{(i)} ) - \tilde{W}_{\mathbf{x}, \mh} (\mathbf{X}^{(i)} )  \right| \leq L_{1,\epsilon} n^{-\frac{1}{2}} \rho^{-\frac{1}{2}} \left( \mh \right) \right\}.
    \end{gather*}
    Let $N_{4, \epsilon} := \max \{ N_{1,\epsilon}, N_2, N_3 \}$ and $L_{4, \epsilon} := 2 \mathrm{diam} (\mathbb{M}) ( L_{1,\epsilon} + 2C L_2 L_3^{1/2} )$. Combining \eqref{T1hatbound} and \eqref{s1hatbound_step1}, we get
    \begin{align*}
        & \mathbb{E} \left[ \sup_{y \in B_{\mathbb{M}}\left(\tilde{m}_{\mh, 1} (\mathbf{x}), \delta\right)} \bigg| \hat{T}_{\mh, 1}(\mathbf{x}, y) - \hat{T}_{\mh, 1}(\mathbf{x}, \tilde{m}_{\mh, 1}(\mathbf{x})) \bigg| \cdot \mathds{1}_{F_{n, \epsilon}} \right] \\
        & \leq \left( 2 \mathrm{diam}(\mathbb{M}) \right) \cdot \delta \cdot \mathbb{E} \left[ \left( n^{-1} \sum_{i=1}^{n} \left| \hat{W}_{\mathbf{x}, \mh} (\mathbf{X}^{(i)} ) - \tilde{W}_{\mathbf{x}, \mh} (\mathbf{X}^{(i)} ) \right| \right) \cdot \mathds{1}_{F_{n, \epsilon}} \right] \\
        & \quad +\mathbb{E} \left[ \sup_{y \in B_{\mathbb{M}}\left(\tilde{m}_{\mh, 1} (\mathbf{x}), \delta\right)} \left| \hat{S}_{\mh,1}(\mx,y) \right| \right] \\
        & \leq L_{4, \epsilon} \delta^{\alpha_{\mathbb{M}}} n^{-\frac{1}{2}}\rho^{-\frac{1}{2}} \left( \mh \right)
    \end{align*}
    whenever $\delta \in (0, \delta_{\mathbb{M}}]$ and $n \geq N_{4, \epsilon}$, where $\mathds{1}_{F_{n, \epsilon}}$ is the indicator function of $F_{n, \epsilon}$. The rest of the proof proceeds as in the proof of \Cref{thm:3.4}.
\end{proof}

\subsection{Proof of Theorem \ref{thm:main4.3} for $\hat{m}_{\mh, 1}(\mathbf{x})$}

The theorem follows from \Cref{thm:4.2} and \Cref{thm:4.4}.

\bigskip
\noindent
{\large\textbf{References for Supplementary Material}}

\bigskip

\noindent Garc\'{\i}a-Portugu\'{e}s, E., Crujeiras, R. M. and Gonz\'{a}lez-Manteiga, W. (2013). Kernel density estimation for directional-linear data. {\it Journal of Multivariate Analysis}, \textbf{121}, 152-175.

\bigskip

\noindent Im, C. J., Jeon, J. M. and Park, B. U. (2025). Local Fr\'{e}chet regression with spherical predictors. {\it Electronic Journal of Statistics}, \textbf{74}, 751-762.

\bigskip

\noindent van der Vaart, A. and Wellner, J. A. (1996). {\it Weak Convergence and Empirical Processes: With Applications to Statistics}. Springer.

\end{appendix}


\begin{thebibliography}{9}
\bibitem[Afsari(2011)]{Afsari (2011)} Afsari, B. (2011). Riemannian $L^p$ center of mass: Existence, uniqueness, and convexity. {\it Proceedings of the American Mathematical Society}, \textbf{139}, 655-673.
\bibitem[Bai et al.(1988)]{Bai et al. (1988)} Bai, Z. D., Radhakrishna Rao, C. and Zhao, L. C. (1988). Kernel estimators of density function of directional data. {\it Journal of Multivariate Analysis}, \textbf{27}, 24-39.
\bibitem[Bhattacharjee and M\"{u}ller(2023)]{Bhattacharjee and Muller (2023)} Bhattacharjee, S. and M\"{u}ller, H.-G. (2023). Single index Fr{\'e}chet regression. {\it Annals of Statistics}, \textbf{51}, 1770-1798.
\bibitem[Biswas and Banerjee(2025)]{Biswas and Banerjee (2025)} Biswas, S. and Banerjee, B. (2025). A semi-parametric torus-to-torus regression model with geometric loss: Application to cyclone data. {\it arXiv:2506.17014}.
\bibitem[Charlier(2013)]{Charlier (2013)} Charlier, B. (2013). Necessary and sufficient condition for the existence of a Fr\'echet mean on the circle. {\it ESAIM: Probability and Statistics}, \textbf{17}, 635-649.
\bibitem[Di Marzio et al.(2009)]{Di Marzio et al. (2009)}
Di Marzio, M., Panzera, A. and Taylor, C. C. (2009). Local polynomial regression for circular predictors. {\it Statistics and Probability Letters}, \textbf{79}, 2066-2075.
\bibitem[Di Marzio et al.(2011)]{Di Marzio et al. (2011)}
Di Marzio, M., Panzera, A. and Taylor, C. C. (2011). Kernel density estimation on the torus. {\it Journal of Statistical Planning and Inference}, \textbf{141}, 2156-2173.
\bibitem[Di Marzio et al.(2014)]{Di Marzio et al. (2014)} Di Marzio, M, Panzera, A. and Taylor, C. C. (2014). Nonparametric regression for spherical data. {\it Journal of the American Statistical Association}, \textbf{109}, 748-763.
\bibitem[Garc\'{\i}a-Portugu\'{e}s et al.(2013)]{Garcia-Portugues et al. (2013)} Garc\'{\i}a-Portugu\'{e}s, E., Crujeiras, R. M. and Gonz\'{a}lez-Manteiga, W. (2013). Kernel density estimation for directional-linear data. {\it Journal of Multivariate Analysis}, \textbf{121}, 152-175.
\bibitem[Eltzner et al.(2018)]{Eltzner et al. (2018)} Eltzner, B., Huckemann, S. and Mardia, K. V. (2018). Torus principal component analysis with applications to RNA structure. {\it Annals of Applied Statistics}, \textbf{12}, 1332-1359.
\bibitem[Hall et al.(1987)]{Hall et al. (1987)} Hall, P., Watson, G. S. and Cabrera, J. (1987). Kernel density estimation with spherical data. {\it Biometrika}, \textbf{74}, 751-762.
\bibitem[Im et al.(2025)]{Im et al. (2025)} Im, C. J., Jeon, J. M. and Park, B. U. (2025). Local Fr\'{e}chet regression with spherical predictors. {\it Electronic Journal of Statistics}, \textbf{74}, 751-762.
\bibitem[Jeon et al.(2021)]{Jeon et al. (2021)} Jeon, J. M., Park, B. U. and Van Keilegom, I. (2021). Additive regression for non-Euclidean responses and predictors. {\it Annals of Statistics}, \textbf{49}, 2611-2641.
\bibitem[Jeon et al.(2022)]{Jeon et al. (2022)} Jeon, J. M., Park, B. U. and Van Keilegom, I. (2022). Nonparametric regression on Lie groups with measurement errors. {\it Annals of Statistics}, \textbf{50}, 2973-3008.
\bibitem[Jung et al.(2021)]{Jung et al. (2021)} Jung, S., Park, K. and Kim, B. (2021). Clustering on the torus by conformal prediction. {\it Annals of Applied Statistics}, \textbf{15}, 1583-1603.
\bibitem[Le and Kendall(1993)]{Le and Kendall (1993)} Le, H. and Kendall, D. G. (1993). The Riemannian structure of Euclidean shape spaces: A novel environment for statistics. {\it Annals of Statistics}, \textbf{21}, 1225-1271.
\bibitem[Lin and M{\"u}ller(2021)]{Lin and Muller (2021)} Lin, Z. and M\"uller, H.-G. (2021). Total variation regularized Fr{\'e}chet regression for metric-space valued data. {\it Annals of Statistics}, \textbf{49}, 3510-3533.
\bibitem[Mardia and Jupp(2000)]{Mardia and Jupp (2000)} Mardia, K. V. and Jupp, P. E. (2000). {\it Directional Statistics}. John Wiley \& Sons.
\bibitem[Panaretos and Zemel(2020)]{Panaretos and Zemel (2020)} Panaretos, V. M. and Zemel, Y. (2020). {\it An invitation to statistics in Wasserstein space}. Springer.
\bibitem[Pelletier(2006)]{Pelletier (2006)} Pelletier, B. (2006). Non-parametric regression estimation on closed Riemannian manifolds. {\it Journal of Nonparametric Statistics}, \textbf{18}, 57-67.
\bibitem[Petersen and M\"{u}ller(2019)]{Petersen and Muller (2019)} Petersen, A. and M\"{u}ller, H.-G. (2019). Fr\'{e}chet regression for random objects with Euclidean predictors. {\it Annals of Statistics}, \textbf{47}, 691-719.
\bibitem[Ruppert and Wand(1994)]{Ruppert and Wand (1994)} Ruppert, D. and Wand, M. P. (1994). Multivariate Locally Weighted Least Squares Regression. {\it The Annals of Statistics}, \textbf{22}, 1346-1370.
\bibitem[Sanborn et al(2024)]{Sanborn et al. (2024)} Sanborn, S., Mathe, J., Papillon, M., Buracas, D., Lillemark, H. J., Shewmake, C., Bertics, A., Pennec, X. and Miolane, N. (2024). Beyond Euclid: An illustrated guide to modern machine learning with geometric, topological, and algebraic structures. {\it arXiv:2407.09468v1}.
\bibitem[Steyer et al(2025)]{Steyer et al. (2025)} Steyer, L., St{\"o}cker, A., Greven, S. and Alzheimer’s Disease Neuroimaging Initiative. (2025). Model-based Fr\'echet regression in (quotient) metric spaces with a focus on elastic curves. {\it Journal of Multivariate Analysis}, \textbf{211}, 105515.
\bibitem[Tucker and Wu(2025)]{Tucker and Wu (2025)} Tucker, D. C. and Wu, Y. (2025). Partially-global Fr\'echet regression. {\it Statistica Sinica}, \textbf{35}, 713-736.
\bibitem[Tucker et al.(2023)]{Tucker et al. (2023)} Tucker, D. C., Wu, Y. and M{\"u}ller, H.-G. (2023). Variable selection for global Fr{\'e}chet regression. {\it Journal of the American Statistical Association}, \textbf{118}, 1023-1037.
\bibitem[Xu and Wang(2023)]{Xu and Wang (2023)} Xu, D. and Wang, Y. (2023). Density estimation for toroidal data using semiparametric mixtures. {\it Statistics and Computing}, \textbf{33}, 140.
\bibitem[Zhou and M\"{u}ller(2022)]{Zhou and Muller (2022)} Zhou, Y. and M\"{u}ller, H.-G. (2022). Network regression with graph Laplacians. {\it Journal of Machine Learning Research}, \textbf{23}, 1-41.
\end{thebibliography}
\end{document}